\documentclass[pra,superscriptaddress,nopacs,twocolumn]{revtex4-2}

\usepackage{xcolor}
\definecolor{darkred}{rgb}{0.8,0.1,0.1}
\definecolor{darkblue2}{rgb}{0.1,0.2,0.8}
\usepackage{amsmath,mathtools}
\usepackage{enumerate}
\usepackage{graphicx,epsfig,amsmath,amssymb,verbatim,color}
\usepackage{dsfont}
\usepackage{float}
\usepackage{tikz}

\newcommand\update[1]{{\color{black} #1}}

\usepackage{hyperref}
\hypersetup{colorlinks=true,citecolor=blue,linkcolor=blue,filecolor=blue,urlcolor=blue,breaklinks=true}
\usepackage{cleveref}
\usepackage{graphicx}
\usepackage{xcolor}

\definecolor{darkblue}{RGB}{0,76,156}
\definecolor{darkkblue}{RGB}{0,0,153}
\definecolor{blue2}{RGB}{102,178,255}

\usepackage{url}
\newtheorem{definition}{Definition}
\newtheorem{proposition}{Proposition}

\newtheorem{theorem}[proposition]{Theorem}
\newtheorem{remark}{Remark}

\def\squareforqed{\hbox{\rlap{$\sqcap$}$\sqcup$}}
\def\qed{\ifmmode\squareforqed\else{\unskip\nobreak\hfil
\penalty50\hskip1em\null\nobreak\hfil\squareforqed
\parfillskip=0pt\finalhyphendemerits=0\endgraf}\fi}
\def\endenv{\ifmmode\;\else{\unskip\nobreak\hfil
\penalty50\hskip1em\null\nobreak\hfil\;
\parfillskip=0pt\finalhyphendemerits=0\endgraf}\fi}
\newenvironment{proof}{\noindent \textbf{{Proof.}~}}{\hfill $\blacksquare$}

\mathchardef\ordinarycolon\mathcode`\:
\mathcode`\:=\string"8000
\def\vcentcolon{\mathrel{\mathop\ordinarycolon}}
\begingroup \catcode`\:=\active
  \lowercase{\endgroup
  \let :\vcentcolon
  }

\makeatletter
\def\resetMathstrut@{%
    \setbox\z@\hbox{%
        \mathchardef\@tempa\mathcode`\[\relax
        \def\@tempb##1"##2##3{\the\textfont"##3\char"}%
        \expandafter\@tempb\meaning\@tempa \relax
    }%
    \ht\Mathstrutbox@\ht\z@ \dp\Mathstrutbox@\dp\z@}

\newcommand{\nc}{\newcommand}
\nc{\rnc}{\renewcommand}
\nc{\beg}{\begin{equation}}
\nc{\eeq}{{\end{equation}}}
\nc{\beqa}{\begin{eqnarray}}
\nc{\eeqa}{\end{eqnarray}}
\nc{\lbar}[1]{\overline{#1}}
\nc{\bra}[1]{\langle#1|}
\nc{\ket}[1]{|#1\rangle}
\nc{\ketbra}[2]{|#1\rangle\!\langle#2|}
\nc{\braket}[2]{\langle#1|#2\rangle}

\nc{\proj}[1]{| #1\rangle\!\langle #1 |}
\nc{\avg}[1]{\langle#1\rangle}
\nc{\Rank}{\operatorname{Rank}}
\nc{\smfrac}[2]{\mbox{$\frac{#1}{#2}$}}
\nc{\tr}{\operatorname{Tr}}
\nc{\ox}{\otimes}
\nc{\dg}{\dagger}
\nc{\dn}{\downarrow}
\nc{\cA}{{\cal A}}
\nc{\cB}{{\cal B}}
\nc{\cC}{{\cal C}}
\nc{\cD}{{\cal D}}
\nc{\cE}{{\cal E}}
\nc{\cF}{{\cal F}}
\nc{\cG}{{\cal G}}
\nc{\cH}{{\cal H}}
\nc{\cI}{{\cal I}}
\nc{\cJ}{{\cal J}}
\nc{\cK}{{\cal K}}
\nc{\cL}{{\cal L}}
\nc{\cM}{{\cal M}}
\nc{\cN}{{\cal N}}
\nc{\cO}{{\cal O}}
\nc{\cP}{{\cal P}}
\nc{\cQ}{{\cal Q}}
\nc{\cR}{{\cal R}}
\nc{\cS}{{\cal S}}
\nc{\cT}{{\cal T}}
\nc{\cV}{{\cal V}}
\nc{\cX}{{\cal X}}
\nc{\cY}{{\cal Y}}
\nc{\cZ}{{\cal Z}}
\nc{\cW}{{\cal W}}
\nc{\csupp}{{\operatorname{csupp}}}
\nc{\qsupp}{{\operatorname{qsupp}}}
\nc{\var}{{\operatorname{var}}}
\nc{\rar}{\rightarrow}
\nc{\lrar}{\longrightarrow}
\nc{\polylog}{{\operatorname{polylog}}}
\nc{\1}{{\mathds{1}}}
\nc{\wt}{{\operatorname{wt}}}
\nc{\av}[1]{{\left\langle {#1} \right\rangle}}
\nc{\supp}{{\operatorname{supp}}}

\nc{\dia}{{\diamondsuit }}

\def\a{\alpha}

\def\O{\Omega}

\nc{\RR}{{{\mathbb R}}}
\nc{\CC}{{{\mathbb C}}}
\nc{\FF}{{{\mathbb F}}}
\nc{\NN}{{{\mathbb N}}}
\nc{\ZZ}{{{\mathbb Z}}}
\nc{\PP}{{{\mathbb P}}}
\nc{\QQ}{{{\mathbb Q}}}
\nc{\UU}{{{\mathbb U}}}
\nc{\EE}{{{\mathbb E}}}
\nc{\id}{{\operatorname{id}}}

\nc{\CHSH}{{\operatorname{CHSH}}}

\nc{\be}{\begin{equation}}
\nc{\ee}{{\end{equation}}}
\nc{\bea}{\begin{eqnarray}}
\nc{\eea}{\end{eqnarray}}
\nc{\<}{\langle}
\rnc{\>}{\rangle}
\nc{\Hom}[2]{\mbox{Hom}(\CC^{#1},\CC^{#2})}
\nc{\rU}{\mbox{U}}

\nc{\ob}[1]{#1}

\nc{\SEP}{{\text{SEP}}}
\nc{\NS}{{\text{NS}}}
\nc{\LOCC}{{\operatorname{LOCC}}}
\nc{\PPT}{{\operatorname{PPT}}}
\nc{\EXT}{{\text{EXT}}}
\nc{\Sym}{{\operatorname{Sym}}}

\nc{\HH}{\mathbb{H}}

\nc{\ERLO}{{E_{\text{r,LO}}}}
\nc{\ERLOCC}{{E_{\text{r,LOCC}}}}
\nc{\ERPPT}{{E_{\text{r,PPT}}}}
\nc{\ERLOCCinfty}{{E^{\infty}_{\text{r,LOCC}}}}
\nc{\Aram}{{\operatorname{\sf A}}}

\rnc{\bar}{\;\rule{0pt}{9.5pt}\right|\;}

\nc{\lset}{\left\{\left.}
\nc{\rset}{\right\}}

\nc{\lsetr}{\left\{}
\nc{\rsetr}{\right.\right\}}
\nc{\barr}{\left|\rule{0pt}{9.5pt}\;}

\let\id\1

\nc{\norm}[2]{\left\lVert#1\right\rVert_{#2\!}}
\nc{\lnorm}[2]{\left\lVert#1\right\rVert_{\ell_{#2}}}

\nc{\EPPT}{{E_{\operatorname{PPT}}}}
\nc{\EPPTone}{{E_{\operatorname{PPT}}^{(1)}}}
\nc{\EK}{{E_{\kappa}}}
\usepackage{pdfpages}
\makeatletter
\AtBeginDocument{\let\LS@rot\@undefined}
\makeatother

\begin{document}
\title{Exact entanglement cost of quantum states and channels under positive-partial-transpose-preserving operations}
 \author{Xin Wang}
 \email{wangxin73@baidu.com}
 \affiliation{Institute for Quantum Computing, Baidu Research, Beijing 100093, China}
  \affiliation{Joint Center for Quantum Information and Computer Science, University of Maryland, College Park, Maryland 20742, USA}

 \author{Mark M. Wilde}
\email{wilde@cornell.edu}
\affiliation{Hearne Institute for Theoretical Physics, Department of Physics and Astronomy,
Center for Computation and Technology, Louisiana State University, Baton Rouge, Louisiana 70803, USA}
\affiliation{School of Electrical and Computer
  Engineering, Cornell University, Ithaca, New York 14850, USA}
  
\begin{abstract} 
This paper establishes single-letter formulas for the exact entanglement cost of simulating quantum channels under free quantum operations that completely preserve positivity of the partial transpose (PPT).
First, we introduce the $\kappa$-entanglement measure for point-to-point quantum channels, based on the idea of the $\kappa$-entanglement of bipartite states, and we establish several fundamental properties for it, including amortization collapse, monotonicity under PPT superchannels, additivity, normalization, faithfulness, and non-convexity.
Second, we introduce and solve the exact entanglement cost for simulating quantum channels in both the parallel and sequential settings, along with the assistance of free PPT-preserving operations. In particular, we establish that the entanglement cost in both cases is given by the same single-letter formula, the $\kappa$-entanglement measure of a quantum channel. We further show that this cost is equal to the largest $\kappa$-entanglement that can be shared or generated by the sender and receiver of the channel. This formula is calculable by a semidefinite program, thus allowing for an efficiently computable solution for general quantum channels.
Noting that the sequential regime is more powerful than the parallel regime, another notable implication of our result is that both regimes have the same power for exact quantum channel simulation, when PPT superchannels are free. For several basic Gaussian quantum channels, we show that the exact entanglement cost is given by the Holevo--Werner formula [Holevo and Werner, \textit{Phys.~Rev.~A} 63, 032312 (2001)], giving an operational meaning of the Holevo-Werner quantity for these channels.
\end{abstract}  

\date{\today}
\maketitle

%


\tableofcontents
\section{Introduction}

\subsection{Background}

Quantum entanglement, the most nonclassical manifestation of quantum mechanics,
has found use in a variety of physical tasks in quantum information processing, quantum cryptography, thermodynamics, and quantum computing \cite{Horodecki2009a}. 
A natural and fundamental problem is to develop a theoretical framework to quantify and describe it. In spite of remarkable recent progress in the resource theory of entanglement (for reviews see, e.g.,  \cite{Plenio2007, Horodecki2009a}),  many fundamental challenges have remained open.

One of the most important aspects of the resource theory of entanglement consists of the interconversions of states, with respect to a class of free operations. In particular, the problem of \emph{entanglement dilution} \cite{BBPS96} asks: 
given a target bipartite state $\rho_{AB}$ and a canonical unit of entanglement represented by the Bell state (or ebit) $\Phi_2\equiv |\Phi_2\rangle\!\langle \Phi_2|$, where $\ket{\Phi_2}=(\ket{00}+\ket{11})/\sqrt2$, what is the minimum rate at which we can produce copies of $\rho_{AB}$ from copies of $\Phi_2$ under a chosen set of free operations? 

The \emph{entanglement cost} \cite{Bennett1996c} was introduced to quantify the minimal rate $R$ of converting $\Phi^{\ox nR}_2$ to  $\rho_{AB}^{\ox n}$ with an arbitrarily high fidelity in the limit as $n$ becomes  large. When local operations and classical communication (LOCC) are allowed for free, the authors of \cite{Hayden2001} proved that the entanglement cost is equal to the regularized entanglement of formation~\cite{Bennett1996c}. When the free operations consist of  quantum operations that completely preserve positivity of the partial transpose (the PPT-preserving operations of \cite{R99,R01}), it is known that the entanglement cost is not equal to the regularized entanglement of formation \cite{Audenaert2003,F03,H06book}.

The \textit{exact entanglement cost} \cite{Audenaert2003} is an alternative and natural way to quantify the cost of entanglement dilution, being defined as the smallest asymptotic rate $R$ at which $\Phi_2^{\ox nR}$ is required in order to reproduce $\rho_{AB}^{\ox n}$ exactly. 
\update{The exact entanglement cost under PPT-preserving operations (PPT entanglement cost) was introduced and solved for a large class of quantum states in \cite{Audenaert2003}, but it has hitherto remained unknown for general quantum states until the recent solutions in \cite{WW18,WW20} (note that \cite{WW20} is a companion paper of the original announcement in \cite{WW18}).}

The above resource-theoretic problems can alternatively be phrased as simulation problems: How many copies of $\Phi_2$ are needed to simulate $n$ copies of a given bipartite state $\rho_{AB}$? As discussed above, the simulation can be either approximate, such that a verifier has little chance of distinguishing the simulation from the ideal case, while it can also be exact, such that a verifier has no chance at all for  distinguishing the simulation from the ideal case.

With this perspective, it is also natural to consider the simulation of a quantum channel, when allowing some set of operations for free and metering the entanglement cost of the simulation. The authors of \cite{BBCW13} defined the entanglement cost of a channel to be the smallest rate $R$ at which  $\Phi_2^{\ox nR}$ is needed, along with the free assistance of LOCC, in order to simulate the channel $\mathcal{N}^{\ox n}$, in such a way that a verifier would have little chance of distinguishing the simulation from the ideal case of $\mathcal{N}^{\ox n}$. In \cite{BBCW13}, it was shown that the regularized entanglement of formation of the channel is equal to its entanglement cost, thus extending the result of \cite{Hayden2001} in a natural way.

In a recent work \cite{Wilde2018}, it was observed that the channel simulation task defined in \cite{BBCW13} is actually a particular kind of simulation, called a parallel channel simulation. The paper \cite{Wilde2018} then defined an alternative notion of channel simulation, called sequential channel simulation, in which the goal is to simulate $n$ uses of the channel $\mathcal{N}$ in such a way that the most general verification strategy would have little chance of distinguishing the simulation from the ideal $n$ uses of the channel. Although a general formula for the entanglement cost in this scenario was not found, it was determined for several key channel models, including erasure, dephasing, three-dimensional Holevo--Werner, and single-mode pure-loss and pure-amplifier bosonic Gaussian channels.

\subsection{Summary of results}


\update{In this paper, we solve significant questions in the resource theory of entanglement, one of which has remained open since the inception of entanglement theory over two decades ago. Namely, we prove that the exact PPT-entanglement cost for quantum channels  has an efficiently computable, single-letter formula, reflecting the fundamental entanglement structure of bipartite quantum states and channels. Along with this claim, we prove that the exact parallel and sequential entanglement costs of quantum channels are given by the same efficiently computable, single-letter formula.}

We note here that all of our results apply to the resource theory of NPT (non-positive partial transpose) entanglement, introduced in \cite{R99,R01} and considered in \cite{Audenaert2003}, rather than to the more standard resource theory of entanglement, as introduced in \cite{Bennett1996c}. The key difference is that the free operations allowed here are completely PPT-preserving (C-PPT-P) operations, whereas the free operations allowed in the standard resource theory are LOCC. Since LOCC is contained in the set of C-PPT-P operations, the operational quantities considered here provide bounds on operational quantities in the standard resource theory.

\update{Our paper is structured as follows. We first introduce the $\kappa$-entanglement measure of a bipartite state and review its desirable properties \footnotemark[1]\footnotetext[1]{Note that $\kappa$-entanglement of quantum states was first introduced and proved to be equal to the exact entanglement cost in the original arXiv version of this paper in 2018 \cite{WW18}, and the related result was published in the companion paper~\cite{WW20}.}, including monotonicity under completely-PPT-preserving channels, additivity, normalization, faithfulness, non-convexity, and non-monogamy. For finite-dimensional states, it is also efficiently computable by means of a semi-definite program. In particular, the $\kappa$-entanglement is equal to the exact entanglement cost of a quantum state. We further evaluate the $\kappa$-entanglement (and the exact entanglement cost) for several bipartite states of interest (cf.~Section \ref{sec:examples-states}), including isotropic states, Werner states, maximally correlated states, some states supported on the $3\times 3$ antisymmetric subspace, and all bosonic Gaussian states.}


In Section~\ref{sec:kappa channel}, we extend the $\kappa$-entanglement measure from bipartite states to point-to-point quantum channels. We prove that it also satisfies several desirable properties, including non-increase under amortization, monotonicity under a class of PPT superchannels, additivity, normalization, faithfulness, and non-convexity. For finite-dimensional channels, it is also efficiently computable by means of a semi-definite program.

In Section~\ref{sec:cost channel}, we prove that the $\kappa$-entanglement of channels has a direct operational meaning as the entanglement cost of both parallel and sequential channel simulation. Thus, the theory of channel simulation significantly simplifies for the setting in which completely-PPT-preserving channels are allowed for free. In addition to all of the properties that it satisfies, this operational interpretation solidifies the $\kappa$-entanglement of a channel as a foundational measure of the entanglement of a quantum channel.

As a last contribution of this paper (cf.,~Sections~\ref{sec:examples of channels} and \ref{sec:gaussian-channels}), we evaluate the $\kappa$-entanglement (and exact entanglement cost) of several important channel models, including erasure, depolarizing, dephasing, and amplitude damping channels. We also leverage recent results in the literature \cite{LMGA17}, regarding the teleportation simulation of bosonic Gaussian channels, in order to evaluate the $\kappa$-entanglement and exact entanglement cost for several fundamental bosonic Gaussian channels. We remark that these latter results provide a direct operational interpretation of the Holevo--Werner quantity \cite{HW01} for these channels.

Finally, we conclude with a summary and some open questions.


\section{$\kappa$-entanglement measure and exact entanglement cost of quantum states}

\subsection{$\kappa$-entanglement measure and its operational meaning}

\update{We first recall an entanglement measure called the $\kappa$-entanglement measure for a bipartite state, which was introduced and analyzed in the original arXiv version of this paper in 2018 \cite{WW18} and published in the companion paper~\cite{WW20}. Here, we review the  important properties of this entanglement measure and its operational meaning as the exact entanglement cost.}

\begin{definition}[$\kappa$-entanglement measure \cite{WW20}]
\label{def:kappa-ent}
Let $\rho_{AB}$ be a bipartite state acting on a separable Hilbert space.  The $\kappa$-entanglement measure is defined as follows:
\begin{equation}
\label{eq:a prime}
E_\kappa(\rho_{AB})\coloneqq \inf_{S_{AB}\ge 0} \{ \log_2 \tr S_{AB} : -S_{AB}^{T_B}\le\rho_{AB}^{T_B}\le S_{AB}^{T_B} \}.
\end{equation}
\end{definition}

In the case that the state $\rho_{AB}$ acts on a finite-dimensional Hilbert space, then $E_\kappa(\rho_{AB})$ is calculable by a semi-definite program, and it is thus efficiently computable with respect to the dimension of the Hilbert space. Throughout this paper, we consider completely-PPT-preserving operations \cite{R99,R01}, defined as a bipartite operation $\mathcal{P}_{AB\to A'B'}$ (completely positive map) such that the map $T_{B'}\circ \mathcal{P}_{AB\to A'B'} \circ T_B$ is also completely positive, where $T_B$ and $T_{B'}$ denote the partial transpose map acting on the input system $B$ and the output system $B'$, respectively. If $\mathcal{P}_{AB\to A'B'}$ is also trace preserving, such that it is a quantum channel, and $T_{B'}\circ \mathcal{P}_{AB\to A'B'} \circ T_B$ is also completely positive, then we say that $\mathcal{P}_{AB\to A'B'}$ is a completely-PPT-preserving channel.

\textbf{Monotonicity under completely-PPT-preserving channels.}
The most important property of the $\kappa$-entanglement measure is that it does not increase under the action of a completely-PPT-preserving channel. Note that an LOCC channel \cite{Bennett1996c,CLMOW14}, as considered in entanglement theory, is a special kind of 
completely-PPT-preserving channel, as observed in \cite{R99,R01}. 
 \begin{theorem}
 [Monotonicity \cite{WW20}]
 \label{prop:ent-monotone}
Let $\rho_{AB}$ be a quantum state acting on a separable Hilbert
space, and let $\{\mathcal{P}_{AB\rightarrow A^{\prime}B^{\prime}}^{x}\}_{x}$
be a set of completely positive, trace non-increasing maps that are each completely
PPT-preserving, such that the sum map $\sum_{x}\mathcal{P}_{AB\rightarrow
A^{\prime}B^{\prime}}^{x}$ is quantum channel. Then the following entanglement
monotonicity inequality holds%
\begin{equation}
E_{\kappa}(\rho_{AB})\geq\sum_{x \, : \, p(x) > 0}p(x)E_{\kappa}\!\left(  \frac{\mathcal{P}%
_{AB\rightarrow A^{\prime}B^{\prime}}^{x}(\rho_{AB})}{p(x)}\right)
,\label{eq:mono-selective-E-kappa}%
\end{equation}
where $p(x)\coloneqq\operatorname{Tr}\mathcal{P}_{AB\rightarrow A^{\prime}B^{\prime}%
}^{x}(\rho_{AB})$. In particular, for a
completely-PPT-preserving quantum channel $\mathcal{P}_{AB\rightarrow
A^{\prime}B^{\prime}}$, the following inequality holds%
\begin{equation}
E_{\kappa}(\rho_{AB})\geq E_{\kappa}\!\left(  \mathcal{P}_{AB\rightarrow
A^{\prime}B^{\prime}}(\rho_{AB})\right)
.\label{eq:mono-non-selective-e-kappa}%
\end{equation}
\end{theorem}

\textbf{Dual representation and additivity.}
The optimization problem dual to $E_\kappa(\rho_{AB})$ in Definition~\ref{def:kappa-ent} is as follows:
\begin{multline}
\label{eq:a dual}
E^{\text{dual}}_\kappa(\rho_{AB}) \coloneqq
\sup_{V_{AB}^{T_B},\, W_{AB}^{T_B}\ge 0} \{ \log_2 \tr \rho_{AB}(V_{AB}-W_{AB}):\\
 V_{AB}+W_{AB}\le \1_{AB} \},
\end{multline}
which can be found by the Lagrange multiplier method (see, e.g., \cite[Section~1.2.2]{Watrous2011b}).
By weak duality \cite[Section~1.2.2]{Watrous2011b}, we have for every bipartite state $\rho_{AB}$ acting on a separable Hilbert space that
\begin{equation}
E^{\text{dual}}_\kappa(\rho_{AB})\leq E_\kappa(\rho_{AB}) .
\label{eq:weak-dual-kappa}
\end{equation}
For all finite-dimensional states $\rho_{AB}$, strong duality holds, so that
\begin{equation}
E_\kappa(\rho_{AB}) = E^{\text{dual}}_\kappa(\rho_{AB}).
\label{eq:strong-dual-e-kappa}
\end{equation}
This follows as a consequence of Slater's theorem. By employing the strong duality equality in \eqref{eq:strong-dual-e-kappa} for the finite-dimensional case, along with the approach from \cite{FAR11}, we conclude that the following equality holds for all bipartite states $\rho_{AB}$ acting on a separable Hilbert space:
\begin{equation}
E_\kappa(\rho_{AB}) = E^{\text{dual}}_\kappa(\rho_{AB}).
\label{eq:strong-dual-e-kappa-inf-dim}
\end{equation}
We provide an explicit proof of \eqref{eq:strong-dual-e-kappa-inf-dim} in Appendix~\ref{app:kappa-to-dual-infty}.
Both the primal and dual SDPs for $E_\kappa$ are important, as the combination of them allows for proving the following additivity of $E_\kappa$ with respect to tensor-product states.
\begin{proposition}[Additivity \cite{WW20}]\label{lemma:add}
For all bipartite states $\rho_{AB}$ and $\omega_{A'B'}$ acting on separable Hilbert spaces, the following additivity identity holds
	\begin{equation}
	E_\kappa(\rho_{AB}\ox\omega_{A'B'})
	=E_\kappa(\rho_{AB})+E_\kappa(\omega_{A'B'}).
	\label{eq:additivity-e-kappa-states}
	\end{equation}
\end{proposition}

\textbf{Relation to logarithmic negativity.}
There is an inequality  relating $E_\kappa$ to the logarithmic negativity \cite{Vidal2002,Plenio2005b}, defined as
 \begin{equation}
  E_{N}(\rho_{AB})\coloneqq\log_2\left \Vert \rho_{AB}^{T_{B}}\right\Vert_1.
  \label{eq:log-neg}
\end{equation}
Let $\rho_{AB}$ be a  bipartite state acting on a separable Hilbert space.
Then%
\begin{equation}
E_{\kappa}(\rho_{AB})\geq E_{N}(\rho_{AB}%
).\label{eq:kappa-greater-than-log-neg}%
\end{equation}
If $\rho_{AB}$ satisfies the binegativity condition
\begin{equation}
|\rho_{AB}^{T_B}|^{T_B} \geq 0,
\label{eq:bi-neg-cond}
\end{equation}
then
\begin{equation}
E_{\kappa}(\rho_{AB})= E_{N}(\rho_{AB}%
).\label{eq:kappa-equal-log-neg}%
\end{equation}

\textbf{Normalization.} $E_\kappa$ is normalized on maximally entangled states, and for finite-dimensional states, it achieves its largest value on maximally entangled states.
\begin{proposition}
[Normalization \cite{WW20}]\label{prop:normalize MES}
Let $\Phi_{AB}^{M}$ be a maximally entangled state of Schmidt
rank~$M$. Then%
\begin{equation}\label{eq:normalize MES}
E_{\kappa}(\Phi_{AB}^{M})=\log_2 M.
\end{equation}
Furthermore, for every bipartite state $\rho_{AB}$, the following bound holds
\begin{equation}\label{eq:dim bound of EK}
E_{\kappa}(\rho_{AB})\le \log_2 \min\{d_A,  d_B\},
\end{equation}
where $d_A$ and $d_B$ denote the dimensions of systems $A$ and $B$, respectively.
\end{proposition}

\textbf{Faithfulness.} $E_\kappa$ is faithful, in the sense that it is non-negative and equal to zero if and only if the state is a PPT state. To be specific, the following proposition holds.
\begin{proposition}
[Faithfulness \cite{WW20}]
\label{prop:faithfulness}
For a state $\rho_{AB}$ acting on a separable Hilbert space, we
have that $E_{\kappa}(\rho_{AB})\geq0$ and $E_{\kappa}(\rho_{AB})=0$ if and
only if $\rho_{AB}^{T_{B}}\geq 0$.
\end{proposition}




\textbf{No convexity.} The $\kappa$-entanglement measure is not generally convex. Due to \eqref{eq:kappa-equal-log-neg} and the fact that the binegativity condition in \eqref{eq:bi-neg-cond} holds for every two-qubit state \cite{Ishizaka2004a},  the non-convexity of $\EK$ boils down to  finding a two-qubit example for which the logarithmic negativity is not convex. In particular, let us choose the two-qubit states 
\begin{align}\label{state:no convex}
\rho_1=\Phi_2,\quad \rho_2=\frac{1}{2}(\proj{00}+\proj{11}), 
\end{align}
and their average $\rho=\frac{1}{2} (\rho_1+\rho_2)$. By direct calculation, we have 
\begin{align}
\EK(\rho) > \frac 1 2(\EK(\rho_1)+\EK(\rho_2)),
\end{align}
which implies that the $\kappa$-entanglement is not convex.


\textbf{No monogamy.} If an entanglement measure $E$ is monogamous \cite{CKW00,T04,KWin04}, then the following inequality should be satisfied for every tripartite state $\rho_{ABC}$:
\begin{align}\label{eq:monogamy}
E(\rho_{AB})+E(\rho_{AC})\le E(\rho_{A(BC)}),
\end{align}
where the entanglement in $E(\rho_{A(BC)})$ is understood to be with respect to the bipartite cut between systems $A$ and $BC$.
It is known that some entanglement measures satisfy the monogamy inequality above \cite{CKW00,KWin04}.
However, the $\kappa$-entanglement measure is not generally monogamous. 
 Consider a state ${\ket\psi}\!{\bra \psi}_{ABC}$ of three qubits, where ${\ket\psi}_{ABC}=\frac{1}{2}(\ket{000}_{ABC}+\ket{011}_{ABC}+\sqrt 2 \ket {110}_{ABC})$.
Due the fact that $\ket\psi_{ABC}$ can be written as
 \begin{equation}
 \ket\psi_{ABC} = [\ket 0_A \otimes \ket \Phi_{BC} + \ket 1_A \otimes \ket{10}_{BC}]/\sqrt{2},
 \end{equation}
 where $\ket{ \Phi}_{BC} = [\ket{00}_{BC} + \ket{11}_{BC} ] / \sqrt{2}$, this state is locally equivalent to $\ket{\Phi}_{AB}\otimes \ket 0_C$ with respect to the bipartite cut $A|BC$.
 One then finds that $\EK(\psi_{A(BC)})=\EK(\Phi_{AB})=E_N(\Phi_{AB})=1$. Furthermore, we have that
 	$\EK(\psi_{AB})=E_N(\psi_{AB}) =\log_2 \frac 3 2$, and
 	$\EK(\psi_{AC})=E_N(\psi_{AC})= \log_2 \frac 3 2$,
 which implies that
 \begin{align}
 \EK(\psi_{AB})+\EK(\psi_{AC}) >\EK(\psi_{A(BC)}).
 \end{align}

\textbf{$\kappa$-entanglement measure is equal to the exact PPT-entanglement cost.}
The $\kappa$-entanglement of a bipartite state is  equal to its exact entanglement cost, when completely-PPT-preserving channels are allowed for free.  Let $\Omega$ represent a set of free channels, which can be either \text{LOCC} or \text{PPT}. 
The one-shot exact entanglement cost of a  state $\rho_{AB}$, under the $\Omega$ channels, is defined as
\begin{align}
E^{(1)}_{\O}(\rho_{AB})= \inf_{\Lambda\in \Omega}\left\{\log_2 d:   \rho_{AB}=\Lambda_{\hat{A}\hat{B}\to AB} (\Phi^{d}_{\hat{A}\hat{B}})\right\},
\end{align}
where  $\Phi^{d}_{\hat{A}\hat{B}}=[1/d]\sum_{i,j=1}^{d}\ketbra{ii}{jj}_{\hat{A}\hat{B}}$ represents the standard maximally entangled state of Schmidt rank~$d$. The exact entanglement cost of a bipartite state $\rho_{AB}$, under the $\Omega$ channels, is defined as
\begin{align}
E_{\O}(\rho_{AB})= \limsup_{n \to \infty} \frac{1}{n}E^{(1)}_{\O}(\rho_{AB}^{\ox n}).
\end{align}
The exact entanglement cost under LOCC channels was previously considered in \cite{N99,TH00,H06book,YC18}, while the exact entanglement cost under PPT channels was considered in \cite{Audenaert2003,Matthews2008}.

In \cite{Audenaert2003}, the following bounds were given for $\EPPT$:
\begin{align}
E_N(\rho_{AB})\le	\EPPT(\rho_{AB})\le \log_2  Z(\rho_{AB}),
\label{eq:ape-bnds}
\end{align}
the lower bound being the logarithmic negativity recalled in \eqref{eq:log-neg},
and the upper bound defined as
\begin{equation}
Z(\rho_{AB})\coloneqq \tr |\rho_{AB}^{T_B}| +\dim(\rho_{AB})\max\{0,-\lambda_{\min}(|\rho_{AB}^{T_B}|^{T_B})\}.
\end{equation}
Due to the presence of the dimension factor $\dim(\rho_{AB})$, the upper bound in \eqref{eq:ape-bnds} clearly only applies in the case that $\rho_{AB}$ is finite-dimensional.

In what follows, we first recast
$E^{(1)}_{\PPT}(\rho_{AB})$ as an optimization problem, by building on previous developments in \cite{Audenaert2003,Matthews2008}. After that, we bound
$E^{(1)}_{\PPT}(\rho_{AB})$
in terms of $E_\kappa$, by observing that $E_\kappa$ is a relaxation of the optimization problem for $E^{(1)}_{\PPT}(\rho_{AB})$. We then finally prove that $E_{\PPT}(\rho_{AB})$ is equal to $E_\kappa$. 

\begin{theorem}[\cite{WW20}]
\label{thm:exact-cost-states}
Let $\rho_{AB}$ be a bipartite state acting on a separable Hilbert space. Then the one-shot exact PPT-entanglement cost $E_{\operatorname{PPT}}^{(1)}%
(\rho_{AB})$ is given by the following
optimization:
\begin{multline}
E_{\operatorname{PPT}}^{(1)}(\rho_{AB})=\inf\big\{  \log_{2}m: \\
-\left(
m-1\right)  G_{AB}^{T_{B}}\leq\rho_{AB}^{T_{B}}\leq\left(  m+1\right)
G_{AB}^{T_{B}},\\ G_{AB}\geq0,\ \operatorname{Tr}G_{AB}=1\big\}  .
\label{eq:op-quantity-PPT-cost}
\end{multline}
\end{theorem}

\begin{theorem}[Operational meaning \cite{WW20}]
\label{th:exact cost}
Let $\rho_{AB}$ be a  bipartite state acting on a separable Hilbert space. Then the exact PPT-entanglement cost of $\rho_{AB}$ is given by
	\begin{align}
	\label{eq:main result EPPT}		\EPPT(\rho_{AB})=E_\kappa(\rho_{AB}).
	\end{align}
\end{theorem}


Note that Theorem~\ref{th:exact cost} constitutes a significant development for entanglement theory, representing the first time that it has been shown that an entanglement measure is not only efficiently computable but also possesses a direct operational meaning. 
   In the work of \cite{BP08,BP10}, it was conjectured that the regularized relative entropy of entanglement is equal to the entanglement cost and distillable entanglement of a bipartite quantum state, with the set of free operations being asymptotically non-entangling maps. However, in spite of the fact that the work of  \cite{BP08,BP10} conjectured a direct operational meaning to the regularized relative entropy of entanglement, this entanglement measure arguably has limited application beyond being a formal expression, due to the fact that there is no known efficient procedure for computing it. See \cite{BBGLPRT22} for recent developments and discussions.

Furthermore, in prior work, most discussions about the structure and properties of entanglement are based on entanglement measures. However, none of these measures, with the exception of the regularized relative entropy of entanglement, possesses a direct operational meaning. Thus, the connection made by Theorem~\ref{th:exact cost} allows for the study of the structure of entanglement via an entanglement measure possessing a direct operational meaning. 
   Given that $E_\kappa = E_\PPT$ is neither convex nor monogamous, this raises questions of whether these properties should really be required or necessary for measures of entanglement, in contrast to the discussions put forward in \cite{T04,Horodecki2009a} based on intuition. Furthermore, 
 $\EK$ is additive (Proposition~\ref{lemma:add}), so that Theorem~\ref{th:exact cost} implies that $\EPPT$ is additive as well:
 \begin{equation}
 \EPPT(\rho_{AB} \otimes \omega_{A'B'}) = 
  \EPPT(\rho_{AB})
  +
   \EPPT(\omega_{A'B'}).
 \end{equation}
 Thus, $\EPPT$ is the only known example of an operational quantity in entanglement theory for which the optimal rate is additive as a function of general quantum states.

\subsection{Exact entanglement cost of particular bipartite states}

\label{sec:examples-states}

To have a better understanding of exact entanglement cost, we evaluate the exact entanglement cost for particular bipartite states of interest, including isotropic states \cite{Horodecki99}, Werner states \cite{W89}, maximally correlated states \cite{R99,R01}, some states supported on the $3\times 3$ antisymmetric subspace, and bosonic Gaussian states \cite{S17}. For isotropic and Werner states, the exact PPT-entanglement cost was already determined \cite{Audenaert2003,H06book}, and so we recall these developments here.

Let $A$ and $B$ be quantum systems, each of dimension~$d$. 
For $t\in [0,1]$ and $d\geq 2$, an isotropic state is defined as follows \cite{Horodecki99}:
\begin{equation}
\rho^{(t,d)}_{AB} \coloneqq t\Phi^d_{AB} + (1-t) \frac{\1_{AB} - \Phi^d_{AB}}{d^2-1}.
\end{equation}
An isotropic state is PPT if and only if $t \leq 1/d$.
It was shown in \cite[Exercise~8.73]{H06book} that $\rho^{(t,d)}_{AB}$ satisfies the binegativity condition: $|(\rho^{(t,d)}_{AB})^{T_B}|^{T_B} \geq 0$. By applying \eqref{eq:ape-bnds}, this implies that 
\begin{align}
\EPPT(\rho^{(t,d)}_{AB}) & = E_N(\rho^{(t,d)}_{AB}) \\
& = \begin{cases}
\log_2 dt & \text{ if } t> \frac{1}{d}    \\
0 & \text{ if } t \leq  \frac{1}{d},
\label{eq:e-ppt-isotropic}
\end{cases}
\end{align}
with the second equality shown in \cite{H01,H06book}.

Let $A$ and $B$ be quantum systems, each of dimension~$d$. A
Werner state is defined for $p\in\lbrack0,1]$ as \cite{W89}
\begin{equation}
W_{AB}^{(p,d)}\coloneqq\left(  1-p\right)  \frac{2}{d\left(  d+1\right)  }\Pi
_{AB}^{\mathcal{S}}+p\frac{2}{d\left(  d-1\right)  }\Pi_{AB}^{\mathcal{A}},
\label{eq:werner-param}
\end{equation}
where $\Pi_{AB}^{\mathcal{S}}\coloneqq\left(  \1_{AB}+ F_{AB}\right)  /2$ and
$\Pi_{AB}^{\mathcal{A}}\coloneqq\left(  \1_{AB}- F_{AB}\right)  /2$
are the
projections onto the symmetric and antisymmetric subspaces of $A$ and $B$, respectively, with $F_{AB}$ denoting the swap operator. A Werner state is PPT if and only if $p\leq 1/2$.
It was shown in \cite{Audenaert2003} that $W^{(p,d)}_{AB}$ satisfies the binegativity condition: $|(W^{(p,d)}_{AB})^{T_B}|^{T_B} \geq 0$. By applying \eqref{eq:ape-bnds}, this implies that~\cite{Audenaert2003}
\begin{align}
\EPPT(W^{(p,d)}_{AB}) & = E_N(W^{(p,d)}_{AB}) 
\\
& = \begin{cases}
\log_2 \left[\frac{2}{d}(2p-1)+1\right] & \text{ if } p> 1/2    \\
0 & \text{ if } p \leq 1/2,
\end{cases}
\end{align}
with the second equality shown in \cite{H01,H06book}.

A maximally correlated state is defined as \cite{R99,R01}
\begin{equation}
\rho^{\mathbf{c}}_{AB} \coloneqq \sum_{i,j=0}^{d-1} c_{ij}\ket{ii}\!\bra{jj},\end{equation}
with the complex coefficients $\mathbf{c} \coloneqq\{c_{ij}\}_{i,j}$ being chosen such that $\sum_{i,j=0}^{d-1}  c_{ij}\ketbra{i}{j}$ is a legitimate quantum state. 
Noting that $(\rho^{\mathbf{c}}_{AB})^{T_B} = \sum_{i,j=0}^{d-1} c_{ij}\ket{ij}\!\bra{ji}$, a direct calculation reveals that
\begin{align}
|(\rho^{\mathbf{c}}_{AB})^{T_B}|= \sum_{i,j=0}^{d-1} |c_{ij}|\ket{ij}\!\bra{ij}.
\end{align}
Considering that $|(\rho^{\mathbf{c}}_{AB})^{T_B}|^{T_B}=|(\rho^{\mathbf{c}}_{AB})^{T_B}|\ge 0$, we have that
\begin{equation}
\EPPT(\rho^{\mathbf{c}}_{AB})=E_N(\rho^{\mathbf{c}}_{AB})
= \log_2\!\left(\sum_{i,j}|c_{ij}|\right).
\label{eq:EPPT-max-corr-states}
\end{equation}

The maximally correlated state $\widehat \omega_\a$  was considered recently in \cite{YC18}:
\begin{align}
\widehat \omega^\a_{AB} & \coloneqq \a \Phi^2_{AB} + \frac{1-\a}{2} (\proj{00}_{AB}+\proj{11}_{AB})\\
& = \frac{\alpha}{2}\ket{00}\!\bra{11}_{AB}
+\frac{\alpha}{2}\ket{11}\!\bra{00}_{AB}
\notag \\
& \qquad + \frac{1}{2}\ket{00}\!\bra{00}_{AB}
+\frac{1}{2}\ket{11}\!\bra{11}_{AB},
\end{align}
where $\alpha \in[0,1]$.
The authors of \cite{YC18} showed that the exact entanglement cost under LOCC is bounded as 
\begin{align}
\left\lfloor\frac{1}{\log_2(\a+1)}\right\rfloor^{-1}\ge E_{\rm LOCC}(\widehat \omega^\a_{AB} ) \ge \log_2(\a+1),
\label{eq:bounds-LOCC-max-corr-ex}
\end{align}
for $0 < \a < \sqrt 2-1$.
However, under PPT-preserving operations, by \eqref{eq:EPPT-max-corr-states}, it holds that
\begin{align}
\EPPT(\widehat \omega^\a_{AB})=\log_2(\a+1).
\end{align}
for $\alpha\in[0,1]$. This demonstrates that the lower bound in \eqref{eq:bounds-LOCC-max-corr-ex} can be understood as arising from the fact that the inequality $E_{\rm LOCC} \geq \EPPT$ generally holds for an arbitrary bipartite state.

The next example indicates the irreversibility of exact PPT entanglement manipulation, and it also implies that $\EPPT$ is generally not equal to the logarithmic negativity $E_N$.
Consider the following rank-two state supported on the $3\times 3$ antisymmetric subspace \cite{Wang2016d}:
\begin{equation}\label{state:gap}
\rho_{v}=\frac{1}{2}(\proj{v_1}+\proj{v_2}),
\end{equation}
with 
$
\ket {v_1}=(\ket {01}-\ket{10})/{\sqrt 2}$ and $ \ket {v_2}=(\ket {02}-\ket{20})/{\sqrt 2}.
$  For the state $\rho_{v}$, 
it holds that
\begin{align}
R_{\max}(\rho_v) & =E_N(\rho_{v})
=\log_2\! \left(1+\frac{1}{\sqrt 2}\right) < \EPPT(\rho_v)=1 \notag \\
& < \log_2 Z(\rho)= 
\log_2 \!\left(1+\frac{13}{4\sqrt 2}\right),
\label{eq:strict-APE}
\end{align}
where $R_{\max}(\rho_v)$ denotes the max-Rains relative entropy \cite{WD16pra}.
The strict inequalities in \eqref{eq:strict-APE} also imply that both the lower and upper bounds from \eqref{eq:ape-bnds}, i.e., from \cite{Audenaert2003}, are generally not tight.

The last examples that we consider are bosonic Gaussian states \cite{S17}. As shown in \cite{Audenaert2003}, all bosonic Gaussian states $\rho^G_{AB}$ satisfy the binegativity condition
$|(\rho^G_{AB})^{T_B}|^{T_B}\geq 0$. Thus, as a consequence of Theorem~\ref{th:exact cost} and Eq.~\eqref{eq:kappa-equal-log-neg}, we conclude that
\begin{equation}
\EPPT(\rho^G_{AB})=E_N(\rho^G_{AB})
\label{eq:EPPT-Gaussian}
\end{equation}
for every bosonic Gaussian state $\rho^G_{AB}$. Note that an explicit expression for the logarithmic negativity of a bosonic Gaussian state is available in \cite[Eq.~(15)]{WEP03}. We stress again that it is not clear whether the equality in \eqref{eq:EPPT-Gaussian} follows from the upper bound in \eqref{eq:ape-bnds}, given that the dimension of a bosonic Gaussian state is generally equal to infinity.

%
%
%
%


\section{$\kappa$-entanglement measure for quantum channels}

\label{sec:kappa channel}

Quantum channels underlie the dynamics of quantum systems and they enable the manipulation of quantum states. In order to better effectively exploit quantum resources, it is important to understand the resource cost of quantum channels. In this section, we extend the $\kappa$-entanglement measure from bipartite states to point-to-point quantum channels. We establish several  properties of the $\kappa$-entanglement of quantum channels, including the fact that it does not increase under amortization, that it is monotone under the action of a class of PPT superchannels, that it is additive, normalized, faithful, and that it is generally not convex. The fact that it is monotone under the action of a class of PPT superchannels is a basic property that we would expect to hold for a good measure of the entanglement of a quantum channel.

In what follows, we consider a channel $\mathcal{N}_{A\to B}$ that takes density operators acting on a separable Hilbert space $\mathcal{H}_A$ to those acting on a separable Hilbert space $\mathcal{H}_B$.
We refer to such channels simply as quantum channels, regardless of whether $\mathcal{H}_A$ or $\mathcal{H}_B$ is finite-dimensional.
If the Hilbert spaces $\mathcal{H}_A$ and $\mathcal{H}_B$ are both finite-dimensional, then we specifically refer to $\mathcal{N}_{A\to B}$ as a finite-dimensional channel. 

We 
also make use of the Choi operator $J^{\mathcal{N}}_{RB}$ \cite{Holevo2011,H11} of the channel $\mathcal{N}_{A\to B}$, defined as
\begin{equation}
J^{\mathcal{N}}_{RB} \coloneqq \mathcal{N}_{A\to B}(\Gamma_{RA})\coloneqq \sum_{i,j} \ket{i}\!\bra{j}_R \otimes \mathcal{N}_{A\to B}(\ket{i}\!\bra{j}_A),
\end{equation}
where $R$ is isomorphic to the channel input $A$, we employ the shorthand $\Gamma_{RA} \equiv \ketbra{\Gamma}{\Gamma}_{RA}$, and $\ket{\Gamma}_{RA}$ denotes the unnormalized maximally entangled vector:
\begin{equation}
\ket{\Gamma}_{RA} \coloneqq \sum_i \ket{i}_R \otimes \ket{i}_A,
\end{equation}
where $\{ \ket{i}_R\}_i$ and 
$\{ \ket{i}_A\}_i$ are orthonormal bases for the Hilbert spaces $\mathcal{H}_R$ and $\mathcal{H}_A$.

\begin{definition}
[$\kappa$-entanglement of a channel]
Let $\mathcal{N}_{A\to B}$ be a quantum channel. Then the $\kappa$-entanglement of the channel $\mathcal{N}_{A\to B}$ is defined as
\begin{multline}
E_\kappa(\cN_{A\to B}) \coloneqq \inf_{Q_{AB}\ge 0} \{ \log_2 \left\Vert \tr_{B}[Q_{AB}] \right \Vert_{\infty}   :\\
-Q_{AB}^{T_B}\le (J_{AB}^{\cN})^{T_B} \le Q_{AB}^{T_B} \}.
\label{eq:kappa-channel-primal}
\end{multline}
\end{definition}

\begin{proposition}
\label{eq:state-opt-for-E_kappa-ch}
Let $\mathcal{N}_{A\rightarrow B}$ be a quantum channel.
Then%
\begin{equation}
E_{\kappa}(\mathcal{N}_{A\rightarrow B})=\sup_{\rho_{RA}}E_{\kappa
}(\mathcal{N}_{A\rightarrow B}(\rho_{RA}%
)),\label{eq:channel-quantity-to-E-kappa}%
\end{equation}
where the supremum is with respect to all states $\rho_{RA}$ with system $R$ arbitrary.
\end{proposition}

\begin{proof}
Due to Proposition~\ref{prop:ent-monotone}, i.e., the fact that $E_{\kappa}$ for states is monotone non-increasing with
respect to completely-PPT-preserving channels (with one such channel being a
local partial trace), it follows from purification, the Schmidt decomposition,
and this local data processing, that it suffices to optimize with respect to
pure states $\rho_{RA}$ with system $R$ isomorphic to system $A$. Thus, we
conclude that%
\begin{equation}
\sup_{\rho_{RA}}E_{\kappa}(\mathcal{N}_{A\rightarrow B}(\rho_{RA}))=\sup
_{\phi_{RA}}E_{\kappa}(\mathcal{N}_{A\rightarrow B}(\phi_{RA})),
\end{equation}
where $\phi_{RA}$ is pure and $R\simeq A$.

By definition, and using the fact that every pure state $\phi_{RA}$ of the form mentioned above can be
represented as $X_{R}\Gamma_{RA}X_{R}^{\dag}$ with $\left\Vert X_{R}%
\right\Vert _{2}=1$, we have that%
\begin{multline}
\sup_{\phi_{RA}}E_{\kappa}(\mathcal{N}_{A\rightarrow B}(\phi_{RA}))
\\=\log_2
\sup_{X_{R}:\left\Vert X_{R}\right\Vert _{2}=1,\left\vert X_{R}\right\vert
>0}\inf_{S_{RB}\geq 0}\{\operatorname{Tr}S_{RB}:\\
-S_{RB}^{T_{B}}\leq X_{R}%
[J_{RB}^{\mathcal{N}}]^{T_{B}}X_{R}^{\dag}\leq S_{RB}^{T_{B}}\},
\end{multline}
where the equality follows because the set of operators $X_{R}$ satisfying
$\left\Vert X_{R}\right\Vert _{2}=1$ and $\left\vert X_{R}\right\vert >0$ is
dense in the set of all operators satisfying $\left\Vert X_{R}\right\Vert
_{2}=1$. Now defining $Q_{RB}$ in terms of $S_{RB}=X_{R}Q_{RB}X_{R}^{\dag}$,
and using the facts that
\begin{multline}
-S_{RB}^{T_{B}}   \leq X_{R}[J_{RB}^{\mathcal{N}}]^{T_{B}}X_{R}^{\dag}\leq
S_{RB}^{T_{B}}\quad  \Leftrightarrow\quad
\\ -Q_{RB}^{T_{B}}\leq\lbrack
J_{RB}^{\mathcal{N}}]^{T_{B}}\leq Q_{RB}^{T_{B}},
\end{multline}
\begin{align}
S_{RB}   \geq0\quad & \Leftrightarrow\quad Q_{RB}\geq0,
\end{align}
for operators $X_{R}$ satisfying $\left\vert X_{R}\right\vert >0$, we find
that
\begin{widetext}
\begin{align}
& \sup_{X_{R}:\left\Vert X_{R}\right\Vert _{2}=1,\left\vert X_{R}\right\vert
>0}\inf_{S_{RB}\geq 0}\{\operatorname{Tr}S_{RB}:-S_{RB}^{T_{B}}\leq X_{R}%
[J_{RB}^{\mathcal{N}}]^{T_{B}}X_{R}^{\dag}\leq S_{RB}^{T_{B}}\}\notag \\
& =\sup_{X_{R}:\left\Vert X_{R}\right\Vert _{2}=1,\left\vert X_{R}\right\vert
>0}\inf_{Q_{RB}\geq 0}\{\operatorname{Tr}X_{R}Q_{RB}X_{R}^{\dag}:-Q_{RB}^{T_{B}}%
\leq\lbrack J_{RB}^{\mathcal{N}}]^{T_{B}}\leq Q_{RB}^{T_{B}}\}\notag \\
& =\sup_{\rho_{R}:\tr \rho_{R}=1,\rho_{R}
>0}\inf_{Q_{RB}\geq 0}\{\operatorname{Tr}[\rho_{R}\operatorname{Tr}%
_{B}[Q_{RB}]]:-Q_{RB}^{T_{B}}\leq\lbrack J_{RB}^{\mathcal{N}}]^{T_{B}}\leq
Q_{RB}^{T_{B}}\}\notag \\
& =\sup_{\rho_{R}:\tr \rho_{R}=1,\rho_{R}
\geq 0}\inf_{Q_{RB}\geq 0}\{\operatorname{Tr}[\rho_{R}\operatorname{Tr}%
_{B}[Q_{RB}]]:-Q_{RB}^{T_{B}}\leq\lbrack J_{RB}^{\mathcal{N}}]^{T_{B}}\leq
Q_{RB}^{T_{B}}\}\notag \\
& =\inf_{Q_{RB}\geq 0}\left[  \sup_{\rho_{R}:\tr \rho_{R}
=1,\rho_{R} >0}\{\operatorname{Tr}[\rho_{R}\operatorname{Tr}_{B}[Q_{RB}]]:-Q_{RB}^{T_{B}}\leq\lbrack J_{RB}%
^{\mathcal{N}}]^{T_{B}}\leq Q_{RB}^{T_{B}}\}\right]  \notag \\
& =\inf_{Q_{RB}\geq 0}\{\left\Vert \operatorname{Tr}_{B}[Q_{RB}]\right\Vert _{\infty
}:-Q_{RB}^{T_{B}}\leq\lbrack J_{RB}^{\mathcal{N}}]^{T_{B}}\leq Q_{RB}^{T_{B}%
}\}.
\end{align}
\end{widetext}
The fourth equality follows from an application of the Sion minimax theorem \cite{sion58},
given that the set of operators satisfying $\tr \rho_{R}=1$ and $\rho_{R} \geq 0$ is compact and both sets over which we are optimizing are convex. Putting everything
together, we conclude \eqref{eq:channel-quantity-to-E-kappa}.
\end{proof}

\subsection{Amortization collapse and monotonicity under a class of PPT superchannels}

In this subsection, we prove that the $\kappa$-entanglement of a quantum channel does not increase under amortization, which is a property that holds for the squashed entanglement of a channel \cite{TGW14,TGW14B},
a channel's max-relative entropy of entanglement \cite{Christandl2017}, and 
the max-Rains information of a channel \cite{BW17}. We additionally prove that this property implies that the $\kappa$-entanglement of a quantum channel does not increase under the action of a class of PPT superchannels.
A PPT superchannel $\Theta^{\operatorname{PPT}}$ is a physical transformation of a quantum channel. The class of PPT superchannels that we consider realizes the following transformation of a channel
$\mathcal{M}_{\hat{A}\rightarrow\hat{B}}$
to a channel $\mathcal{N}_{A\rightarrow B}$
in terms of 
completely-PPT-preserving channels $\mathcal{P}_{A\rightarrow\hat{A}A_{M}B_{M}%
}^{\text{pre}}$ and $\mathcal{P}_{A_{M}\hat{B}B_{M}}^{\text{post}}$:
\begin{multline}
\mathcal{N}_{A\rightarrow B}=\Theta^{\operatorname{PPT}}(\mathcal{M}_{\hat
{A}\rightarrow\hat{B}})\coloneqq\\
\mathcal{P}_{A_{M}\hat{B}B_{M}}^{\text{post}}%
\circ\mathcal{M}_{\hat{A}\rightarrow\hat{B}}\circ\mathcal{P}_{A\rightarrow
\hat{A}A_{M}B_{M}}^{\text{pre}}.
\label{eq:superchannel-action}
\end{multline}
We also state that the same property holds for the max-Rains information of a quantum channel, due to the main result of \cite{BW17}, while a channel's squashed entanglement and max-relative entropy of entanglement do not increase under the action of an LOCC superchannel.

We begin our development with the following amortization inequality:
\begin{proposition}
[Amortization inequality]
\label{prop:amort-ineq-E-kappa}Let $\rho_{A^{\prime
}AB^{\prime}}$ be a quantum state
acting on a separable Hilbert space
and let $\mathcal{N}_{A\rightarrow B}$ be a 
quantum channel. Then the following amortization inequality holds%
\begin{equation}
E_{\kappa}(\mathcal{N}_{A\rightarrow B}(\rho_{A^{\prime}AB^{\prime}%
}))-E_{\kappa}(\rho_{A^{\prime}AB^{\prime}})\leq E_{\kappa}(\mathcal{N}%
_{A\rightarrow B}).
\end{equation}

\end{proposition}

\begin{proof}
A proof for this inequality follows similarly to the proof of \cite[Proposition~1]{BW17}. We first rewrite the desired inequality as%
\begin{equation}
E_{\kappa}(\mathcal{N}_{A\rightarrow B}(\rho_{A^{\prime}AB^{\prime}}))\leq
E_{\kappa}(\mathcal{N}_{A\rightarrow B})+E_{\kappa}(\rho_{A^{\prime}%
AB^{\prime}}),
\end{equation}
and then once again as%
\begin{equation}
2^{E_{\kappa}(\mathcal{N}_{A\rightarrow B}(\rho_{A^{\prime}AB^{\prime}}))}%
\leq2^{E_{\kappa}(\mathcal{N}_{A\rightarrow B})}\cdot2^{E_{\kappa}%
(\rho_{A^{\prime}AB^{\prime}})}.\label{eq:amort-rewrite-1}%
\end{equation}
Consider that
\begin{widetext}
\begin{align}
2^{E_{\kappa}(\rho_{A^{\prime}AB^{\prime}})} &  =\inf\left\{
\operatorname{Tr}S_{A^{\prime}AB^{\prime}}:-S_{A^{\prime}AB^{\prime}%
}^{T_{B^{\prime}}}\leq\rho_{A^{\prime}AB^{\prime}}^{T_{B^{\prime}}}\leq
S_{A^{\prime}AB^{\prime}}^{T_{B^{\prime}}},\ S_{A^{\prime}AB^{\prime}}%
\geq0\right\}  ,\\
2^{E_{\kappa}(\mathcal{N}_{A\rightarrow B})} &  =\inf\left\{  \left\Vert
\operatorname{Tr}_{B}Q_{RB}\right\Vert _{\infty}:-Q_{RB}^{T_{B}}\leq\left[
J_{RB}^{\mathcal{N}}\right]  ^{T_{B}}\leq Q_{RB}^{T_{B}},\ Q_{RB}%
\geq0\right\}  .
\end{align}
\end{widetext}
Let $S_{A^{\prime}AB^{\prime}}$ be an arbitrary operator satisfying%
\begin{equation}
-S_{A^{\prime}AB^{\prime}}^{T_{B^{\prime}}}\leq\rho_{A^{\prime}AB^{\prime}%
}\leq S_{A^{\prime}AB^{\prime}}^{T_{B^{\prime}}},\ S_{A^{\prime}AB^{\prime}%
}\geq0,\label{eq:constr-for-rho-E-kap}%
\end{equation}
and let $Q_{RB}$ be an arbitrary operator satisfying%
\begin{equation}
-Q_{RB}^{T_{B}}\leq J_{RB}^{\mathcal{N}}\leq Q_{RB}^{T_{B}},\ Q_{RB}%
\geq0.\label{eq:constr-for-ch-E-kap}%
\end{equation}
Then let%
\begin{equation}
F_{A^{\prime}BB^{\prime}}=\langle\Gamma|_{RA}(S_{A^{\prime}AB^{\prime}}\otimes
Q_{RB})|\Gamma\rangle_{RA},
\end{equation}
where $|\Gamma\rangle_{RA}$ denotes the unnormalized maximally entangled
vector. It follows that $F_{A^{\prime}BB^{\prime}}\geq0$ because
$S_{A^{\prime}AB^{\prime}}\geq0$ and $Q_{RB}\geq0$. Furthermore, we have from
\eqref{eq:constr-for-rho-E-kap} and \eqref{eq:constr-for-ch-E-kap} that%
\begin{align}
F_{A^{\prime}BB^{\prime}}^{T_{BB^{\prime}}} &  =\left[  \langle\Gamma
|_{RA}(S_{A^{\prime}AB^{\prime}}\otimes Q_{RB})|\Gamma\rangle_{RA}\right]
^{T_{BB^{\prime}}}\\
&  =\langle\Gamma|_{RA}(S_{A^{\prime}AB^{\prime}}^{T_{B^{\prime}}}\otimes
Q_{RB}^{T_{B}})|\Gamma\rangle_{RA}\\
&  \geq\langle\Gamma|_{RA}(\rho_{A^{\prime}AB^{\prime}}^{T_{B^{\prime}}%
}\otimes\left[  J_{RB}^{\mathcal{N}}\right]  ^{T_{B}})|\Gamma\rangle_{RA}\\
&  =\left[  \langle\Gamma|_{RA}(\rho_{A^{\prime}AB^{\prime}}\otimes
J_{RB}^{\mathcal{N}})|\Gamma\rangle_{RA}\right]  ^{T_{BB^{\prime}}}\\
&  =\left[  \mathcal{N}_{A\rightarrow B}(\rho_{A^{\prime}AB^{\prime}})\right]
^{T_{BB^{\prime}}}.
\end{align}
Similarly, we have that%
\begin{equation}
-F_{A^{\prime}BB^{\prime}}^{T_{BB^{\prime}}}\leq\left[  \mathcal{N}%
_{A\rightarrow B}(\rho_{A^{\prime}AB^{\prime}})\right]  ^{T_{BB^{\prime}}},
\end{equation}
by using $-S_{A^{\prime}AB^{\prime}}^{T_{B^{\prime}}}\leq\rho_{A^{\prime
}AB^{\prime}}^{T_{B^{\prime}}}$ and $-Q_{RB}^{T_{B}}\leq\left[  J_{RB}%
^{\mathcal{N}}\right]  ^{T_{B}}$. Thus, $F_{A^{\prime}BB^{\prime}}$ is
feasible for $2^{E_{\kappa}(\mathcal{N}_{A\rightarrow B}(\rho_{A^{\prime
}AB^{\prime}}))}$.

Finally, consider that%
\begin{align}
& 2^{E_{\kappa}(\mathcal{N}_{A\rightarrow B}(\rho_{A^{\prime}AB^{\prime}}))} \notag \\
&
\leq\operatorname{Tr}F_{A^{\prime}BB^{\prime}}\\
&  =\operatorname{Tr}\langle\Gamma|_{RA}(S_{A^{\prime}AB^{\prime}}\otimes
Q_{RB})|\Gamma\rangle_{RA}\\
&  =\operatorname{Tr}S_{A^{\prime}AB^{\prime}}Q_{AB}^{T_{A}}\\
&  =\operatorname{Tr}\left[  S_{A^{\prime}AB^{\prime}}\operatorname{Tr}%
_{B}Q_{AB}^{T_{A}}\right]  \\
&  \leq\operatorname{Tr}S_{A^{\prime}AB^{\prime}}\left\Vert \operatorname{Tr}%
_{B}Q_{AB}^{T_{A}}\right\Vert _{\infty}\\
&  =\operatorname{Tr}S_{A^{\prime}AB^{\prime}}\left\Vert \operatorname{Tr}%
_{B}Q_{AB}\right\Vert _{\infty}.
\end{align}
The inequality above follows from H\"older's inequality. The last equality
follows because the spectrum of an operator remains invariant under the action
of a transpose. Since the inequality above holds for all $S_{A^{\prime
}AB^{\prime}}$ and $Q_{RB}$ satisfying \eqref{eq:constr-for-rho-E-kap} and
\eqref{eq:constr-for-ch-E-kap}, respectively, we conclude the inequality in~\eqref{eq:amort-rewrite-1}.
\end{proof}

\begin{definition}[Amortized $\kappa$-ent.~of a channel]
Following \cite{KW17a}, we define the amortized $\kappa$-entanglement of a quantum channel $\mathcal{N}_{A \to B}$ as 
\begin{equation}
E_\kappa^{\mathcal{A}}(\mathcal{N}_{A \to B}) \coloneqq \sup_{\rho_{A'AB'}} [
E_\kappa(\mathcal{N}_{A \to B}(\rho_{A'AB'})) - E_\kappa(\rho_{A'AB'})
].
\end{equation}
where the supremum is with respect to every state $\rho_{A'AB'}$, with the $A'$ and $B'$ systems arbitrary.
\end{definition}

In spite of the possibility that amortization might increase $E_\kappa$, a consequence of Proposition~\ref{prop:amort-ineq-E-kappa} is that in fact it does not:
\begin{proposition}
\label{prop:amort-does-not-increase-E-kappa}
Let $\mathcal{N}_{A \to B}$ be a quantum channel. Then the $\kappa$-entanglement of a channel does not increase under amortization:
\begin{equation}
E_\kappa^{\mathcal{A}}(\mathcal{N}_{A \to B}) = E_\kappa(\mathcal{N}_{A \to B}).
\end{equation}
\end{proposition}

\begin{proof}
The inequality 
$E_\kappa^{\mathcal{A}}(\mathcal{N}_{A \to B}) \geq E_\kappa(\mathcal{N}_{A \to B})$ follows from Proposition~\ref{eq:state-opt-for-E_kappa-ch}, by identifying $A'$ with $R$, setting $B'$ to be a trivial system, and noting that 
$E_\kappa(\rho_{A'AB'})$ vanishes for this choice. The opposite inequality is a direct consequence of Proposition~\ref{prop:amort-ineq-E-kappa}.
\end{proof}

\begin{theorem}
[Monotonicity]
Let $\mathcal{M}_{\hat{A}\rightarrow\hat{B}
}$ be a quantum channel and  $\Theta^{\operatorname{PPT}}$ a completely-PPT-preserving superchannel of the form in \eqref{eq:superchannel-action}.
The channel measure $E_{\kappa}$ is monotone
under the action of the superchannel $\Theta^{\operatorname{PPT}}$, in the sense that%
\begin{equation}
E_{\kappa}(\mathcal{M}_{\hat{A}\rightarrow\hat{B}}) 
\geq 
E_{\kappa}(\Theta^{\operatorname{PPT}}(\mathcal{M}_{\hat{A}\rightarrow\hat{B}%
})).
\end{equation}

\end{theorem}

\begin{proof}
The proof is similar to that of  \cite[Proposition~6]{BHKW18}. Let
$\rho_{A^{\prime}AB^{\prime}}$ be an arbitrary input state. Then we have that
\begin{align}
& E_{\kappa}(\mathcal{N}_{A\rightarrow B}(\rho_{A^{\prime}AB^{\prime}%
}))-E_{\kappa}(\rho_{A^{\prime}AB^{\prime}}) \notag \\
& =E_{\kappa}((\mathcal{P}_{A_{M}\hat{B}B_{M}}^{\text{post}}\circ
\mathcal{M}_{\hat{A}\rightarrow\hat{B}}\circ\mathcal{P}_{A\rightarrow\hat
{A}A_{M}B_{M}}^{\text{pre}})(\rho_{A^{\prime}AB^{\prime}}))\notag \\
& \qquad -E_{\kappa}%
(\rho_{A^{\prime}AB^{\prime}})\\
& \leq E_{\kappa}((\mathcal{P}_{A_{M}\hat{B}B_{M}}^{\text{post}}%
\circ\mathcal{M}_{\hat{A}\rightarrow\hat{B}}\circ\mathcal{P}_{A\rightarrow
\hat{A}A_{M}B_{M}}^{\text{pre}})(\rho_{A^{\prime}AB^{\prime}}))\notag \\
& \qquad -E_{\kappa
}(\mathcal{P}_{A\rightarrow\hat{A}A_{M}B_{M}}^{\text{pre}}(\rho_{A^{\prime
}AB^{\prime}}))\\
& \leq E_{\kappa}((\mathcal{M}_{\hat{A}\rightarrow\hat{B}}\circ\mathcal{P}%
_{A\rightarrow\hat{A}A_{M}B_{M}}^{\text{pre}})(\rho_{A^{\prime}AB^{\prime}%
}))\notag \\
& \qquad -E_{\kappa}(\mathcal{P}_{A\rightarrow\hat{A}A_{M}B_{M}}^{\text{pre}}%
(\rho_{A^{\prime}AB^{\prime}}))\\
& \leq E_{\kappa}^{\mathcal{A}}(\mathcal{M}_{\hat{A}\rightarrow\hat{B}})\\
& =E_{\kappa}(\mathcal{M}_{\hat{A}\rightarrow\hat{B}}).
\end{align}
The first inequality follows because $E_{\kappa}(\mathcal{P}_{A\rightarrow
\hat{A}A_{M}B_{M}}^{\text{pre}}(\rho_{A^{\prime}AB^{\prime}}))\leq E_{\kappa
}(\rho_{A^{\prime}AB^{\prime}})$, given that $E_{\kappa}$ does not increase
under the action of the completely PPT-preserving channel $\mathcal{P}%
_{A\rightarrow\hat{A}A_{M}B_{M}}^{\text{pre}}$ (Proposition~\ref{prop:ent-monotone}). The second inequality follows
from a similar reasoning, but with respect to the completely-PPT-preserving channel $\mathcal{P}%
_{A_{M}\hat{B}B_{M}}^{\text{post}}$. The last inequality follows because
$\mathcal{P}_{A\rightarrow\hat{A}A_{M}B_{M}}^{\text{pre}}(\rho_{A^{\prime
}AB^{\prime}})$ is a particular bipartite state to consider at the input of
the channel $\mathcal{M}_{\hat{A}\rightarrow\hat{B}}$, but the quantity
$E_{\kappa}^{\mathcal{A}}$ involves an optimization over all such states. The
final equality is a consequence of
Proposition~\ref{prop:amort-does-not-increase-E-kappa}.
\end{proof}

\begin{remark} We remark here that the same inequality holds for the max-Rains information of a channel $R_{\max}(\cN)$, defined in \cite{WD16,WFD17} and considered further in \cite{BW17} (see also \cite{TWW14}). That is, for $\mathcal{M}_{\hat{A}\rightarrow\hat{B}
}$ a quantum channel and  $\Theta^{\operatorname{PPT}}$ a completely-PPT-preserving superchannel of the form in \eqref{eq:superchannel-action},
the following inequality holds
\begin{equation}
R_{\max}(\mathcal{M}_{\hat{A}\rightarrow\hat{B}}) 
\geq 
R_{\max}(\Theta^{\operatorname{PPT}}(\mathcal{M}_{\hat{A}\rightarrow\hat{B}%
})).
\end{equation}
This follows because $R_{\max}$ does not increase under amortization, as shown in \cite{BW17}, and because the max-Rains relative entropy does not increase under the action of a completely-PPT-preserving channel \cite{WD16pra}.

Furthermore, a similar inequality holds for the squashed entanglement  $E_{\operatorname{sq}}$ of a channel and for a channel's max-relative entropy of entanglement $E_{\max}$. In particular, let $\Theta^{\operatorname{LOCC}}$ denote an LOCC superchannel, which realizes the following transformation of a channel
$\mathcal{M}_{\hat{A}\rightarrow\hat{B}}$
to a channel $\mathcal{N}_{A\rightarrow B}$
in terms of 
LOCC channels $\mathcal{L}_{A\rightarrow\hat{A}A_{M}B_{M}%
}^{\operatorname{pre}}$ and $\mathcal{L}_{A_{M}\hat{B}B_{M}}^{\operatorname{post}}$:
\begin{align}
\mathcal{N}_{A\rightarrow B}& =\Theta^{\operatorname{LOCC}}(\mathcal{M}_{\hat
{A}\rightarrow\hat{B}})\\
& \coloneqq\mathcal{L}_{A_{M}\hat{B}B_{M}}^{\operatorname{post}}%
\circ\mathcal{M}_{\hat{A}\rightarrow\hat{B}}\circ\mathcal{L}_{A\rightarrow
\hat{A}A_{M}B_{M}}^{\operatorname{pre}}.
\label{eq:superchannel-action-LOCC}
\end{align} 
Then the following inequalities hold:
\begin{align}
E_{\operatorname{sq}}(\mathcal{M}_{\hat{A}\rightarrow\hat{B}}) 
& \geq 
E_{\operatorname{sq}}(\Theta^{\operatorname{LOCC}}(\mathcal{M}_{\hat{A}\rightarrow\hat{B}%
}))\\
E_{\max}(\mathcal{M}_{\hat{A}\rightarrow\hat{B}}) 
& \geq 
E_{\max}(\Theta^{\operatorname{LOCC}}(\mathcal{M}_{\hat{A}\rightarrow\hat{B}%
})),
\end{align}
with both inequalities following because these measures do not increase under amortization, as shown in \cite{TGW14,TGW14B} and \cite{Christandl2017}, respectively, and the squashed entanglement \cite{CW04} and max-relative entropy of entanglement of states \cite{D09,Dat09} do not increase under LOCC channels.
\end{remark}

\subsection{Dual representation and additivity}

The optimization that is dual to \eqref{eq:kappa-channel-primal} is as follows:
\begin{multline}
\label{eq:a dual channel}
E^{\text{dual}}_\kappa(\cN_{A\to B}) \coloneqq 
\\
\sup_{\substack{V_{AB}^{T_B},\, W_{AB}^{T_B},
 \rho_{A}\ge 0}} \{ \log_2 \tr J_{AB}^{\cN}(V_{AB}-W_{AB}) :\\  V_{AB}+W_{AB}\le \rho_A\ox\1_B, \, \tr \rho_A=1 \}.
\end{multline}
This follows from applying the Lagrange multiplier method.
By weak duality, we have that
\begin{equation}
E^{\text{dual}}_\kappa(\cN_{A\to B}) \leq E_\kappa(\cN_{A\to B}).
\label{eq:weak-duality-kappa-channel}
\end{equation}
If the channel $\mathcal{N}_{A \to B}$ is finite-dimensional, then strong duality holds, so that
\begin{equation}
E^{\text{dual}}_\kappa(\cN_{A\to B}) = E_\kappa(\cN_{A\to B}).
\end{equation}
Furthermore, by employing the fact that $E^{\text{dual}}_\kappa(\cN_{A\to B}) = \sup_{\rho_{RA}} E^{\text{dual}}_\kappa(\cN_{A\to B}(\rho_{RA}))$, Proposition~\ref{eq:state-opt-for-E_kappa-ch}, and \eqref{eq:strong-dual-e-kappa-inf-dim}, we conclude that the following equality holds for a quantum channel~$\mathcal{N}_{A \to B}$:
\begin{equation}
E^{\text{dual}}_\kappa(\cN_{A\to B}) = E_\kappa(\cN_{A\to B}).
\label{eq:strong-dual-eq-ch-infty}
\end{equation}

The additivity of $E_\kappa$ with respect to tensor-product channels follows from both the primal and dual representations of $E_\kappa(\cN)$.

\begin{proposition}[Additivity]
\label{prop: add kappa channel}
   Given two  quantum channels $\cN_{A\to B}$ and $\cM_{A'\to B'}$, it holds that 
	\begin{align}
		E_\kappa(\cN_{A\to B} \ox\cM_{A'\to B'})=E_\kappa(\cN_{A\to B})+E_\kappa(\cM_{A'\to B'}).
	\end{align}
\end{proposition}
\begin{proof}
	The proof is similar to that of  Proposition~\ref{lemma:add}. To be self-contained, we show the details as follows.
First, by definition, we can write $E_\kappa(\mathcal{N}_{A\to B})$ as 
	\begin{multline}
	E_\kappa(\mathcal{N}_{A\to B})=\inf_{Q_{AB}\ge 0} \{\log_2 \Vert \tr_B Q_{AB}\Vert_\infty:\\
	 -Q^{T_B}_{AB} \le (J_{AB}^{\cN})^{T_B} \le Q^{T_B}_{AB}\}.
		\end{multline}
		Let $Q_{AB}$ be an arbitrary operator satisfying $-Q^{T_B}_{AB} \le (J_{AB}^{\cN})^{T_B} \le Q^{T_B}_{AB}, \, Q_{AB}\ge 0$, and let $P_{A'B'}$ be an arbitrary operator satisfying $-P^{T_{B'}}_{A'B'} \le (J_{A'B'}^{\cM})^{T_{B'}} \le P^{T_{B'}}_{A'B'}, \, P_{A'B'}\ge 0$. Then $Q_{AB}\otimes P_{A'B'}$ satisfies
		\begin{align}
		-(Q_{AB}\otimes P_{A'B'})^{T_{BB'}} & \le (J_{AB}^{\cN} \otimes J_{A'B'}^{\cM})^{T_{BB'}} \notag \\
		& \le (Q_{AB}\otimes P_{A'B'})^{T_{BB'}}, \\
		 Q_{AB}\otimes P_{A'B'}& \ge 0,
		\end{align}
		so that
		\begin{align}
		& E_\kappa(\mathcal{N}_{A\to B}\otimes \mathcal{M}_{A'\to B'}) \notag \\
		& \leq
		\log_2 \Vert\tr_{BB'} Q_{AB}\otimes P_{A'B'} \Vert_\infty \\
		& = \log_2 \Vert\tr_{B} Q_{AB} \Vert_\infty +
		\log_2 \Vert\tr_{B'}P_{A'B'} \Vert_\infty.
		\end{align}		
 Since the inequality holds for all $Q_{AB}$ and $P_{A'B'}$ satisfying the above conditions, we conclude that
	\begin{align}\label{eq:a sub channel}
		E_\kappa(\cN \ox \cM)\le E_\kappa(\cN)+E_\kappa(\cM).
	\end{align}
	
	To see the super-additivity of $E_\kappa$ for quantum channels, let us suppose
	that $\{V^1_{AB},W^1_{AB},\rho^1_A\}$ and $\{V^2_{A'B'},W^2_{A'B'},\rho^2_{A'}\}$
	are arbitrary operators satisfying the conditions in \eqref{eq:a dual channel} for
	$\cN_{A\to B}$ and $\cM_{A'\to B'}$, respectively.
	Now we choose 
	\begin{align}
		R_{ABA'B'}
		&=V^1_{AB} \ox V^2_{A'B'}
		+ W^1_{AB} \ox W^2_{A'B'},\\
		S_{ABA'B'}
		&=V^1_{AB} \ox W^2_{A'B'}
		+ W^1_{AB} \ox V^2_{A'B'}.
	\end{align}
	One can verify from \eqref{eq:a dual channel} that
	\begin{align}
	R_{ABA'B'}^{T_{BB'}},\, S_{ABA'B'}^{T_{BB'}}& \ge 0, \\
	R_{ABA'B'}+S_{ABA'B'} & =(V^1_{AB}+W^1_{AB})\ox(V^2_{A'B'}+W^2_{A'B'})\notag \\
	& \le \rho_{A}^1\ox\rho_{A'}^2\ox\1_{BB'},
	\end{align}
	which implies that $\{R_{ABA'B'},S_{ABA'B'},\rho^1_{A}\ox\rho^2_{A'}\}$ is feasible for 
	$E_\kappa(\mathcal{N}_{A\to B}\ox\mathcal{M}_{A'\to B'})$
	in \eqref{eq:a dual channel}. 
	Thus, we have that
	\begin{align}	 
	& E^{\text{dual}}_\kappa(\cN_{A\to B}\ox\cM_{A'\to B'}) \notag \\
	&
	\ge \log_2 \tr (J_{AB}^{\cN}\ox J_{A'B'}^{\cM}) (R_{ABA'B'}-S_{ABA'B'})\\
	&=\log_2 [\tr J_{AB}^{\cN} (V^1_{AB}-W^1_{AB}) \cdot \tr J_{A'B'}^{\cM} (V^2_{A'B'}-W^2_{A'B'})]\\
	&= \log_2 (\tr J_{AB}^{\cN}  (V^1_{AB}-W^1_{AB})) \notag \\
	& \qquad + \log_2( \tr J_{A'B'}^{\cM}  (V^2_{A'B'}-W^2_{A'B'})).
	\end{align}
	Since the inequality has been shown for arbitrary $\{V^1_{AB},W^1_{AB},\rho_{A}^1\}$ and $\{V^2_{A'B'},W^2_{A'B'},\rho_{A'}^2\}$
	satisfying the conditions in \eqref{eq:a dual channel} for
	$\cN_{A\to B}$ and $\cM_{A'\to B'}$, respectively, we conclude that
	\begin{multline}
	E^{\text{dual}}_\kappa(\cN_{A\to B}\ox\cM_{A'\to B'}) \geq 
E^{\text{dual}}_\kappa(\cN_{A\to B})\\
+E^{\text{dual}}_\kappa(\cM_{A'\to B'}).
\label{eq:add-dual-ineq-ch}
	\end{multline}
	The proof is concluded by combining 
\eqref{eq:a sub channel}, \eqref{eq:add-dual-ineq-ch}, and \eqref{eq:strong-dual-eq-ch-infty}.
\end{proof}

\subsection{Normalization, faithfulness, and no convexity}

In this subsection, we prove that the $\kappa$-entanglement of a quantum channel is normalized, faithful, and generally not convex.

\begin{proposition}
[Normalization]
\label{prop:normalization-channels}
Let $\operatorname{id}_{A\to B}^{M}$ be a noiseless quantum channel with dimension~$d_A=d_B=M$. Then%
\begin{equation}
E_{\kappa}(\operatorname{id}^{M})=\log_2 M.
\end{equation}
Moreover, for every finite-dimensional quantum channel $\cN_{A\to B}$,
\begin{equation}
E_{\kappa}(\cN_{A\to B})\le \min \{\log_2 d_A, \log_2 d_B\}.
\end{equation}
\end{proposition}

\begin{proof}
By Propositions~\ref{prop:normalize MES} and \ref{eq:state-opt-for-E_kappa-ch}, we have
\begin{align}
E_{\kappa}(\cN_{A\to B})& =\sup_{\rho_{RA}}E_{\kappa
}(\cN_{A\to B}(\rho_{RA}
))\\
&  = \sup_{\psi_{RA}}E_{\kappa
}(\cN_{A\to B}(\psi_{RA}%
)) \\
& \le \log_2 \min \{ d_A,  d_B\},
\end{align}
where, in the second equality, the optimization is with respect to pure states with system $R$ isomorphic to the channel input system $A$.

This implies that $E_{\kappa}(\operatorname{id}^{M})\le\log_2 M$. Furthermore,
\begin{align}
E_{\kappa}(\operatorname{id}^{M})
\geq E_{\kappa
}(\operatorname{id}_{A\rightarrow B}(\Phi^M_{RA}%
)) = \log_2 M,
\end{align}
where $\Phi^M_{RA}$ denotes a maximally entangled state of Schmidt rank~$M$ and the second equality follows from Proposition~\ref{prop:normalize MES}.
\end{proof}

\begin{proposition}
[Faithfulness]
\label{prop:faithful-channels}
Let $\cN_{A\to B}$  be a quantum channel. Then $E_{\kappa}(\cN_{A\to B})\geq0$ and $E_{\kappa}(\cN_{A\to B})=0$ if and
only if $\cN_{A\to B}$ is a PPT entanglement binding channel \cite{HHH00}.
\end{proposition}

\begin{proof}
To see that $E_{\kappa}(\cN_{A\to B})\geq 0$, we could utilize the dual representation in \eqref{eq:a dual channel} and the equality in \eqref{eq:strong-dual-eq-ch-infty}, or alternatively employ Propositions~\ref{prop:faithfulness} and \ref{eq:state-opt-for-E_kappa-ch} to find that
\begin{align}
E_{\kappa}(\mathcal{N}_{A\rightarrow B})=\sup_{\rho_{RA}}E_{\kappa
}(\cN_{A\rightarrow B}(\rho_{RA}%
)) \ge 0.
\end{align}

Now if $\cN_{A\to B}$ is a PPT entanglement binding channel (as defined in \cite{HHH00}), then the state $\cN_{A\rightarrow B}(\rho_{RA})$ is PPT for every input state $\rho_{RA}$. Thus, $E_{\kappa}(\mathcal{N}_{A\rightarrow B})=0$.  On the other hand, if
$E_{\kappa}(\mathcal{N}_{A\rightarrow B})=0$, then 
for every $\rho_{RA}$ it holds that $E_{\kappa}(\mathcal{N}_{A\rightarrow B}(\rho_{RA}))=0$. By Proposition~\ref{prop:faithfulness}, we conclude that $\cN_{A\rightarrow B}(\rho_{RA})$ is PPT for every state $\rho_{RA}$, and thus $\cN_{A\to B}$ is a PPT entanglement binding channel.
\end{proof}

\begin{proposition}[No convexity] The $\kappa$-entanglement of quantum channel  is not generally convex.
\end{proposition}
\begin{proof}
To see this, we construct channels with Choi states given by the examples in Eq.~\eqref{state:no convex}. Let us choose the following qubit channels:
\begin{align}
\cN_1(\rho)&=\rho,  \\
\cN_2(\rho)&=\proj{0}\rho\proj{0}+\proj{1}\rho\proj{1}.
\end{align}
Since $\cN_1$ is a qubit noiseless channel, Proposition~\ref{prop:normalization-channels} implies that $\EK(\cN_1)=1$. Noting that $\cN_2$ is a PPT entanglement binding channel, Proposition~\ref{prop:faithful-channels} implies that $\EK(\cN_2)=0$.

Let $\cN=\frac{1}{2} (\cN_1+\cN_2)$ denote the uniform mixture of the two channels.
The mixed channel $\cN$ is actually a dephasing channel with dephasing parameter $1/2$. Then we have that $\EK(\cN) \geq \log_2 \frac{3}{2}$, which follows by inputting one share of the maximally entangled state. 
Thus, we find that	
\begin{align}
\EK(\cN) 
		> \frac 1 2(\EK(\cN_1)+\EK(\cN_1)).
\end{align}	
This concludes the proof.	
\end{proof}

\section{Exact entanglement cost of quantum channels}

\label{sec:cost channel}

In this section, we introduce two channel simulation tasks. First, we consider the exact parallel simulation of a quantum channel, when completely-PPT-preserving channels are allowed for free and the goal is to meter the entanglement cost. We also consider the exact sequential simulation of a quantum channel. In both cases, the entanglement cost is equal to the $\kappa$-entanglement of the channel, thus endowing it with a direct operational meaning. After these results are established, we focus on PPT-simulable  \cite{KW17a} and resource-seizable \cite{Wilde2018} channels, demonstrating that the theory significantly simplifies for these kinds of channels.

\subsection{Exact parallel simulation of quantum channels}

Another fundamental problem is to quantify the entanglement
required for an exact simulation of an arbitrary quantum
channel, via free channels (LOCC or PPT) and by making use of an entangled resource state.
Recall that $\Omega$ represents the set of free channels. Also,  two quantum channels $\cN_{A\to B}$ and $\cM_{A\to B}$ are equal if for orthonormal bases $\{\ket{i}_A\}_i$ and $\{\ket{k}_B\}_k$, the following equalities hold for all $i,j,k,l\in \mathbb{N}$:
\begin{equation}
\bra{k}_B\cN_{A\to B}(\ket{i}_A \bra{j}_A)\ket{l}_B = \bra{k}_B\cM_{A\to B}(\ket{i}_A \bra{j}_A)\ket{l}_B.
\label{eq:channel-equality}
\end{equation}
This is equivalent to the Choi operators of the channels being equal:
\begin{equation}
\cN_{A\to B}(\Gamma_{RA}) = \cM_{A\to B}(\Gamma_{RA}).
\label{eq:channel-equality-Choi}
\end{equation}
Furthermore, the following identity holds for an arbitrary state $\rho_{CS}$ with $S \simeq R \simeq A$:
\begin{equation}
\bra{\Gamma}_{SR} [\rho_{CS}
\otimes  \cN_{A\to B}(\Gamma_{RA})]\ket{\Gamma}_{SR} = \cN_{A\to B}(\rho_{CA}),\label{eq:PS-tele}
\end{equation}
understood intuitively as a post-selected variant \cite{B05,HM04} of quantum teleportation \cite{Bennett1993}.
From the identity in \eqref{eq:PS-tele}, we conclude that if two channels are equal in the sense of \eqref{eq:channel-equality} and \eqref{eq:channel-equality-Choi}, then there is no physical procedure that can distinguish them.

We define the one-shot exact entanglement cost of a quantum channel $\cN_{A\to B}$, under the $\Omega$ channels, as
\begin{multline}
E^{(1)}_{\O}(\cN_{A\to B})= \inf_{\Lambda\in \Omega}\big\{\log_2 d: \\
  \cN_{A\to B}(\Gamma_{RA})=\Lambda_{\hat{A}\hat{B}A\to B} (\Gamma_{RA} \otimes \Phi^{d}_{\hat{A}\hat{B}})\big \}.
\end{multline}
Figure~\ref{fig:one-shot channel} depicts this simulation task.
 The exact parallel entanglement cost of quantum channel $\cN_{A\to B}$, under the $\Omega$ channels, is defined as
\begin{align}
E_{\O}^{(p)}(\cN_{A\to B})= \limsup_{n \to \infty} \frac{1}{n}E^{(1)}_{\O}(\cN_{A\to B}^{\ox n}).
\end{align}

\begin{figure}
\begin{center}
\includegraphics[
width=\linewidth
]{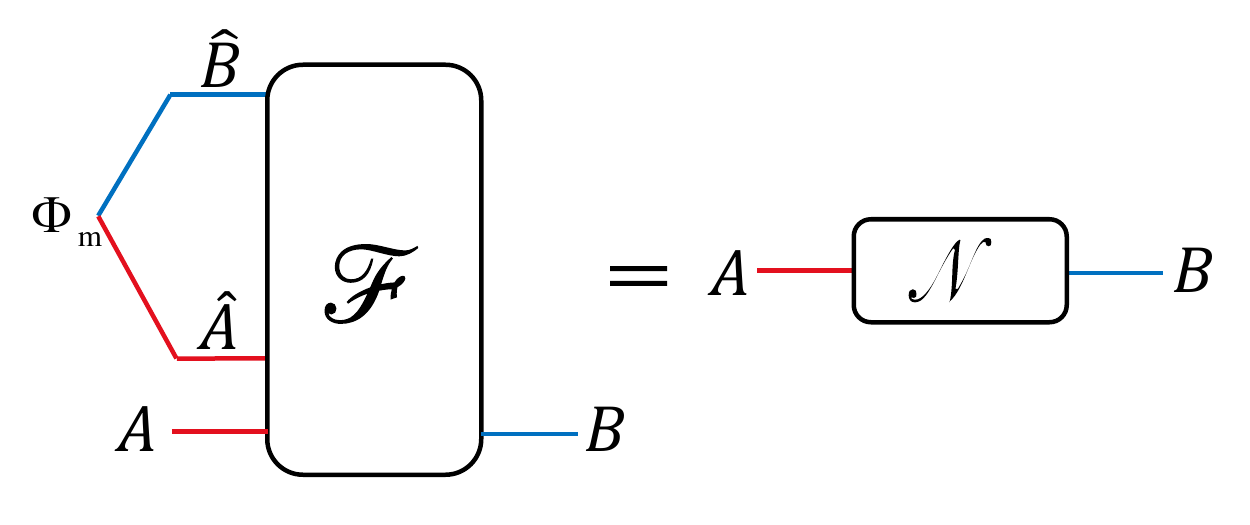}
\end{center}
\caption{Simulating the quantum channel $\cN$ via a free channel $\cF_{A\hat{A} \hat{B}\to B}$ and a maximally entangled state~$\Phi_m$.}
	\label{fig:one-shot channel}
\end{figure}

\begin{theorem}
\label{prop:one-shot channel}
The one-shot exact PPT-entanglement cost $E_{\operatorname{PPT}}
^{(1)}(\mathcal{N}_{A\rightarrow B})$ of a quantum channel $\mathcal{N}_{A\rightarrow
B}$ is given by the following optimization:
\begin{multline}
E_{\operatorname{PPT}}^{(1)}(\mathcal{N}_{A\rightarrow B})=
\inf_{m\in \mathbb{Z}^+, Q_{AB}\geq
0}\big\{
\log_2 m:\operatorname{Tr}_{B}Q_{AB}=\1_{A},\\
-\left(  m-1\right)  Q_{AB}^{T_{B}}\leq(J_{AB}^{\mathcal{N}}%
)^{T_{B}}\leq\left(  m+1\right)  Q_{AB}^{T_{B}} \big\}  .\label{eq:ch-sim-thm-parallel}%
\end{multline}
\end{theorem}
\begin{proof}
The proof is somewhat similar to the proof of Theorem~\ref{thm:exact-cost-states}, which is available in \cite{WW20}. The achievability part
features a construction of a completely-PPT-preserving channel $\mathcal{P}%
_{\hat{A}\hat{B}\rightarrow AB}$ such that $\mathcal{P}_{A\hat{A}\hat
{B}\rightarrow B}(X_{A}\otimes\Phi_{\hat{A}\hat{B}}^{m})=\mathcal{N}%
_{A\rightarrow B}(X_{A})$ for every input operator $X_{A}$ (including density
operators), and then the converse part demonstrates that the constructed
channel is essentially the only form that is needed to consider for the
one-shot exact PPT-entanglement cost task.

First, in order to have an exact simulation of a channel, it is only necessary
to check the simulation on a single input, the maximally entangled
vector $\ket{\Gamma}_{RA}$. So we require that%
\begin{equation}
\mathcal{P}_{A\hat{A}\hat{B}\rightarrow B}(\Gamma_{RA}\otimes\Phi_{\hat{A}%
\hat{B}}^{m})=\mathcal{N}_{A\rightarrow B}(\Gamma_{RA}%
),\label{eq:simulated-channel-ach-ch-sim}%
\end{equation}
where $\Gamma_{RA}$ is the unnormalized maximally entangled operator.

We now prove the achievability part. Let $m\geq1$ be a positive integer and
$Q_{AB}$ a Choi operator for a quantum channel (i.e., $Q_{AB}\geq
0,\ \operatorname{Tr}_{B}Q_{AB}=\1_{A}$) such that the following inequalities
hold%
\begin{equation}
-\left(  m-1\right)  Q_{AB}^{T_{B}}\leq(J_{AB}^{\mathcal{N}})^{T_{B}}%
\leq\left(  m+1\right)  Q_{AB}^{T_{B}}.\label{eq:ch-cond-sim-for-PPT-pres}%
\end{equation}
Then we take the completely-PPT-preserving channel $\mathcal{P}_{A\hat{A}%
\hat{B}\rightarrow B}$ to have a Choi operator given by%
\begin{equation}
J_{A\hat{A}\hat{B}B}^{\mathcal{P}}=J_{AB}^{\mathcal{N}}\otimes\Phi_{\hat
{A}\hat{B}}^{m}+Q_{AB}\otimes(\1_{\hat{A}\hat{B}}-\Phi_{\hat{A}\hat{B}}^{m}).
\end{equation}
Observe that $J_{A\hat{A}\hat{B}B}^{\mathcal{P}}\geq0$. Furthermore, we have
that%
\begin{align}
& \operatorname{Tr}_{B}J_{A\hat{A}\hat{B}B}^{\mathcal{P}}  \notag\\
& =\operatorname{Tr}%
_{B}J_{AB}^{\mathcal{N}}\otimes\Phi_{\hat{A}\hat{B}}^{m}+\operatorname{Tr}%
_{B}Q_{AB}\otimes(\1_{\hat{A}\hat{B}}-\Phi_{\hat{A}\hat{B}}^{m})\\
& =\1_{A}\otimes\Phi_{\hat{A}\hat{B}}^{m}+\1_{A}\otimes(\1_{\hat{A}\hat{B}}%
-\Phi_{\hat{A}\hat{B}}^{m})\\
& =\1_{A\hat{A}\hat{B}}.
\end{align}
Thus, $\mathcal{P}_{A\hat{A}\hat{B}\rightarrow B}$ is a quantum channel.
Setting $|\Gamma\rangle_{AA^{\prime}\hat{A}\hat{A}^{\prime}\hat{B}\hat
{B}^{\prime}}\coloneqq |\Gamma\rangle_{AA^{\prime}}\otimes|\Gamma\rangle_{\hat
{A}\hat{A}^{\prime}}\otimes|\Gamma\rangle_{\hat{B}\hat{B}^{\prime}}$, its
action on the input $\Gamma_{RA}\otimes\Phi_{\hat{A}\hat{B}}^{m}$ is given by
\begin{widetext}
\begin{align}
& \langle\Gamma|_{AA^{\prime}\hat{A}\hat{A}^{\prime}\hat{B}\hat{B}^{\prime}%
}\left(  \Gamma_{RA}\otimes\Phi_{\hat{A}\hat{B}}^{m}\otimes J_{A^{\prime}%
\hat{A}^{\prime}\hat{B}^{\prime}B}^{\mathcal{P}}\right)  |\Gamma
\rangle_{AA^{\prime}\hat{A}\hat{A}^{\prime}\hat{B}\hat{B}^{\prime}%
}\nonumber\\
& =\langle\Gamma|_{AA^{\prime}\hat{A}\hat{A}^{\prime}\hat{B}\hat{B}^{\prime}%
}\left(  \Gamma_{RA}\otimes\Phi_{\hat{A}\hat{B}}^{m}\otimes J_{A^{\prime}%
B}^{\mathcal{N}}\otimes\Phi_{\hat{A}^{\prime}\hat{B}^{\prime}}^{m}\right)
|\Gamma\rangle_{AA^{\prime}\hat{A}\hat{A}^{\prime}\hat{B}\hat{B}^{\prime}%
}\nonumber\\
& \qquad+\langle\Gamma|_{AA^{\prime}\hat{A}\hat{A}^{\prime}\hat{B}\hat
{B}^{\prime}}\left(  \Gamma_{RA}\otimes\Phi_{\hat{A}\hat{B}}^{m}\otimes
Q_{A^{\prime}B}\otimes(\1_{\hat{A}^{\prime}\hat{B}^{\prime}}-\Phi_{\hat
{A}^{\prime}\hat{B}^{\prime}}^{m})\right)  |\Gamma\rangle_{AA^{\prime}\hat
{A}\hat{A}^{\prime}\hat{B}\hat{B}^{\prime}}\\
& =\langle\Gamma|_{AA^{\prime}}(\Gamma_{RA}\otimes J_{A^{\prime}%
B}^{\mathcal{N}})|\Gamma\rangle_{AA^{\prime}}\\
& =\mathcal{N}_{A\rightarrow B}(\Gamma_{RA}).
\end{align}
The second equality follows because%
\begin{align}
\left(  \langle\Gamma|_{\hat{A}\hat{A}^{\prime}}\otimes\langle\Gamma|_{\hat
{B}\hat{B}^{\prime}}\right)  \left(  \Phi_{\hat{A}\hat{B}}^{m}\otimes
\Phi_{\hat{A}^{\prime}\hat{B}^{\prime}}^{m}\right)  \left(  |\Gamma
\rangle_{\hat{A}\hat{A}^{\prime}}\otimes|\Gamma\rangle_{\hat{B}\hat{B}%
^{\prime}}\right)    & =\operatorname{Tr}\Phi_{\hat{A}\hat{B}}^{m}\Phi
_{\hat{A}\hat{B}}^{m}=1,\\
\left(  \langle\Gamma|_{\hat{A}\hat{A}^{\prime}}\otimes\langle\Gamma|_{\hat
{B}\hat{B}^{\prime}}\right)  \left(  \Phi_{\hat{A}\hat{B}}^{m}\otimes
\1_{\hat{A}^{\prime}\hat{B}^{\prime}}\right)  \left(  |\Gamma\rangle_{\hat
{A}\hat{A}^{\prime}}\otimes|\Gamma\rangle_{\hat{B}\hat{B}^{\prime}}\right)
& =\operatorname{Tr}\Phi_{\hat{A}\hat{B}}^{m}=1.
\end{align}
Thus, for the constructed channel, we have that
\eqref{eq:simulated-channel-ach-ch-sim}  holds. Finally, we need to show that
the constructed channel $\mathcal{P}_{A\hat{A}\hat{B}\rightarrow B}$ is
completely-PPT-preserving:%
\begin{equation}
(J_{A\hat{A}\hat{B}B}^{\mathcal{P}})^{T_{\hat{B}B}}\geq
0.\label{eq:need-PPT-pres-ch-sim}%
\end{equation}
Consider that
\begin{align}
 (J_{A\hat{A}\hat{B}B}^{\mathcal{P}})^{T_{\hat{B}B}}
& =(J_{AB}^{\mathcal{N}})^{T_{B}}\otimes(\Phi_{\hat{A}\hat{B}}^{m}%
)^{T_{\hat{B}}}+Q_{AB}^{T_{B}}\otimes(\1_{\hat{A}\hat{B}}-\Phi_{\hat{A}\hat{B}%
}^{m})^{T_{\hat{B}}}\\
& =\frac{1}{m}(J_{AB}^{\mathcal{N}})^{T_{B}}\otimes(F_{\hat{A}\hat{B}}%
)+Q_{AB}^{T_{B}}\otimes(\1_{\hat{A}\hat{B}}-\frac{1}{m}F_{\hat{A}\hat{B}})\\
& =\frac{1}{m}(J_{AB}^{\mathcal{N}})^{T_{B}}\otimes(\Pi_{\hat{A}\hat{B}%
}^{\mathcal{S}}-\Pi_{\hat{A}\hat{B}}^{\mathcal{A}})+Q_{AB}^{T_{B}}\otimes
(\Pi_{\hat{A}\hat{B}}^{\mathcal{S}}+\Pi_{\hat{A}\hat{B}}^{\mathcal{A}}%
-\frac{1}{m}[\Pi_{\hat{A}\hat{B}}^{\mathcal{S}}-\Pi_{\hat{A}\hat{B}%
}^{\mathcal{A}}])\\
& =\left[  \frac{1}{m}(J_{AB}^{\mathcal{N}})^{T_{B}}+\left(  1-\frac{1}%
{m}\right)  Q_{AB}^{T_{B}}\right]  \otimes\Pi_{\hat{A}\hat{B}}^{\mathcal{S}%
}+\left[  \left(  1+\frac{1}{m}\right)  Q_{AB}^{T_{B}}-\frac{1}{m}%
(J_{AB}^{\mathcal{N}})^{T_{B}}\right]  \otimes\Pi_{\hat{A}\hat{B}%
}^{\mathcal{A}}\\
& =\frac{1}{m}\left[  (J_{AB}^{\mathcal{N}})^{T_{B}}+\left(  m-1\right)
Q_{AB}^{T_{B}}\right]  \otimes\Pi_{\hat{A}\hat{B}}^{\mathcal{S}}+\frac{1}%
{m}\left[  \left(  m+1\right)  Q_{AB}^{T_{B}}-(J_{AB}^{\mathcal{N}})^{T_{B}%
}\right]  \otimes\Pi_{\hat{A}\hat{B}}^{\mathcal{A}}.
\end{align}
\end{widetext}
Applying the condition in \eqref{eq:ch-cond-sim-for-PPT-pres}, we conclude
\eqref{eq:need-PPT-pres-ch-sim}. Thus, we have shown that for all $m$ and
$Q_{AB}$ satisfying \eqref{eq:ch-cond-sim-for-PPT-pres} and $Q_{AB}%
\geq0,\ \operatorname{Tr}_{B}Q_{AB}=\1_{A}$, there exists a
completely-PPT-preserving channel $\mathcal{P}_{A\hat{A}\hat{B}\rightarrow B}$
such that \eqref{eq:simulated-channel-ach-ch-sim} holds. Now taking an infimum
over all such $m$ and $Q_{AB}$, we conclude that the right-hand side of
\eqref{eq:ch-sim-thm-parallel} is greater than or equal to
$E_{\operatorname{PPT}}^{(1)}(\mathcal{N}_{A\rightarrow B})$.

To see the opposite inequality, let $\mathcal{P}_{A\hat{A}\hat{B}\rightarrow
B}$ be a completely-PPT-preserving channel such that
\eqref{eq:simulated-channel-ach-ch-sim} holds. Then preceding $\mathcal{P}%
_{A\hat{A}\hat{B}\rightarrow B}$ by the isotropic twirling channel
$\mathcal{T}_{\hat{A}\hat{B}}$ results in a completely-PPT-preserving channel
$\mathcal{P}_{A\hat{A}\hat{B}\rightarrow B}^{\prime}=\mathcal{P}_{A\hat{A}%
\hat{B}\rightarrow B}\circ\mathcal{T}_{\hat{A}\hat{B}}$ achieving the same
simulation task, and so it suffices to focus on the channel $\mathcal{P}%
_{A\hat{A}\hat{B}\rightarrow B}^{\prime}$ in order to establish an expression
for the one-shot exact PPT-entanglement cost. Consider that%
\begin{align}
J_{R\hat{A}^{\prime}\hat{B}^{\prime}B}^{\mathcal{P}^{\prime}}  &
=\mathcal{P}_{A\hat{A}\hat{B}\rightarrow B}^{\prime}(\Gamma_{RA}\otimes
\Gamma_{\hat{A}^{\prime}\hat{A}}\otimes\Gamma_{\hat{B}^{\prime}\hat{B}%
})\nonumber\\
& =(\mathcal{P}_{A\hat{A}\hat{B}\rightarrow B}\circ\mathcal{T}_{\hat{A}\hat
{B}})(\Gamma_{RA}\otimes\Gamma_{\hat{A}^{\prime}\hat{A}}\otimes\Gamma_{\hat
{B}^{\prime}\hat{B}}).
\end{align}
Considering that
\begin{align}
& \mathcal{T}_{\hat{A}\hat{B}}(\Gamma_{\hat{A}^{\prime}\hat{A}}\otimes
\Gamma_{\hat{B}^{\prime}\hat{B}})  \notag \\
& =\Phi_{\hat{A}\hat{B}}^{m}\otimes
\operatorname{Tr}_{\hat{A}\hat{B}}[\Phi_{\hat{A}\hat{B}}^{m}(\Gamma_{\hat
{A}^{\prime}\hat{A}}\otimes\Gamma_{\hat{B}^{\prime}\hat{B}})]\nonumber\\
& \qquad+\frac{\1_{\hat{A}\hat{B}}-\Phi_{\hat{A}\hat{B}}^{m}}{m^{2}%
-1}\operatorname{Tr}_{\hat{A}\hat{B}}[(\1_{\hat{A}\hat{B}}-\Phi_{\hat{A}\hat
{B}}^{m})(\Gamma_{\hat{A}^{\prime}\hat{A}}\otimes\Gamma_{\hat{B}^{\prime}%
\hat{B}})]\\
& =\Phi_{\hat{A}\hat{B}}^{m}\otimes\Phi_{\hat{A}^{\prime}\hat{B}^{\prime}}%
^{m}+\frac{\1_{\hat{A}\hat{B}}-\Phi_{\hat{A}\hat{B}}^{m}}{m^{2}-1}%
\otimes(\1_{\hat{A}\hat{B}}-\Phi_{\hat{A}\hat{B}}^{m}),
\end{align}
with the equalities understood in terms of entanglement swapping \cite{Bennett1993}, 
we conclude that
\begin{align}
& (\mathcal{P}_{A\hat{A}\hat{B}\rightarrow B}\circ\mathcal{T}_{\hat{A}\hat{B}%
})(\Gamma_{RA}\otimes\Gamma_{\hat{A}^{\prime}\hat{A}}\otimes\Gamma_{\hat
{B}^{\prime}\hat{B}})\nonumber\\
& =\mathcal{P}_{A\hat{A}\hat{B}\rightarrow B}(\Gamma_{RA}\otimes\Phi_{\hat
{A}\hat{B}}^{m})\otimes\Phi_{\hat{A}^{\prime}\hat{B}^{\prime}}^{m}\notag \\
& \qquad +\mathcal{P}_{A\hat{A}\hat{B}\rightarrow B}\!\left(  \Gamma_{RA}\otimes
\frac{\1_{\hat{A}\hat{B}}-\Phi_{\hat{A}\hat{B}}^{m}}{m^{2}-1}\right)
\otimes(\1_{\hat{A}\hat{B}}-\Phi_{\hat{A}\hat{B}}^{m})\\
& =\mathcal{N}_{A\rightarrow B}(\Gamma_{RA})\otimes\Phi_{\hat{A}^{\prime}
\hat{B}^{\prime}}^{m}+ \notag \\
& \qquad \mathcal{P}_{A\hat{A}\hat{B}\rightarrow B}\!\left(
\Gamma_{RA}\otimes\frac{\1_{\hat{A}\hat{B}}-\Phi_{\hat{A}\hat{B}}^{m}}{m^{2}%
-1}\right)  \otimes(\1_{\hat{A}^{\prime}\hat{B}^{\prime}}-\Phi_{\hat{A}%
^{\prime}\hat{B}^{\prime}}^{m})\\
& =J_{RB}^{\mathcal{N}}\otimes\Phi_{\hat{A}^{\prime}\hat{B}^{\prime}}%
^{m}+Q_{RB}\otimes(\1_{\hat{A}^{\prime}\hat{B}^{\prime}}-\Phi_{\hat{A}^{\prime
}\hat{B}^{\prime}}^{m}).
\end{align}
where we have used the assumption that \eqref{eq:simulated-channel-ach-ch-sim}
holds and set%
\begin{equation}
Q_{RB}=\mathcal{P}_{A\hat{A}\hat{B}\rightarrow B}\!\left(  \Gamma_{RA}%
\otimes\frac{\1_{\hat{A}\hat{B}}-\Phi_{\hat{A}\hat{B}}^{m}}{m^{2}-1}\right)  ,
\end{equation}
from which it follows that $Q_{RB}\geq0$ and $\operatorname{Tr}_{B}%
Q_{RB}=\1_{R}$. In order for the channel $\mathcal{P}_{A\hat{A}\hat
{B}\rightarrow B}^{\prime}$\ to be completely-PPT-preserving, it is necessary
that%
\begin{equation}
(J_{R\hat{A}^{\prime}\hat{B}^{\prime}B}^{\mathcal{P}^{\prime}})^{T_{\hat
{B}^{\prime}B}}\geq0.
\end{equation}
Writing this out and using calculations given above, we find that it is
necessary that the following operator is positive semi-definite:
\begin{multline}
\frac{1}{m}\left[  (J_{AB}^{\mathcal{N}})^{T_{B}}+\left(  m-1\right)
Q_{AB}^{T_{B}}\right]  \otimes\Pi_{\hat{A}\hat{B}}^{\mathcal{S}}\\
+\frac{1}%
{m}\left[  \left(  m+1\right)  Q_{AB}^{T_{B}}-(J_{AB}^{\mathcal{N}})^{T_{B}%
}\right]  \otimes\Pi_{\hat{A}\hat{B}}^{\mathcal{A}}.
\end{multline}
Since $\Pi_{\hat{A}\hat{B}}^{\mathcal{S}}$ and $\Pi_{\hat{A}\hat{B}%
}^{\mathcal{A}}$ project onto orthogonal subspaces, we find that the condition
\eqref{eq:ch-cond-sim-for-PPT-pres} is necessary. Thus, it follows that the
quantity on the right-hand side of \eqref{eq:ch-sim-thm-parallel}\ is less
than or equal to $E_{\operatorname{PPT}}^{(1)}(\mathcal{N}_{A\rightarrow B})$.
\end{proof}

\update{
\begin{proposition}
\label{prop:ch-one-shot-bnd}
Let $\mathcal{N}_{A\rightarrow B}$ be a quantum channel.
Then%
\begin{equation}
\log_2(2^{E_{\kappa}(\mathcal{N})}-1)\leq E_{\operatorname{PPT}}^{(1)}%
(\mathcal{N}_{A\rightarrow B})\leq\log_2(2^{E_{\kappa}(\mathcal{N})}+2).
\end{equation}
\end{proposition}
}
\begin{proof}
The idea of the proof is to use the technique of SDP relaxation. Consider that
\begin{widetext}
\begin{align}
  E_{\operatorname{PPT}}^{(1)}(\mathcal{N}_{A\rightarrow B}) 
&  =\inf\left\{  \log_2 m:-\left(  m-1\right)  Q_{AB}^{T_{B}}\leq(J_{AB}%
^{\mathcal{N}})^{T_{B}}\leq\left(  m+1\right)  Q_{AB}^{T_{B}},\ Q_{AB}%
\geq0,\ \operatorname{Tr}_{B}Q_{AB}=\1_{A}\right\} \notag \\
&  \geq\inf\left\{  \log_2 m:-\left(  m+1\right)  Q_{AB}^{T_{B}}\leq
(J_{AB}^{\mathcal{N}})^{T_{B}}\leq\left(  m+1\right)  Q_{AB}^{T_{B}}%
,\ Q_{AB}\geq0,\ \operatorname{Tr}_{B}Q_{AB}=\1_{A}\right\} \notag \\
&  =\inf\left\{  \log_2 m:-R_{AB}^{T_{B}}\leq(J_{AB}^{\mathcal{N}})^{T_{B}}\leq
R_{AB}^{T_{B}},\ R_{AB}\geq0,\ \operatorname{Tr}_{B}R_{AB}=\left(  m+1\right)
\1_{A}\right\} \notag \\
&  =\inf\left\{  \log_2(\left\Vert \operatorname{Tr}_{B}R_{AB}\right\Vert
_{\infty}-1):-R_{AB}^{T_{B}}\leq(J_{AB}^{\mathcal{N}})^{T_{B}}\leq
R_{AB}^{T_{B}},\ R_{AB}\geq0\right\} \notag \\
&  =\log_2(2^{E_{\kappa}(\mathcal{N})}-1).
\end{align}
The first inequality follows by relaxing the constraint $-\left(  m-1\right)
Q_{AB}^{T_{B}}\leq(J_{AB}^{\mathcal{N}})^{T_{B}}$ to $-\left(  m+1\right)
Q_{AB}^{T_{B}}\leq(J_{AB}^{\mathcal{N}})^{T_{B}}$. The second equality follows
by absorbing $m$ into $Q_{AB}$ and setting $R_{AB}=\left(  m+1\right)  Q_{AB}%
$. The last equality follows from the definition of $E_{\kappa}(\mathcal{N})$.

Similarly, we have that $E_{\operatorname{PPT}}^{(1)}(\mathcal{N}%
_{A\rightarrow B})\leq\log_2(2^{E_{\kappa}(\mathcal{N})}+2)$ following the chain of inequalities:
\begin{align}
 E_{\text{PPT}}^{(1)}(\cN_{A\to B}) 
&  =\inf\left\{  \log_2 m:-\left(  m-1\right)  Q_{AB}^{T_{B}}\leq(J_{AB}%
^{\mathcal{N}})^{T_{B}}\leq\left(  m+1\right)  Q_{AB}^{T_{B}},\ Q_{AB}%
\geq0,\ \operatorname{Tr}_{B}Q_{AB}=\1_{A}, m\in\mathbb{N}, m \geq 2\right\} \notag \\
& \le \inf\left\{  \log_2 m:-\left(  m-1\right)  Q_{AB}^{T_{B}}\leq(J_{AB}%
^{\mathcal{N}})^{T_{B}}\leq\left(  m-1\right)  Q_{AB}^{T_{B}},\ Q_{AB}%
\geq0,\ \operatorname{Tr}_{B}Q_{AB}=\1_{A}, m\in\mathbb{N}, m \geq 2\right\} \notag \\
& = \inf\left\{  \log_2 \left\lfloor \mu\right\rfloor:-\left(  \left\lfloor \mu\right\rfloor-1\right)  Q_{AB}^{T_{B}}\leq(J_{AB}%
^{\mathcal{N}})^{T_{B}}\leq\left(  \left\lfloor \mu\right\rfloor-1\right)  Q_{AB}^{T_{B}},\ Q_{AB}%
\geq0,\ \operatorname{Tr}_{B}Q_{AB}=\1_{A},  \mu \geq 2\right\} \notag \\
& \le \inf\left\{  \log_2 \left\lfloor \mu\right\rfloor:-\left(  \mu-2\right)  Q_{AB}^{T_{B}}\leq(J_{AB}%
^{\mathcal{N}})^{T_{B}}\leq\left(  \mu-2\right)  Q_{AB}^{T_{B}},\ Q_{AB}%
\geq0,\ \operatorname{Tr}_{B}Q_{AB}=\1_{A},  \mu \geq 2\right\} \notag \\
& \le \inf\left\{  \log_2  \mu:-\left(  \mu-2\right)  Q_{AB}^{T_{B}}\leq(J_{AB}%
^{\mathcal{N}})^{T_{B}}\leq\left(  \mu-2\right)  Q_{AB}^{T_{B}},\ Q_{AB}%
\geq0,\ \operatorname{Tr}_{B}Q_{AB}=\1_{A},  \mu \geq 2\right\} \notag \\
&  =\inf\left\{  \log_2(\left\Vert \operatorname{Tr}_{B}R_{AB}\right\Vert
_{\infty}+2):-R_{AB}^{T_{B}}\leq(J_{AB}^{\mathcal{N}})^{T_{B}}\leq
R_{AB}^{T_{B}},\ R_{AB}\geq0\right\} \notag \\
&  =\log_2(2^{E_{\kappa}(\mathcal{N})}+2).
\end{align}
\end{widetext}
The first inequality follows since we choose more restricted condition $\left(  m-1\right)  Q_{AB}^{T_{B}}\leq(J_{AB}%
^{\mathcal{N}})^{T_{B}}\leq\left(  m-1\right)  Q_{AB}^{T_{B}}$. The second inequality follows since  $-(\left\lfloor \mu\right\rfloor
-1)
\leq -(\mu-2)$ and $\mu-2\leq\left\lfloor \mu\right\rfloor
-1$. In this case, the set over which we are optimizing becomes smaller. The third inequality follows since $\left\lfloor \mu\right\rfloor\le\mu$ in the loss function. We also take $R_{AB}=(\mu-2)Q_{AB}$ to simplify the optimization and then arrive at the final equality following the  definition of $E_{\kappa}(\mathcal{N})$.
\end{proof}

\update{
\begin{theorem}[Exact parallel  cost]
\label{thm:parallel-cost-asymp}
Let $\mathcal{N}_{A\rightarrow B}$ be a  quantum channel.
Then the exact parallel entanglement cost of $\mathcal{N}_{A\rightarrow B}$ is equal to its $\kappa$-entanglement:
\begin{equation}
E^{(p)}_{\operatorname{PPT}}(\mathcal{N}_{A\rightarrow B})=E_{\kappa}(\mathcal{N}%
_{A\rightarrow B}).
\end{equation}
\end{theorem}
}
\begin{proof}
The main idea behind the proof is to employ the one-shot bound in Proposition~\ref{prop:ch-one-shot-bnd}
and then the additivity relation from Proposition~\ref{prop: add kappa channel}. Consider that%
\begin{align}
E^{(p)}_{\operatorname{PPT}}(\mathcal{N}_{A\rightarrow B})  & =\limsup
_{n\rightarrow\infty}\frac{1}{n}E_{\operatorname{PPT}}^{(1)}(\mathcal{N}%
_{A\rightarrow B}^{\otimes n})\\
& \leq\limsup_{n\rightarrow\infty}\frac{1}{n}\log_2(2^{E_{\kappa}(\mathcal{N}%
^{\otimes n})}+2)\\
& =\limsup_{n\rightarrow\infty}\frac{1}{n}\log_2(2^{nE_{\kappa}(\mathcal{N}%
)}+2)\\
& =E_{\kappa}(\mathcal{N}_{A\rightarrow B}).
\end{align}
Similarly, $E_{\operatorname{PPT}}(\mathcal{N}_{A\rightarrow B})\geq
E_{\kappa}(\mathcal{N}_{A\rightarrow B})$.
\end{proof}

%
%
 
\subsection{Exact sequential simulation of quantum channels}

A more general notion of channel simulation, called sequential channel
simulation, was recently proposed and studied in \cite{Wilde2018}. In this section, we define
and characterize exact sequential channel simulation, as opposed to the
approximate sequential channel simulation focused on in \cite{Wilde2018}. For concreteness, we
set the free channels $\Omega$ to be completely-PPT-preserving channels. The
main idea behind sequential channel simulation is to simulate $n$ uses of the
channel $\mathcal{N}_{A\rightarrow B}$\ in such a way that they can be called
in an arbitrary order, i.e., on demand when they are needed. An $(n,M)$ exact
sequential channel simulation code consists of a maximally entangled resource
state $\Phi_{\overline{A}_{0}\overline{B}_{0}}^{M}$\ of Schmidt rank $M$ and a
set%
\begin{equation}
\{\mathcal{P}_{A_{i}\overline{A}_{i-1}\overline{B}_{i-1}\rightarrow
B_{i}\overline{A}_{i}\overline{B}_{i}}^{(i)}\}_{i=1}^{n}\label{eq:sim-prot}%
\end{equation}
of completely-PPT-preserving channels. Note that the systems $\overline{A}%
_{n}\overline{B}_{n}$ of the final completely-PPT-preserving channel $\mathcal{P}_{A_{n}%
\overline{A}_{n-1}\overline{B}_{n-1}\rightarrow B_{n}\overline{A}_{n}%
\overline{B}_{n}}^{(n)}$ can be taken trivial without loss of generality. As
before, Alice has access to all systems labeled by $A$, Bob has access to all
systems labeled by $B$, and they are in distant laboratories. The structure of
this simulation protocol is intended to be compatible with a discrimination
strategy that can test the actual $n$ channels versus the above simulation in
a sequential way, along the lines discussed in \cite{CDP08a,GDP09}%
\ and\ \cite{G12}.

We define the simulation to be exact if the following equalities hold for
orthonormal bases $\{|i\rangle_{A}\}_{A}$ and $\{|k\rangle_{B}\}_{k}$ and for
all $i_{1},j_{1},k_{1},l_{1},\ldots,i_{n},j_{n},k_{n},l_{n}\in\mathbb{N}$:
\begin{widetext}
\begin{equation}
p^{\{i_{r},j_{r},k_{r},l_{r}\}_{r=1}^{n}}=\prod\limits_{r=1}^{n}\langle
k_{r}|_{B_{r}}\mathcal{N}_{A_{r}\rightarrow B_{r}}(|i_{r}\rangle\!\langle
j_{r}|_{A_{r}})|l_{r}\rangle_{B_{r}}, \label{eq:equality-sequential-sim}%
\end{equation}
where%
\begin{align}
P_{\overline{A}_{1}\overline{B}_{1}}^{i_{1},j_{1},k_{1},l_{1}}  &  \coloneqq\langle
k_{1}|_{B_{1}}\left[  \mathcal{P}_{A_{1}\overline{A}_{0}\overline{B}%
_{0}\rightarrow B_{1}\overline{A}_{1}\overline{B}_{1}}^{(1)}(|i_{1}%
\rangle\!\langle j_{1}|_{A_{1}}\otimes\Phi_{\overline{A}_{0}\overline{B}_{0}%
}^{M})\right]  |l_{1}\rangle_{B_{1}},\\
P_{\overline{A}_{2}\overline{B}_{2}}^{i_{2},j_{2},k_{2},l_{2},i_{1}%
,j_{1},k_{1},l_{1}}  &  \coloneqq\langle k_{2}|_{B_{2}}\left[  \mathcal{P}%
_{A_{2}\overline{A}_{1}\overline{B}_{1}\rightarrow B_{2}\overline{A}%
_{2}\overline{B}_{2}}^{(2)}(|i_{2}\rangle\!\langle j_{2}|_{A_{2}}\otimes
P_{\overline{A}_{1}\overline{B}_{1}}^{i_{1},j_{1},k_{1},l_{1}})\right]
|l_{2}\rangle_{B_{2}},\\
&  \vdots\nonumber\\
P_{\overline{A}_{n-1}\overline{B}_{n-1}}^{\{i_{r},j_{r},k_{r},l_{r}%
\}_{r=1}^{n-1}}  &  \coloneqq\langle k_{n-1}|_{B_{n-1}}[\mathcal{P}_{A_{n-1}%
\overline{A}_{n-2}\overline{B}_{n-2}\rightarrow B_{n-1}\overline{A}%
_{n-1}\overline{B}_{n-1}}^{(n-1)}(|i_{n-1}\rangle\!\langle j_{n-1}|_{A_{n-1}%
}\nonumber\\
&  \qquad\otimes P_{\overline{A}_{n-2}\overline{B}_{n-2}}^{\{i_{r},j_{r}%
,k_{r},l_{r}\}_{r=1}^{n-2}})]|l_{n-1}\rangle_{B_{n-1}},\\
p^{\{i_{r},j_{r},k_{r},l_{r}\}_{r=1}^{n}}  &  \coloneqq\langle k_{n}|_{B_{n}}\left[
\mathcal{P}_{A_{n}\overline{A}_{n-1}\overline{B}_{n-1}\rightarrow B_{n}}%
^{(n)}(|i_{n}\rangle\!\langle j_{n}|_{A_{n}}\otimes P_{\overline{A}%
_{n-1}\overline{B}_{n-1}}^{\{i_{r},j_{r},k_{r},l_{r}\}_{r=1}^{n-1}})\right]
|l_{n}\rangle_{B_{n}}.
\end{align}
\end{widetext}
Figure~\ref{fig:sequential-sim} depicts the channel simulation and the exact
simulation condition in \eqref{eq:equality-sequential-sim}.

\begin{figure}[ptb]
\begin{center}
\includegraphics[
width=\linewidth
]{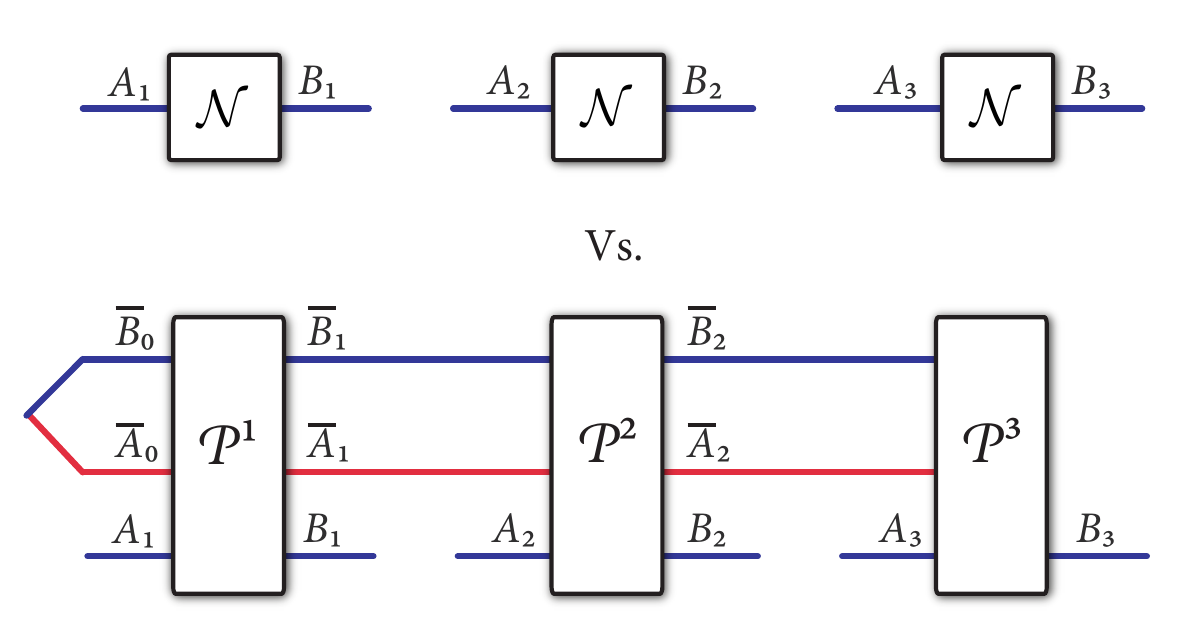}
\end{center}
\caption{The top part of the figure depicts the $n=3$ sequential uses of the
channel $\mathcal{N}_{A\rightarrow B}$\ that should be simulated. The bottom
part of the figure depicts the simulation. The simulation is considered to be
exact, as written in \eqref{eq:equality-sequential-sim}, if, after inputting
the operator $|i_{r}\rangle\!\langle j_{r}|_{A_{r}}$ to the input system $A_{r}$
and contracting the output system $B_{r}$ in terms of $\langle k_{r}|_{B_{r}%
}(\cdot)|l_{r}\rangle_{B_{r}}$, the resulting numbers are the same for both
the original channels and their simulation, for all possible $|i_{r}%
\rangle_{A_{r}}$, $|j_{r}\rangle_{A_{r}}$, $|k_{r}\rangle_{B_{r}}$, and
$|l_{r}\rangle_{B_{r}}$ and for $r\in\left\{  1,\ldots,n\right\}  $.}%
\label{fig:sequential-sim}%
\end{figure}

By defining the completely-PPT-preserving quantum channel $\mathcal{P}%
_{A^{n}\overline{A}_{0}\overline{B}_{0}\rightarrow B^{n}}$ as the serial
composition of the individual channels in \eqref{eq:sim-prot} (depicted in
Figure~\ref{fig:simulation-parallel})%
\begin{multline}
\mathcal{P}_{A^{n}\overline{A}_{0}\overline{B}_{0}\rightarrow B^{n}%
}\coloneqq(\mathcal{P}_{A_{n}\overline{A}_{n-1}\overline{B}_{n-1}\rightarrow B_{n}%
}^{(n)}\circ\\
\mathcal{P}_{A_{n-1}\overline{A}_{n-2}\overline{B}_{n-2}%
\rightarrow B_{n-1}\overline{A}_{n-1}\overline{B}_{n-1}}^{(n-1)}\circ
\cdots \circ\label{eq:serial-comp-protocol}\\
\mathcal{P}_{A_{2}\overline{A}_{1}\overline{B}_{1}\rightarrow
B_{2}\overline{A}_{2}\overline{B}_{2}}^{(2)}\circ\mathcal{P}_{A_{1}%
\overline{A}_{0}\overline{B}_{0}\rightarrow B_{1}\overline{A}_{1}\overline
{B}_{1}}^{(1)}),
\end{multline}
we conclude that the condition in \eqref{eq:equality-sequential-sim}\ is
equivalent to the following condition:%
\begin{equation}
(\mathcal{N}_{A\rightarrow B})^{\otimes n}(\Gamma_{R^{n}A^{n}})=\mathcal{P}%
_{A^{n}\overline{A}_{0}\overline{B}_{0}\rightarrow B^{n}}(\Gamma_{R^{n}A^{n}%
}\otimes\Phi_{\overline{A}_{0}\overline{B}_{0}}^{M}%
),\label{eq:choi-equivalence}%
\end{equation}
where $\Gamma_{R^{n}A^{n}}\coloneqq\bigotimes\limits_{i=1}^{n}\Gamma_{R_{i}A_{i}}$.
This latter condition is depicted in Figure~\ref{fig:sequential-sim-choi}.

\begin{figure}[ptb]
\begin{center}
\includegraphics[
width=\linewidth
]{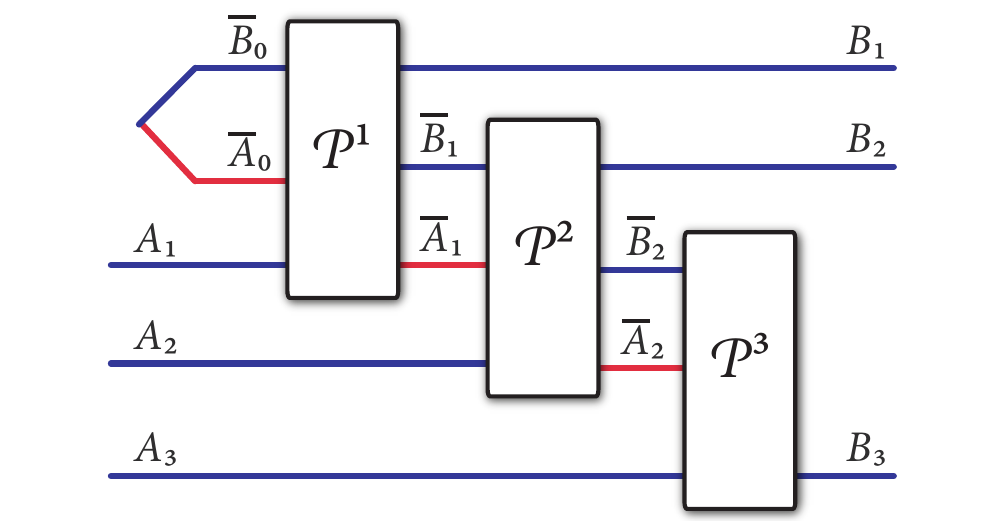}
\end{center}
\caption{The channel in \eqref{eq:serial-comp-protocol}, defined as the serial
composition of the completely-PPT-preserving channels in the simulation.}%
\label{fig:simulation-parallel}%
\end{figure}

\begin{figure}[ptb]
\begin{center}
\includegraphics[
width=\linewidth
]{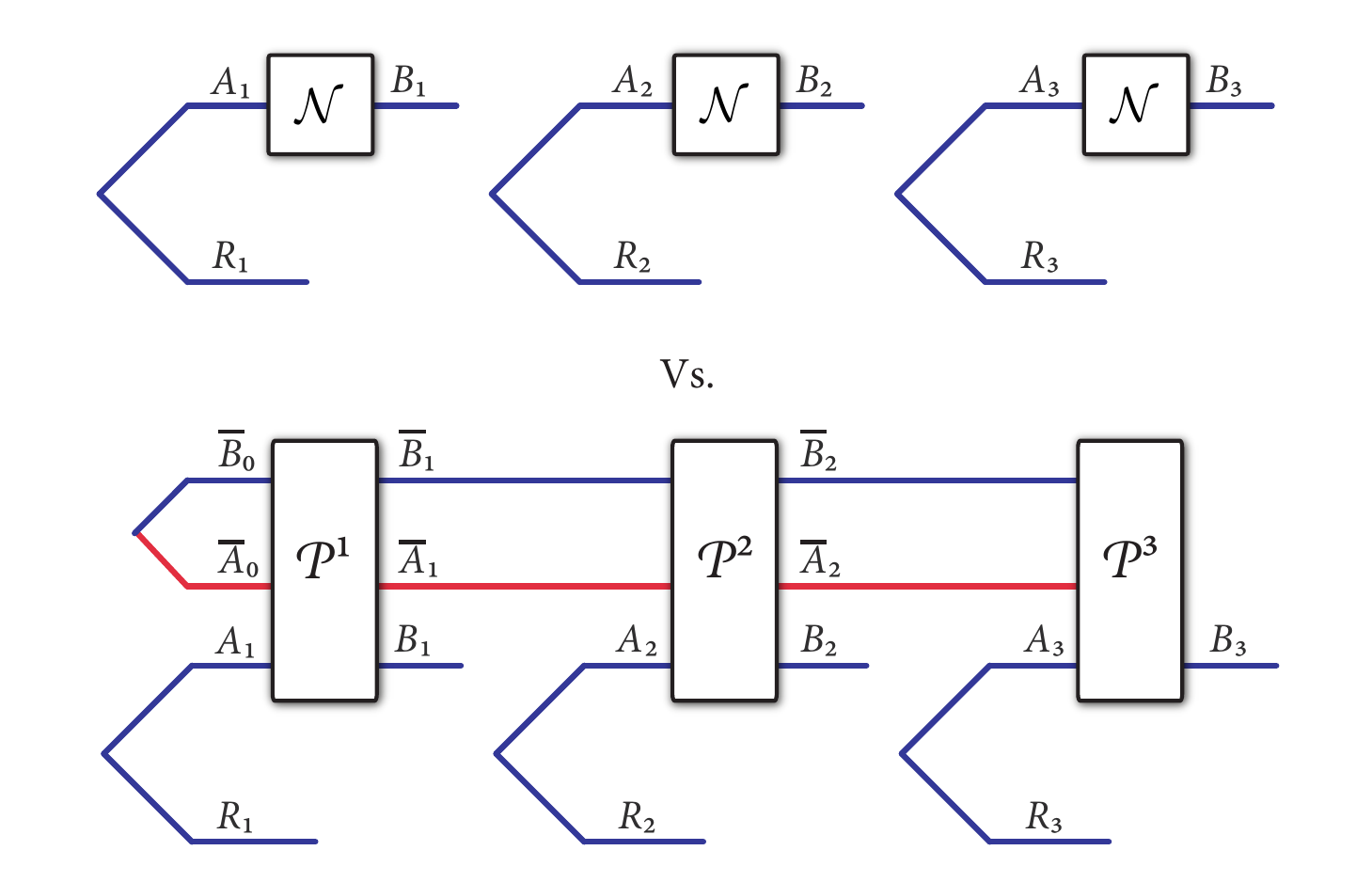}
\end{center}
\caption{The exact channel simulation condition in
\eqref{eq:equality-sequential-sim} is equivalent to the condition that the
Choi operators as depicted above are equal, as written in
\eqref{eq:choi-equivalence}.}%
\label{fig:sequential-sim-choi}%
\end{figure}

The $n$-shot exact sequential simulation cost of the channel $\mathcal{N}%
_{A\rightarrow B}$ is then defined as%
\begin{multline}
E_{\operatorname{PPT}}(\mathcal{N}_{A\rightarrow B},n)\coloneqq\inf\big\{  \log_2
M:\\
(\mathcal{N}_{A\rightarrow B})^{\otimes n}(\Gamma_{R^{n}A^{n}}%
)=\mathcal{P}_{A^{n}\overline{A}_{0}\overline{B}_{0}\rightarrow B^{n}}%
(\Gamma_{R^{n}A^{n}}\otimes\Phi_{\overline{A}_{0}\overline{B}_{0}}%
^{M})\big\}  ,
\end{multline}
where the optimization is with respect to sequential protocols of the form in
\eqref{eq:sim-prot} and the channel $\mathcal{P}_{A^{n}\overline{A}%
_{0}\overline{B}_{0}\rightarrow B^{n}}$ is defined as in
\eqref{eq:serial-comp-protocol}. The exact (sequential)\ simulation cost of
the channel $\mathcal{N}_{A\rightarrow B}$ is defined as%
\begin{equation}
E_{\operatorname{PPT}}(\mathcal{N}_{A\rightarrow B})\coloneqq\limsup_{n\rightarrow
\infty}\frac{1}{n}E_{\operatorname{PPT}}(\mathcal{N}_{A\rightarrow B},n).
\end{equation}

The condition in \eqref{eq:choi-equivalence}\ illustrates that a sequential
simulation is a particular kind of parallel simulation, but with more
constraints. That is, in a parallel simulation, the channel $\mathcal{P}%
_{A^{n}\overline{A}_{0}\overline{B}_{0}\rightarrow B^{n}}$ can be arbitrary,
whereas in a sequential simulation, it is constrained to have the form in
\eqref{eq:sim-prot}. For this reason, we can immediately conclude the
following bound for every integer $n\geq 1$:%
\begin{equation}
E_{\operatorname{PPT}}^{(1)}((\mathcal{N}_{A\rightarrow B})^{\otimes n})\leq
E_{\operatorname{PPT}}(\mathcal{N}_{A\rightarrow B}%
,n),\label{eq:parallel-lower-seq}%
\end{equation}
which in turn implies that%
\begin{equation}
E_{\operatorname{PPT}}^{(p)}(\mathcal{N}_{A\rightarrow B})\leq
E_{\operatorname{PPT}}(\mathcal{N}_{A\rightarrow B}).
\end{equation}

\subsection{Physical justification for definition of exact sequential channel
simulation}

The most general method for distinguishing the $n$ channel uses from its
simulation is with an adaptive discrimination strategy. Such a strategy was
described in \cite{Wilde2018} and consists of an initial state $\rho_{R_{1}A_{1}}$, a set
$\{\mathcal{A}_{R_{i}B_{i}\rightarrow R_{i+1}A_{i+1}}^{(i)}\}_{i=1}^{n-1}$ of
adaptive channels, and a quantum measurement $\{Q_{R_{n}B_{n}},\1_{R_{n}B_{n}%
}-Q_{R_{n}B_{n}}\}$. Let us employ the shorthand $\{\rho,\mathcal{A},Q\}$ to
abbreviate such a discrimination strategy. Note that, in performing a
discrimination strategy, the discriminator has a full description of the
channel $\mathcal{N}_{A\rightarrow B}$ and the simulation protocol, which
consists of $\Phi_{\overline{A}_{0}\overline{B}_{0}}$ and the set in
\eqref{eq:sim-prot}. If this discrimination strategy is performed on the $n$
uses of the actual channel $\mathcal{N}_{A\rightarrow B}$, the relevant states
involved are%
\begin{equation}
\rho_{R_{i+1}A_{i+1}}\coloneqq \mathcal{A}_{R_{i}B_{i}\rightarrow R_{i+1}A_{i+1}%
}^{(i)}(\rho_{R_{i}B_{i}}),
\end{equation}
for $i\in\left\{  1,\ldots,n-1\right\}  $ and
\begin{equation}
\rho_{R_{i}B_{i}}\coloneqq \mathcal{N}_{A_{i}\rightarrow B_{i}}(\rho_{R_{i}A_{i}%
}),
\end{equation}
for $i\in\left\{  1,\ldots,n\right\}  $. If this discrimination strategy is
performed on the simulation protocol discussed above, then the relevant states
involved are%
\begin{align}
\tau_{R_{1}B_{1}\overline{A}_{1}\overline{B}_{1}} &  \coloneqq \mathcal{P}%
_{A_{1}\overline{A}_{0}\overline{B}_{0}\rightarrow B_{1}\overline{A}%
_{1}\overline{B}_{1}}^{(1)}(\tau_{R_{1}A_{1}}\otimes\Phi_{\overline{A}%
_{0}\overline{B}_{0}}),\nonumber\\
\tau_{R_{i+1}A_{i+1}\overline{A}_{i}\overline{B}_{i}} &  \coloneqq \mathcal{A}%
_{R_{i}B_{i}\rightarrow R_{i+1}A_{i+1}}^{(i)}(\tau_{R_{i}B_{i}\overline{A}%
_{i}\overline{B}_{i}}),
\end{align}
for $i\in\left\{  1,\ldots,n-1\right\}  $, where $\tau_{R_{1}A_{1}}%
=\rho_{R_{1}A_{1}}$, and
\begin{equation}
\tau_{R_{i}B_{i}\overline{A}_{i}\overline{B}_{i}}\coloneqq \mathcal{P}%
_{A_{i}\overline{A}_{i-1}\overline{B}_{i-1}\rightarrow B_{i}\overline{A}%
_{i}\overline{B}_{i}}^{(i)}(\tau_{R_{i}A_{i}\overline{A}_{i-1}\overline
{B}_{i-1}}),
\end{equation}
for $i\in\left\{  2,\ldots,n\right\}  $. The discriminator then performs the
measurement $\{Q_{R_{n}B_{n}},\1_{R_{n}B_{n}}-Q_{R_{n}B_{n}}\}$ and guesses
\textquotedblleft actual channel\textquotedblright\ if the outcome is
$Q_{R_{n}B_{n}}$ and \textquotedblleft simulation\textquotedblright\ if the
outcome is $\1_{R_{n}B_{n}}-Q_{R_{n}B_{n}}$. Figure~\ref{fig:adaptive-prot}%
\ depicts the discrimination strategy in the case that the actual channel is
called $n=3$ times and in the case that the simulation is performed.%

\begin{figure}
[ptb]
\begin{center}
\includegraphics[
width=\linewidth
]%
{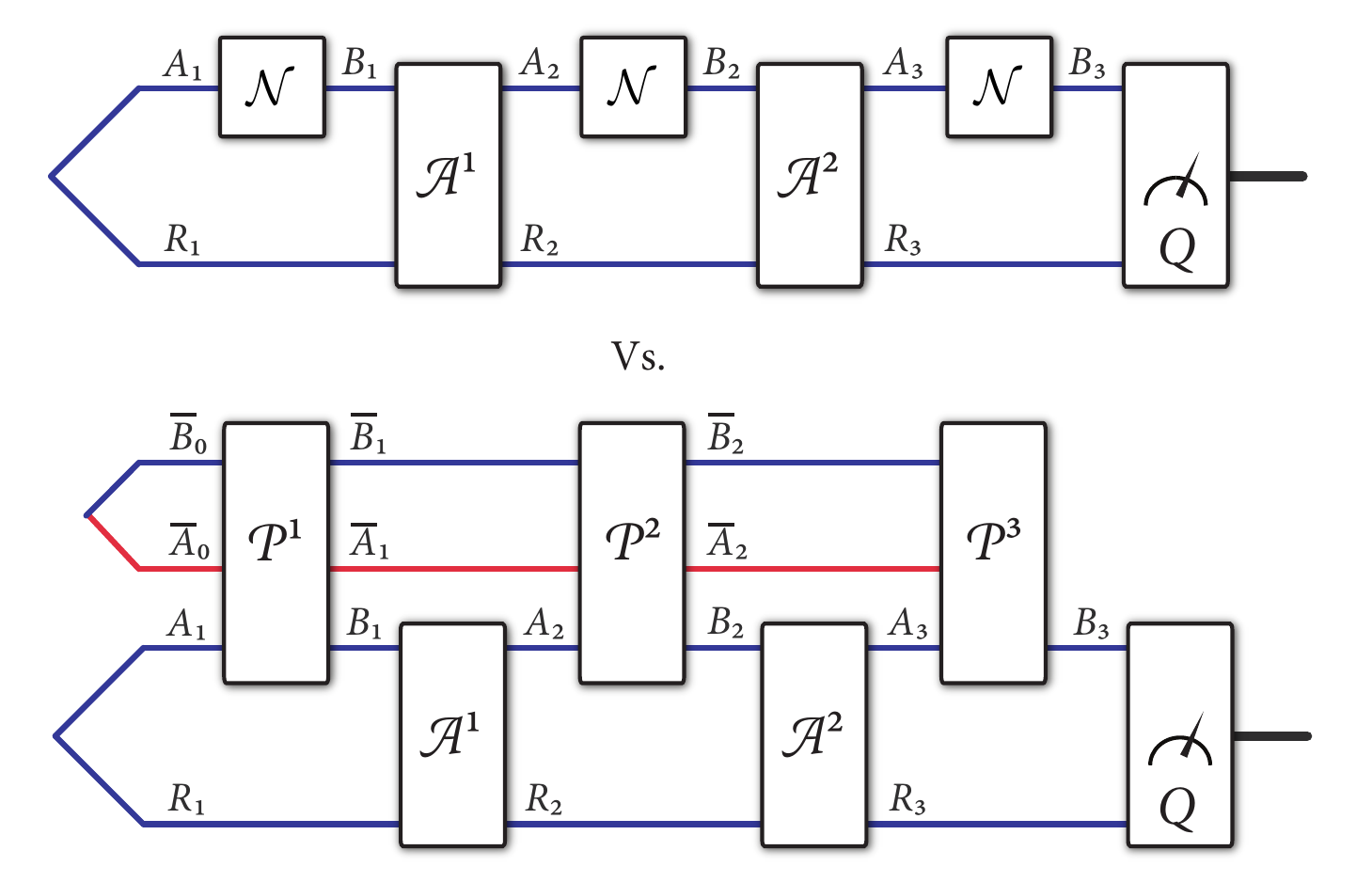}%
\caption{An adaptive protocol for discriminating the original channels (top) from
their simulation (bottom).}%
\label{fig:adaptive-prot}%
\end{center}
\end{figure}

From the physical point of view, the $n$ channel uses of $\mathcal{N}%
_{A\rightarrow B}$ are perfectly indistinguishable from the simulation if
every possible discrimination strategy as described above leads to the exact
same final decision probabilities. That is, for all possible discrimination
strategies, the original channels and their simulation are indistinguishable if the following equality holds%
\begin{equation}
\operatorname{Tr}Q_{R_{n}B_{n}}\rho_{R_{n}B_{n}}=\operatorname{Tr}%
Q_{R_{n}B_{n}}\tau_{R_{n}B_{n}}.\label{eq:physical-equiv-PPT-seq-sim}%
\end{equation}

We now prove that this physical notion of exact channel simulation is
equivalent to the more mathematical notion of exact channel simulation
described in the previous section. First, suppose that the physical notion of
exact channel simulation holds; i.e., the equality in
\eqref{eq:physical-equiv-PPT-seq-sim} holds for all possible discrimination
strategies. Then this means that $\rho_{R_{n}B_{n}}=\tau_{R_{n}B_{n}}$ for all
possible discrimination strategies. One possible strategy could be to pick the
input state for each system $A_{i}$ as one of the following states%
\begin{equation}
\rho_{A}^{x,y}=\left\{
\begin{array}
[c]{cc}%
|x\rangle\!\langle x|_{A} & \text{if }x=y\\
\frac{1}{2}\left(  |x\rangle_{A}+|y\rangle_{A}\right)  \left(  \langle
x|_{A}+\langle y|_{A}\right)   & \text{if }x<y\\
\frac{1}{2}\left(  |x\rangle_{A}+i|y\rangle_{A}\right)  \left(  \langle
x|_{A}-i\langle y|_{A}\right)   & \text{if }x>y
\end{array}
\right.  .\label{eq-nqt:density-op-basis}%
\end{equation}
and the output system $B_{i}$ could be measured in the same way, but with
respect to an orthonormal basis for the output system. Then all input state
choices and measurement outcomes could be stored in auxiliary classical
registers. Consider that for all $x,y$ such that $x<y$, the following holds
\begin{align}
|x\rangle\!\langle y|_{A} &  =\left(  \rho_{A}^{x,y}-\frac{1}{2}\rho_{A}%
^{x,x}-\frac{1}{2}\rho_{A}^{y,y}\right)  \notag \\
& \qquad -i\left(  \rho_{A}^{y,x}-\frac{1}%
{2}\rho_{A}^{x,x}-\frac{1}{2}\rho_{A}^{y,y}\right)  ,\\
|y\rangle\!\langle x|_{A} &  =\left(  \rho_{A}^{x,y}-\frac{1}{2}\rho_{A}%
^{x,x}-\frac{1}{2}\rho_{A}^{y,y}\right) \notag  \\
&\qquad  +i\left(  \rho_{A}^{y,x}-\frac{1}%
{2}\rho_{A}^{x,x}-\frac{1}{2}\rho_{A}^{y,y}\right)  ,
\end{align}
so that linear combinations of all the outcomes realize the operator basis discussed in the mathematical definition of equivalence.
Since the equivalence holds for all possible discrimination strategies, we can
collect the data from them in the auxiliary registers, and then finally
conclude that the condition in \eqref{eq:equality-sequential-sim}\ holds.%

\begin{figure}
[ptb]
\begin{center}
\includegraphics[
width=\linewidth
]%
{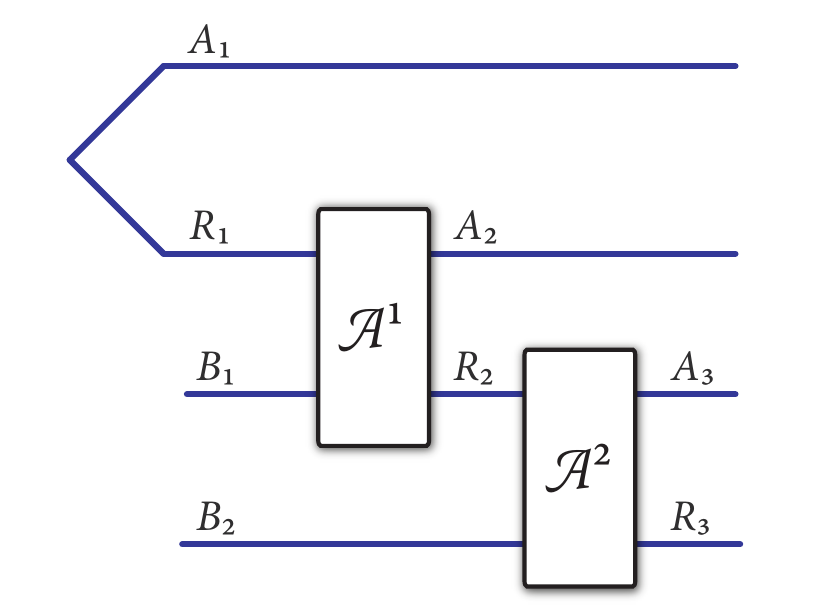}%
\caption{The discrimination strategy $\rho_{R_1 A_1}$ and $\{\mathcal{A}_{R_{i}B_{i}\rightarrow R_{i+1}A_{i+1}}^{(i)}\}_{i=1}^{n-1}$ represented as a single channel
$\mathcal{A}_{B^{n}\rightarrow A^{n}R_{n}}$, as written in \eqref{eq:disc-strat-to-parallel}.}%
\label{fig:disc-strat-parallel-channel}%
\end{center}
\end{figure}

To see that the mathematical notion of exact sequential simulation implies the
physical one, we use the method of post-selected teleportation, essentially
the same idea as what was used in the proof of \cite[Theorem~4]{BSW11}. Consider the channel defined by
the serial composition of the channels in the discrimination strategy
$\{\rho,\mathcal{A},Q\}$:
\begin{multline}
\mathcal{A}_{B^{n}\rightarrow A^{n}R_{n}}=\mathcal{A}_{R_{n-1}B_{n-1}%
\rightarrow R_{n}A_{n}}^{(n-1)}\circ\cdots\circ
\\
\mathcal{A}_{R_{2}%
B_{2}\rightarrow R_{3}A_{3}}^{(2)}\circ
\mathcal{A}_{R_{1}B_{1}\rightarrow
R_{2}A_{2}}^{(1)}\circ\rho_{R_{1}A_{1}},
\label{eq:disc-strat-to-parallel}
\end{multline}
where the notation $\rho_{R_{1}A_{1}}$ indicates a preparation channel that
tensors in the state $\rho_{R_{1}A_{1}}$.
Figure~\ref{fig:disc-strat-parallel-channel}\ depicts this channel. By acting
on both sides of the exact simulation condition with the channel and then the
projection onto
$|\Gamma\rangle\!\langle\Gamma|_{A^{n}S^{n}}$, with $S\simeq R$, we find that
\begin{widetext}
\begin{multline}
\langle\Gamma|_{A^{n}S^{n}}\left[  \mathcal{A}_{B^{n}\rightarrow A^{n}R_{n}%
}\circ(\mathcal{N}_{A\rightarrow B})^{\otimes n}(\Gamma_{S^{n}A^{n}})\right]
|\Gamma\rangle_{A^{n}S^{n}}\\
=\langle\Gamma|_{A^{n}S^{n}}\left[  \mathcal{A}_{B^{n}\rightarrow A^{n}R_{n}%
}\circ\mathcal{P}_{A^{n}\overline{A}_{0}\overline{B}_{0}\rightarrow B^{n}%
}(\Gamma_{S^{n}A^{n}}\otimes\Phi_{\overline{A}_{0}\overline{B}_{0}}%
^{M})\right]  |\Gamma\rangle_{A^{n}S^{n}}.
\label{eq:postselected-tele-arg-1}
\end{multline}
where%
\begin{equation}
|\Gamma\rangle_{A^{n}S^{n}}=|\Gamma\rangle_{A_{1}S_{1}}\otimes|\Gamma
\rangle_{A_{2}S_{2}}\otimes\cdots\otimes|\Gamma\rangle_{A_{n}S_{n}}.
\end{equation}
From the method of post-selected teleportation, we conclude that%
\begin{align}
\langle\Gamma|_{A^{n}S^{n}}\left[  \mathcal{A}_{B^{n}\rightarrow A^{n}R_{n}%
}\circ(\mathcal{N}_{A\rightarrow B})^{\otimes n}(\Gamma_{S^{n}A^{n}})\right]
|\Gamma\rangle_{A^{n}S^{n}} &  =\rho_{R_{n}B_{n}},\\
\langle\Gamma|_{A^{n}S^{n}}\left[  \mathcal{A}_{B^{n}\rightarrow A^{n}R_{n}%
}\circ\mathcal{P}_{A^{n}\overline{A}_{0}\overline{B}_{0}\rightarrow B^{n}%
}(\Gamma_{S^{n}A^{n}}\otimes\Phi_{\overline{A}_{0}\overline{B}_{0}}%
^{M})\right]  |\Gamma\rangle_{A^{n}S^{n}} &  =\tau_{R_{n}B_{n}}.
\end{align}
\end{widetext}
Putting these together, we finally conclude that
\begin{equation}
\rho_{R_{n}B_{n}}
=\tau_{R_{n}B_{n}}.
\label{eq:postselected-tele-arg-last}
\end{equation}
Thus, no physical discrimination strategy can distinguish
the original channels from their simulation if the exact simulation condition
in \eqref{eq:choi-equivalence} holds. Figure~\ref{fig:teleportation-argument} depicts the operator $\mathcal{A}_{B^{n}\rightarrow A^{n}R_{n}%
}\circ\mathcal{P}_{A^{n}\overline{A}_{0}\overline{B}_{0}\rightarrow B^{n}%
}(\Gamma_{S^{n}A^{n}}\otimes\Phi_{\overline{A}_{0}\overline{B}_{0}}%
^{M})$ in order to help visualize the above argument.%

\begin{figure}
[ptb]
\begin{center}
\includegraphics[
width=\linewidth
]%
{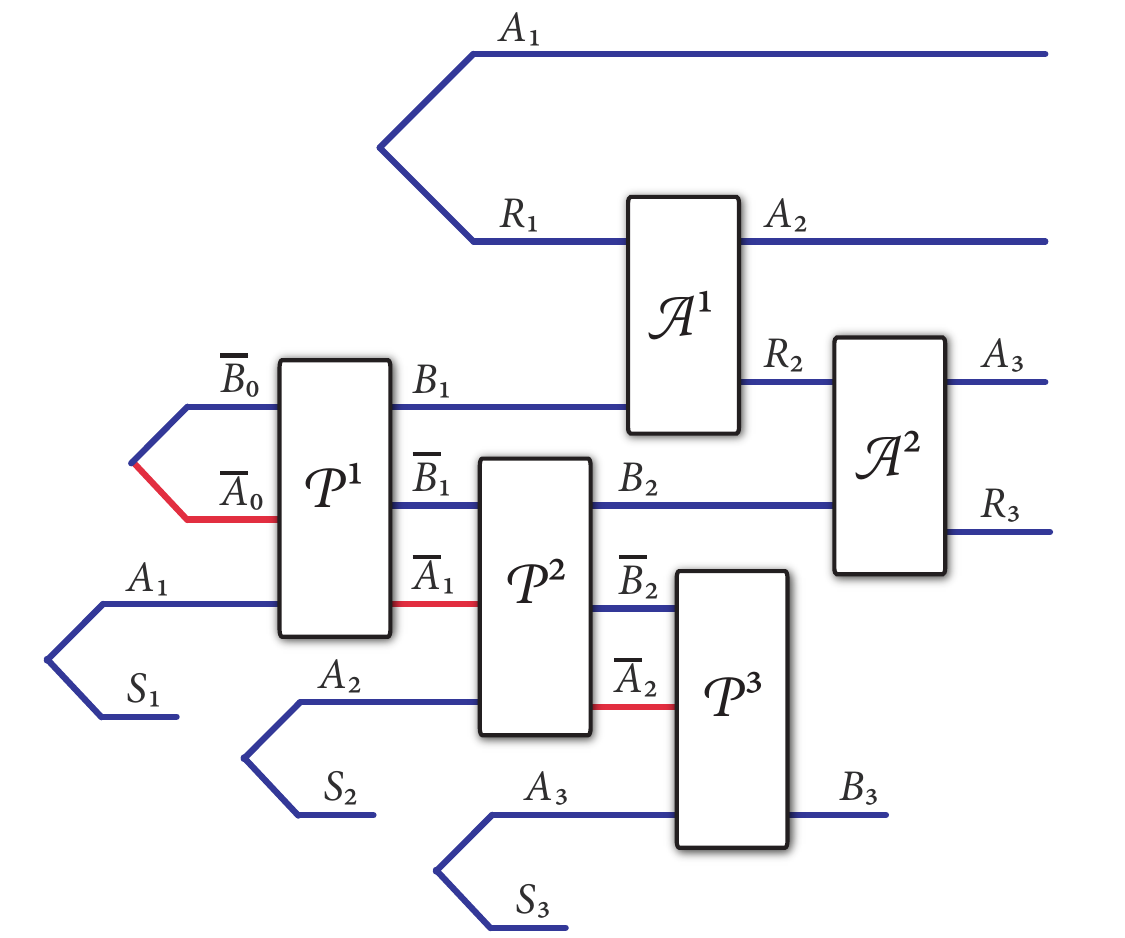}%
\caption{This figure depicts the operator 
$\mathcal{A}_{B^{n}\rightarrow A^{n}R_{n}%
}\circ\mathcal{P}_{A^{n}\overline{A}_{0}\overline{B}_{0}\rightarrow B^{n}%
}(\Gamma_{S^{n}A^{n}}\otimes\Phi_{\overline{A}_{0}\overline{B}_{0}}%
^{M})$ in
order to help visualize the
argument in \eqref{eq:postselected-tele-arg-1}--\eqref{eq:postselected-tele-arg-last}. By projecting the systems $S_{1}A_{1}$ onto $\langle
\Gamma|_{S_{1}A_{1}}$, $S_{2}A_{2}$ onto $\langle\Gamma|_{S_{2}A_{2}}$, and
$S_{3}A_{3}$ onto $\langle\Gamma|_{S_{3}A_{3}}$, the method of post-selected
teleportation guarantees that the remaining state is $\tau_{R_{3}B_{3}}$,
which is the final state of the bottom part of Figure~\ref{fig:adaptive-prot}.}%
\label{fig:teleportation-argument}%
\end{center}
\end{figure}

\subsection{Exact sequential channel simulation cost}

We first establish the following bounds on the $n$-shot exact sequential
simulation cost:

\begin{proposition}
\label{prop:n-shot-seq-bnd}
Let $\mathcal{N}_{A\rightarrow B}$ be a quantum channel such that $E_{\kappa}(\mathcal{N}) > 0$. Then the $n$-shot
exact sequential simulation cost is bounded as%
\begin{align}
\log_2\left[  2^{nE_{\kappa}(\mathcal{N})}-1\right] & \leq
E_{\operatorname{PPT}}(\mathcal{N}_{A\rightarrow B},n)\\
& \leq\log_2\left[
\frac{2^{\left(  n+1\right)  E_{\kappa}(\mathcal{N})}-1}{2^{E_{\kappa
}(\mathcal{N})}-1}\right]  .
\end{align}
If $E_{\kappa}(\mathcal{N})=0$, then 
$E_{\operatorname{PPT}}(\mathcal{N}_{A\rightarrow B},n)=0$.
\end{proposition}

\begin{proof}
Suppose that $E_{\kappa}(\mathcal{N}) > 0$.
The inequality%
\begin{equation}
\log_2\left[  2^{nE_{\kappa}(\mathcal{N})}-1\right]  \leq
E_{\operatorname{PPT}}(\mathcal{N}_{A\rightarrow B},n)
\end{equation}
is a direct consequence of \eqref{eq:parallel-lower-seq}, Proposition~\ref{prop:ch-one-shot-bnd}, and Proposition~\ref{prop: add kappa channel}.

So we now prove the other inequality. The main idea behind the construction is
for the $i$th completely-PPT-preserving channel to perform the following exact
simulation:%
\begin{multline}
\mathcal{P}_{A_{i}\overline{A}_{i-1}\overline{B}_{i-1}\rightarrow
B_{i}\overline{A}_{i}\overline{B}_{i}}^{(i)}(\rho_{A_{i}}\otimes
\Phi_{\overline{A}_{i-1}\overline{B}_{i-1}}^{M_{i-1}})=\\
\mathcal{N}%
_{A\rightarrow B}(\rho_{A_{i}})\otimes\Phi_{\overline{A}_{i}\overline{B}_{i}%
}^{M_{i}},
\label{eq:ach-constr-seq-sim}
\end{multline}
for $i\in\left\{  1,\ldots,n-1\right\}  $ and for the $n$th
completely-PPT-preserving channel to perform the following exact simulation:%
\begin{equation}
\mathcal{P}_{A_{n}\overline{A}_{n-1}\overline{B}_{n-1}\rightarrow B_{n}}%
^{(n)}(\rho_{A_{n}}\otimes\Phi_{\overline{A}_{n-1}\overline{B}_{n-1}}%
^{M_{n-1}})=\mathcal{N}_{A\rightarrow B}(\rho_{A_{n}}).
\end{equation}
Note that, in order to perform the simulation in \eqref{eq:ach-constr-seq-sim}, we could actually simulate
the channel $\mathcal{N}_{A\rightarrow B}\otimes\operatorname{id}^{M_{i}}$,
and then send one share of the maximally entangled state $\Phi_{\overline
{A}_{i}\overline{B}_{i}}^{M_{i}}$ through the exactly simulated identity
channel $\operatorname{id}^{M_{i}}$ to produce the output in \eqref{eq:ach-constr-seq-sim}.

Thus, we should now determine an upper bound on the simulation cost when using this
construction. The most effective way to do so is to start from the final
($n$th) simulation. By the one-shot bound from Proposition~\ref{prop:ch-one-shot-bnd}, its cost $\log_2 M_{n-1}$ is
bounded as%
\begin{equation}
\log_2 M_{n-1}\leq\log_2\left[  2^{E_{\kappa}(\mathcal{N})}+1\right]  .
\end{equation}
The cost $\log_2 M_{n-2}$ of the $n-1$ simulation is then bounded as%
\begin{align}
\log_2 M_{n-2}  & \leq\log_2\left[  2^{E_{\kappa}(\mathcal{N}\otimes
\operatorname{id}^{M_{n-1}})}+1\right]  \\
& \leq\log_2\left[  2^{E_{\kappa}(\mathcal{N})+\log_2 M_{n-1}}+1\right]  \\
& =\log_2\left[  2^{E_{\kappa}(\mathcal{N})}M_{n-1}+1\right]  \\
& \leq\log_2\left[  2^{E_{\kappa}(\mathcal{N})}\left(  2^{E_{\kappa}%
(\mathcal{N})}+1\right)  +1\right]  \\
& =\log_2\left[  \sum_{\ell=0}^{2}2^{\ell E_{\kappa}(\mathcal{N})}\right]  ,
\end{align}
where we made use of the subadditivity inequality from Proposition~\ref{prop: add kappa channel}.
Performing this kind of reasoning iteratively, going backward until the first
simulation, we find the following bound:%
\begin{equation}
\log_2 M_{0}\leq\log_2\left[  \sum_{\ell=0}^{n}2^{\ell E_{\kappa}(\mathcal{N}%
)}\right]  =\log_2\left[  \frac{2^{\left(  n+1\right)  E_{\kappa}(\mathcal{N}%
)}-1}{2^{E_{\kappa}(\mathcal{N})}-1}\right]  .
\end{equation}

If $E_{\kappa}(\mathcal{N})=0$, then 
the channel $\mathcal{N}$ is PPT entanglement binding by Proposition~\ref{prop:faithful-channels} and thus can be simulated at no cost, so that 
$E_{\operatorname{PPT}}(\mathcal{N}_{A\rightarrow B},n)=0$.
This concludes the proof.
\end{proof}

\begin{theorem}[Exact sequential  cost]
\label{thm:seq-ch-sim-kappa-ent}
Let $\mathcal{N}_{A\rightarrow B}$ be a quantum channel.
Then the exact sequential channel simulation cost of $\mathcal{N}%
_{A\rightarrow B}$ is equal to its $\kappa$-entanglement:%
\begin{equation}
E_{\operatorname{PPT}}(\mathcal{N}_{A\rightarrow B})=E_{\kappa}(\mathcal{N}%
_{A\rightarrow B}).
\end{equation}

\end{theorem}

\begin{proof}
First suppose that $E_{\kappa}(\mathcal{N}) > 0$.
The lower bound follows from Proposition~\ref{prop:n-shot-seq-bnd} and Theorem~\ref{thm:parallel-cost-asymp}. The upper bound
follows from Proposition~\ref{prop:n-shot-seq-bnd}:%
\begin{align}
& \limsup_{n\rightarrow\infty}\frac{1}{n}E_{\operatorname{PPT}}(\mathcal{N}%
_{A\rightarrow B},n)  \notag \\
& \leq\limsup_{n\rightarrow\infty}\frac{1}{n}\log_2\left[
\frac{2^{\left(  n+1\right)  E_{\kappa}(\mathcal{N})}-1}{2^{E_{\kappa
}(\mathcal{N})}-1}\right]  \\
& =\limsup_{n\rightarrow\infty}\frac{1}{n}\log_2\left[  \frac{2^{nE_{\kappa
}(\mathcal{N})}-2^{-E_{\kappa}(\mathcal{N})}}{1-2^{-E_{\kappa}(\mathcal{N})}%
}\right]  \\
& =E_{\kappa}(\mathcal{N}).
\end{align}

If $E_{\kappa}(\mathcal{N}) = 0$, then the channel $\mathcal{N}$ is PPT entanglement binding by Proposition~\ref{prop:faithful-channels} and thus can be simulated at no cost.
This concludes the proof.
\end{proof}

\bigskip
By combining Theorems~\ref{thm:parallel-cost-asymp} and \ref{thm:seq-ch-sim-kappa-ent}, we reach the conclusion that the exact entanglement cost of  parallel and sequential simulation of quantum channels are in fact equal and given by the $\kappa$-entanglement of the channel. Thus, the $\kappa$-entanglement is a fundamental measure of the entanglement of a quantum channel. Not only is it efficiently computable by means of a semi-definite program (for finite-dimensional channels), but it also possesses a direct operational meaning in terms of these channel simulation tasks. It is the only known channel entanglement measure possessing these properties, and from this perspective, it can be helpful in understanding the fundamental structure of entanglement of quantum channels. 

\subsection{PPT-simulable channels}

Although the theory of exact simulation of quantum channels under PPT operations simplifies significantly due to Theorems~\ref{thm:parallel-cost-asymp} and \ref{thm:seq-ch-sim-kappa-ent}, there is a class of channels for which the theory is even simpler. These channels were defined in \cite{KW17a} and are known as PPT-simulable channels. In this section, we recall their definition and show how the theory of exact entanglement cost is quite simple for certain PPT-simulable channels. 

\begin{definition}
[PPT-simulable channel \cite{KW17a}]A channel $\mathcal{N}_{A\rightarrow B}$\ is
PPT-simulable with associated resource state $\omega_{A^{\prime}B^{\prime}}$
 if there exists a completely PPT-preserving channel
$\mathcal{P}_{AA^{\prime}B^{\prime}\rightarrow B}$ such that, for every input
state $\rho_{A}$%
\begin{equation}
\mathcal{N}_{A\rightarrow B}(\rho_{A})=\mathcal{P}_{AA^{\prime}B^{\prime
}\rightarrow B}(\rho_{A}\otimes\omega_{A^{\prime}B^{\prime}}).
\label{eq:PPT-simulable-ch}
\end{equation}

\end{definition}

A particular kind of PPT-simulable channel is one that is resource-seizable, as defined in \cite[Section~VI]{Wilde2018}:

\begin{definition}
[Resource-seizable \cite{Wilde2018}]\label{def:res-seize}Let $\mathcal{N}_{A\rightarrow B}$ be
a PPT-simulable channel with associated resource state $\omega_{A^{\prime
}B^{\prime}}$. The channel $\mathcal{N}_{A\rightarrow B}$ is resource-seizable
if there exists
a PPT state $\tau_{A_{M}AB_{M}}$ and a completely PPT-preserving
post-processing channel $\mathcal{D}_{A_{M}BB_{M}\rightarrow A^{\prime
}B^{\prime}}$ such that%
\begin{equation}
\mathcal{D}_{A_{M}BB_{M}\rightarrow A^{\prime}B^{\prime}}(\mathcal{N}%
_{A\rightarrow B}(\tau_{A_{M}AB_{M}}))=\omega_{A^{\prime}B^{\prime}}.
\end{equation}

\end{definition}

For PPT-simulable channels, it follows that the exact entanglement cost of sequential channel simulation is bounded from above by the exact entanglement cost of the underlying resource state:

\begin{theorem}
\label{prop:PPT-sim}Let $\mathcal{N}_{A\rightarrow B}$ be a PPT-simulable
channel with associated resource state $\omega_{A^{\prime}B^{\prime}}$. Then
the PPT-assisted entanglement cost of a channel is bounded from above as%
\begin{equation}
E_{\operatorname{PPT}}(\mathcal{N}_{A\rightarrow B})\leq E_{\operatorname{PPT}%
}(\omega_{A^{\prime}B^{\prime}})=E_{\kappa}(\omega_{A^{\prime}B^{\prime}}).
\end{equation}

\end{theorem}

\begin{proof}
The proof for this inequality follows the same reasoning given in \cite[Corollary~1]{Wilde2018}. First simulate a large number of copies of the resource state $\omega_{A^{\prime
}B^{\prime}}$ and then use the PPT-preserving channel $\mathcal{P}%
_{AA^{\prime}B^{\prime}\rightarrow B}$ from \eqref{eq:PPT-simulable-ch} to simulate the channel $\mathcal{N}%
_{A\rightarrow B}$. The equality follows from Proposition~\ref{th:exact cost}.
\end{proof}

\bigskip

If a PPT-simulable channel is additionally resource-seizable, then its exact entanglement cost is given by the $\kappa$-entanglement of the underlying resource state:

\begin{theorem}
\label{thm:res-seize-simplify}
Let $\mathcal{N}_{A\rightarrow B}$ be a PPT-simulable channel with associated
resource state $\omega_{A^{\prime}B^{\prime}}$. Suppose furthermore that it is
resource-seizable, as given in Definition~\ref{def:res-seize}. Then%
\begin{align}
E_{\operatorname{PPT}}(\mathcal{N}_{A\rightarrow B}) & =E_{\operatorname{PPT}%
}^{(p)}(\mathcal{N}_{A\rightarrow B})=E_{\kappa}(\mathcal{N}_{A\rightarrow
B})\\
& =E_{\operatorname{PPT}}(\omega_{A^{\prime}B^{\prime}})=E_{\kappa}%
(\omega_{A^{\prime}B^{\prime}}).
\end{align}

\end{theorem}

\begin{proof}
The following inequality
\begin{equation}
E_{\operatorname{PPT}}(\mathcal{N}_{A\rightarrow B})\leq E_{\operatorname{PPT}%
}(\omega_{A^{\prime}B^{\prime}})=E_{\kappa}(\omega_{A^{\prime}B^{\prime}}).
\end{equation}
is a consequence of Theorem~\ref{prop:PPT-sim}. To establish the opposite
inequality, consider that we always have that%
\begin{equation}
E_{\operatorname{PPT}}(\mathcal{N}_{A\rightarrow B})\geq E_{\operatorname{PPT}%
}^{(p)}(\mathcal{N}_{A\rightarrow B}),
\end{equation}
where $E_{\operatorname{PPT}}^{(p)}$ denotes the exact parallel simulation
entanglement cost. From Theorem~\ref{thm:parallel-cost-asymp}, we have that%
\begin{equation}
E_{\operatorname{PPT}}^{(p)}(\mathcal{N}_{A\rightarrow B})=E_{\kappa
}(\mathcal{N}_{A\rightarrow B}).
\end{equation}
So it suffices to prove that%
\begin{equation}
E_{\kappa}(\mathcal{N}_{A\rightarrow B})=E_{\kappa}(\omega_{A^{\prime
}B^{\prime}}).
\end{equation}

Letting $\rho_{RA}$ be an arbitrary input state, we have that%
\begin{align}
E_{\kappa}(\mathcal{N}_{A\rightarrow B}(\rho_{RA})) &  =E_{\kappa}%
(\mathcal{P}_{AA^{\prime}B^{\prime}\rightarrow B}(\rho_{RA}\otimes
\omega_{A^{\prime}B^{\prime}}))\\
&  \leq E_{\kappa}(\rho_{RA}\otimes\omega_{A^{\prime}B^{\prime}})\\
&  =E_{\kappa}(\omega_{A^{\prime}B^{\prime}}),
\end{align}
where the inequality follows from the monotonicity of $E_{\kappa}$ under
PPT-preserving channels and the final equality follows because the bipartite
cut is taken as $RAA^{\prime}|B^{\prime}$. Since this holds for an arbitrary
input state $\rho_{RA}$, we conclude that%
\begin{equation}
E_{\kappa}(\omega_{A^{\prime}B^{\prime}})\geq E_{\kappa}(\mathcal{N}%
_{A\rightarrow B}).
\end{equation}

Now we prove the opposite inequality, by using the fact that $\mathcal{N}%
_{A\rightarrow B}$ is resource-seizable. Let $\tau_{A_{M}AB_{M}}$ be the input
PPT state from Definition~\ref{def:res-seize}.\ Consider that%
\begin{align}
E_{\kappa}(\omega_{A^{\prime}B^{\prime}})  & =E_{\kappa}(\mathcal{D}%
_{A_{M}BB_{M}\rightarrow A^{\prime}B^{\prime}}(\mathcal{N}_{A\rightarrow
B}(\tau_{A_{M}AB_{M}})))\\
& \leq E_{\kappa}(\mathcal{N}_{A\rightarrow B}(\tau_{A_{M}AB_{M}}))\\
& =E_{\kappa}(\mathcal{N}_{A\rightarrow B}(\tau_{A_{M}AB_{M}}))-E_{\kappa
}(\tau_{A_{M}AB_{M}})\\
& \leq E_{\kappa}(\mathcal{N}_{A\rightarrow B}).
\end{align}
The first inequality follows because $E_{\kappa}$ does not increase under the
action of the completely PPT-preserving channel $\mathcal{D}_{A_{M}%
BB_{M}\rightarrow A^{\prime}B^{\prime}}$ (Theorem~\ref{prop:ent-monotone}). The second
equality follows because $\tau_{A_{M}AB_{M}}$ is a PPT state, so that
$E_{\kappa}(\tau_{A_{M}AB_{M}})=0$. The final inequality is a consequence of
the amortization inequality in Proposition~\ref{prop:amort-ineq-E-kappa}.
\end{proof}

%
%
%

\subsection{Relationship to other quantities}

A previously known efficiently computable upper bound for quantum capacity is the partial transposition bound \cite{HW01}:
\begin{align}
Q_{\Theta}(\cN)\coloneqq\log_2 \left\|  T_{B\to B} \circ \cN_{A \to B} \right\|_{\Diamond},
\end{align} 
where $T_{B\to B}$ is the transpose map and  $\|\cdot\|_{\Diamond}$ is the completely bounded trace norm or diamond norm. Note that $\|\cdot\|_{\Diamond}$ for finite-dimensional channels is  efficiently computable via semidefinite programming \cite{Watrous2012}.

 \begin{proposition}
 \label{prop:HW-lower-e-kappa}
 	For every quantum channel $\cN_{A\to B}$, we have that
 	\begin{align}
 	Q_{\Theta} (\cN_{A\to B}) \le E_{\kappa}(\cN_{A\to B}).
 	\end{align}
 \end{proposition}
 \begin{proof}
 	Given an arbitrary quantum channel $\cN_{A\to B}$, it holds that
 	\begin{align}
 	E_\kappa(\cN_{A\to B}) &= \sup_{\phi_{RA}}E_\kappa(\cN_{A\to B}(\phi_{RA})) \label{eq:channel K sup def}\\
 	&\ge \sup_{\phi_{RA}}E_N(\cN_{A\to B}(\phi_{RA})) \label{eq:channel K larger}\\
 	&= \sup_{\phi_{RA}}\log_2\|\cN_{A\to B}(\phi_{RA})^{T_B}\|_1\\
 	&=\log_2 \left\| T_{B\to B} \circ \cN_{A \to B}\right\|_{\Diamond}.
 	\end{align}
 	The equality in \eqref{eq:channel K sup def} follows from Proposition~\ref{eq:state-opt-for-E_kappa-ch}.
 	The inequality in \eqref{eq:channel K larger} follows from the property of $E_\kappa$ in Eq.~\eqref{eq:kappa-equal-log-neg}. The last equality follows due to the definition of the completely bounded trace norm.
 \end{proof}

\begin{remark}
	For qubit-input qubit-output channels, we have that
	\begin{equation}
		E_{\kappa}(\cN_{A\to B}) = Q_{\Theta} (\cN_{A\to B}).
	\end{equation}
	This follows because it suffices to optimize $E_{\kappa}(\cN_{A\to B})$ with respect to two-qubit input states $\phi_{RA}$, and then the output state consists of two qubits, so that the result of \cite{Ishizaka2004a} applies. That is, for this case,
	\begin{align}
	E_{\kappa}(\cN_{A\to B}) & = \sup_{\phi_{RA}}E_{\kappa}(\cN_{A\to B}(\phi_{RA})) \\
	& = 
	\sup_{\phi_{RA}}E_N(\cN_{A\to B}(\phi_{RA})) \\
	& = Q_{\Theta} (\cN_{A\to B}).
	\end{align}
\end{remark}

\section{Exact entanglement cost of fundamental channels} 

\label{sec:examples of channels}

Theorem~\ref{thm:res-seize-simplify} provides a formula for the exact PPT-entanglement cost of an arbitrary resource-seizable, PPT-simulable channel, given in terms of the entanglement cost of the underlying resource state
$\omega_{A^{\prime}B^{\prime}}$. We detail some simple examples here for which this simplified formula applies. We also consider amplitude damping channels, for which it is necessary to invoke Theorems~\ref{thm:parallel-cost-asymp} and \ref{thm:seq-ch-sim-kappa-ent} in order to determine their exact entanglement costs.

Let us begin by recalling the notion of a covariant channel $\mathcal{N}%
_{A\rightarrow B}$ \cite{Hol02}. For a group $G$ with unitary channel
representations $\{\mathcal{U}_{A}^{g}\}_{g \in G}$ and $\{\mathcal{V}_{B}^{g}%
\}_{g  \in G}$ acting on the input system $A$ and output system $B$ of the channel
$\mathcal{N}_{A\rightarrow B}$, the channel $\mathcal{N}_{A\rightarrow B}$ is
covariant with respect to the group$~G$ if the following equality holds for all $g \in G$:
\begin{equation}
\mathcal{N}_{A\rightarrow B}\circ\mathcal{U}_{A}^{g}=\mathcal{V}_{B}^{g}%
\circ\mathcal{N}_{A\rightarrow B}.
\label{eq:cov-to-help}
\end{equation}
If the averaging channel is such that $\frac{1}{\left\vert G\right\vert }%
\sum_{g}\mathcal{U}_{A}^{g}(X)=\operatorname{Tr}[X]I/\left\vert A\right\vert $, then we simply say that the channel
$\mathcal{N}_{A\rightarrow B}$ is covariant.

Then from \cite[Section~7]{CDP09}, we
conclude that a covariant channel is PPT-simulable with associated
resource state given by the Choi state of the channel, i.e., $\omega
_{A^{\prime}B^{\prime}}=\mathcal{N}_{A\rightarrow B}(\Phi_{A^{\prime}A})$. As
such, covariant channels are resource-seizable, so that the equality in
Theorem~\ref{thm:res-seize-simplify} applies to all covariant channels.
Thus, the exact entanglement cost of a covariant channel is equal to
the exact entanglement cost of its Choi state.

\subsection{Erasure channel}

The quantum erasure channel is denoted by
\begin{align}
	\cE_p(\rho) = (1-p)\rho + p \ket{e}\!\bra{e},
\end{align}
where $\rho$ is a $d$-dimensional input state,  $p\in\left[  0,1\right]  $ is
the erasure probability, and $|e\rangle\!\langle e|$ is a pure erasure state
orthogonal to every input state, so that the output state has $d+1$ dimensions. This channel is covariant.

The Choi matrix of $\cE_p$ is given by
\begin{align}
J_{\cE_p}= (1-p)\sum_{i,j =0}^{d-1} \ket{ii}\!\bra{jj} + p\sum_{i=0}^{d-1}\ket{i}\!\bra{i}\ox \proj e.
\end{align}
By direct calculation, we find that
\begin{align}
	\EPPT(\cE_p)& =\EPPT(J_{\cE_p}/d)\\
	& =E_N(J_{\cE_p}/d)\\
	& =\log_2 (d[1-p]+p).
\end{align}

\subsection{Depolarizing channel}

Consider the qudit depolarizing channel:
\begin{align}
\cN_{D,p}(\rho) = (1-p) \rho + \frac{p}{d^2-1} \sum_{\substack{0\leq i,j\leq d-1 \\ (i,j)\neq (0,0)}} X^i Z^j \rho (X^i Z^j)^\dagger,
\end{align} 
where $p\in[0,1]$ and $X,Z$ are the generalized Pauli operators. This channel is covariant.

The Choi matrix of $\cN_{D,p}$ is 
\begin{align}
	 J_{\cN_{D,p}} = d\left[(1-p)\Phi_{AB} +  \frac{p}{d^2-1}(\1_{AB}-\Phi_{AB})\right],
\end{align}
 where $\Phi = \frac{1}{d} \sum_{i,j=0}^{d-1} \ket{ii}\!\bra{jj}$. Observe that the state $\frac{J_{\cN_{D,p}}}{d}$ is an isotropic state. Applying the previous result from \eqref{eq:e-ppt-isotropic}, we conclude that
 \begin{align}
 \EPPT(\cN_{D,p})=
 &  \begin{cases}
 \log_2 d(1-p) & \text{ if } 1-p\ge \frac{1}{d}    \\
 0 & \text{ if } 1-p< \frac{1}{d}  
 \end{cases}
 \end{align}
 
\subsection{Dephasing channel}\label{sec: dephasing}

The qubit dephasing channel is given as
\begin{equation}
\mathcal{D}_q(\rho)=(1-q)\rho+qZ\rho Z.
\end{equation}
 Note that this channel is covariant with respect to the
Heisenberg--Weyl group of unitaries.
The Choi matrix of $\mathcal{D}_q$ is as follows:
\begin{align}
J_{\mathcal{D}_q}=2[(1-q) \psi_1+ q\psi_2],
\end{align}
where 
\begin{align}
\ket{\psi_1}=\frac{1}{\sqrt 2}(\ket{00}+\ket{11}),\quad
\ket{\psi_2}=\frac{1}{\sqrt 2}(\ket{00}-\ket{11}).
\end{align}
By direct calculation, we find that
\begin{align}
\EPPT(\mathcal{D}_q)& =\EPPT(J_{\mathcal{D}_q}/2)\\
& =E_N(J_{\mathcal{D}_q}/2)\\
& =\log_2 (1+2|q-1/2|).
\end{align}
We note that this approach also works for a $d$-dimensional dephasing channel. 

\subsection{Amplitude damping channel}

\label{sec:amp-damp-ch}

An amplitude damping channel corresponds to the process of asymmetric relaxation in a quantum system, which is a key noise process in quantum information science.
The qubit amplitude damping channel is given as $\cN_{AD,r}=\sum_{i=0}^1 E_i\cdot E_i^\dag$ 
with 
\begin{align}
E_0=\ketbra{0}{0}+\sqrt{1-r}\ketbra{1}{1}, \quad
E_1=\sqrt{r}\ketbra{0}{1},
\end{align} 
and where $r\in[0,1]$ is the damping parameter.
This channel is covariant with respect to $\{I,Z\}$, but not with respect to a one-design. So Theorem~\ref{thm:res-seize-simplify} does not apply, and we instead need to evaluate the exact entanglement cost of this channel by applying Theorems~\ref{thm:parallel-cost-asymp} and \ref{thm:seq-ch-sim-kappa-ent}.

We plot $\EPPT(\cN_{AD,r})$ in Figure~\ref{fig: AD}  and compare it with the max-Rains information of \cite{WD16,WFD17}. The fact that there is a gap between these two quantities demonstrates that the resource theory of entanglement (exact PPT case) is irreversible, given that the max-Rains information is an upper bound on the exact distillable entanglement of an arbitrary channel \cite{BW17}. 
\begin{figure}
\center
\includegraphics[width=\linewidth]{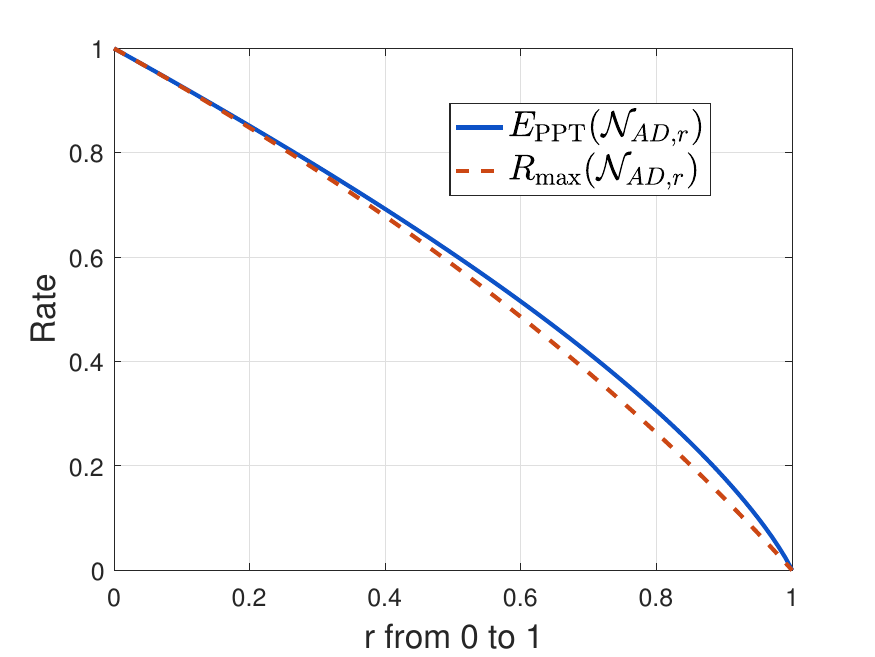}
\caption{This plot demonstrates the difference between  $\EPPT(\cN_{AD,r})$ and $R_{\max}(\cN_{AD,r})$, where $\cN_{AD,r}$ is the amplitude damping channel in Section~\ref{sec:amp-damp-ch}. The solid line depicts  $\EPPT(\cN_{AD,r})$ while the dashed line depicts $R_{\max}(\cN_{AD,r})$. The parameter $r$ ranges from $0$ to $1$, and the units of the rate (vertical axis) are ebits per channels use.}
\label{fig: AD} 
\end{figure}

\section{Exact entanglement cost of quantum Gaussian channels}

\label{sec:gaussian-channels}

In this subsection, we determine formulas for the exact entanglement cost of particular quantum Gaussian channels, which include all single-mode bosonic Gaussian channels with the exception of the pure-loss and pure-amplifier channels. In this sense, the results found here are complementary to those found recently in \cite[Theorem~2]{Wilde2018}.
The presentation and background given in this section largely follows that given recently in \cite{Wilde2018}.

\subsection{Preliminary observations about the exact entanglement cost of single-mode bosonic Gaussian channels}

The starting point for our analysis of single-mode bosonic Gaussian channels
is the Holevo classification from \cite{Holevo2007}, in which canonical forms
for all single-mode bosonic Gaussian channels have been given, classifying
them up to local Gaussian unitaries acting on the input and output of the
channel. It then suffices for us to focus our attention on the canonical
forms, as it is self-evident from definitions that local unitaries do not
alter the exact entanglement cost of a quantum channel. The thermal and amplifier
channels form the class C discussed in \cite{Holevo2007}, and the
additive-noise channels form the class B$_{2}$ discussed in the same work. The
classes that remain are labeled A, B$_{1}$, and D in \cite{Holevo2007}. The
channels in A and D\ are entanglement-breaking \cite{Holevo2008}, and are thus entanglement-binding, and as a
consequence of 
Proposition~\ref{prop:faithful-channels} and Theorems~\ref{thm:parallel-cost-asymp} and \ref{thm:seq-ch-sim-kappa-ent}, they have zero exact entanglement cost. Channels in the
class B$_{1}$ are perhaps not interesting for practical applications, and as
it turns out, they have infinite quantum capacity \cite{Holevo2007}. Thus,
their exact entanglement cost is also infinite, because a channel's quantum capacity
is a lower bound on its distillable entanglement, which is in turn a lower
bound on its partial transposition bound. The partial transposition bound is  finally a lower bound on its $\kappa$-entanglement, as shown in Proposition~\ref{prop:HW-lower-e-kappa}. For the same reason, the
exact entanglement cost of the bosonic identity channel is also infinite.

\subsection{Thermal, amplifier, and additive-noise bosonic Gaussian channels}

In light of the previous discussion, for the remainder of this section, let us focus
our attention on the thermal, amplifier, and additive-noise channels. Each of
these are defined respectively by the following Heisenberg input-output
relations:%
\begin{align}
\hat{b}  &  =\sqrt{\eta}\hat{a}+\sqrt{1-\eta}\hat{e}%
,\label{eq:thermal-channel}\\
\hat{b}  &  =\sqrt{G}\hat{a}+\sqrt{G-1}\hat{e}^{\dag}%
,\label{eq:amplifier-channel}\\
\hat{b}  &  =\hat{a}+\left(  x+ip\right)  /\sqrt{2},
\label{eq:additive-noise-channel}%
\end{align}
where $\hat{a}$, $\hat{b}$, and $\hat{e}$ are the field-mode annihilation
operators for the sender's input, the receiver's output, and the environment's
input of these channels, respectively.

The channel in \eqref{eq:thermal-channel} is a thermalizing channel, in which
the environmental mode is prepared in a thermal state $\theta(N_{B})$\ of mean
photon number $N_{B}\geq0$, defined as%
\begin{equation}
\theta(N_{B})\coloneqq \frac{1}{N_{B}+1}\sum_{n=0}^{\infty}\left(  \frac{N_{B}%
}{N_{B}+1}\right)  ^{n}|n\rangle\!\langle n|,
\end{equation}
where $\left\{  |n\rangle\right\}  _{n=0}^{\infty}$ is the orthonormal,
photonic number-state basis. When $N_{B}=0$, the state $\theta(N_{B})$ reduces to the
vacuum state, in which case the resulting channel in
\eqref{eq:thermal-channel} is called the pure-loss channel---it is said to be
quantum-limited in this case because the environment is injecting the minimum
amount of noise allowed by quantum mechanics. The parameter $\eta\in(0,1)$ is
the transmissivity of the channel, representing the average fraction of
photons making it from the input to the output of the channel. Let
$\mathcal{L}_{\eta,N_{B}}$ denote this channel, and we make the further
abbreviation $\mathcal{L}_{\eta}\equiv\mathcal{L}_{\eta,N_{B}=0}$ when it is
the pure-loss channel. The channel in \eqref{eq:thermal-channel} is
entanglement-breaking when $\left(  1-\eta\right)  N_{B}\geq\eta$
\cite{Holevo2008}, and is thus entanglement-binding in this case, and as a
consequence of 
Proposition~\ref{prop:faithful-channels} and Theorems~\ref{thm:parallel-cost-asymp} and \ref{thm:seq-ch-sim-kappa-ent}, it has zero exact entanglement cost for these values.

The channel in \eqref{eq:amplifier-channel} is an amplifier channel, and the
parameter $G>1$ is its gain. For this channel, the environment is prepared in
the thermal state $\theta(N_{B})$. If $N_{B}=0$, the amplifier channel is
called the pure-amplifier channel---it is said to be quantum-limited for a
similar reason as stated above. Let $\mathcal{A}_{G,N_{B}}$ denote this
channel, and we make the further abbreviation $\mathcal{A}_{G}\equiv
\mathcal{A}_{G,N_{B}=0}$ when it is the quantum-limited amplifier channel. The
channel in \eqref{eq:amplifier-channel} is entanglement-breaking when $\left(
G-1\right)  N_{B}\geq1$ \cite{Holevo2008}, and is thus entanglement-binding, and as a
consequence of 
Proposition~\ref{prop:faithful-channels} and Theorems~\ref{thm:parallel-cost-asymp} and \ref{thm:seq-ch-sim-kappa-ent}, it has zero exact entanglement cost for these values.

Finally, the channel in \eqref{eq:additive-noise-channel} is an additive-noise
channel, representing a quantum generalization of the classical additive white
Gaussian noise channel. In \eqref{eq:additive-noise-channel}, $x$ and $p$ are
zero-mean, independent Gaussian random variables each having variance $\xi\geq0$. Let
$\mathcal{T}_{\xi}$ denote this channel. The channel in
\eqref{eq:additive-noise-channel} is entanglement-breaking when $\xi\geq1$
\cite{Holevo2008}, and is thus entanglement-binding, and as a
consequence of 
Proposition~\ref{prop:faithful-channels} and Theorems~\ref{thm:parallel-cost-asymp} and \ref{thm:seq-ch-sim-kappa-ent}, it has zero exact entanglement cost for these values.

Kraus representations for the channels in
\eqref{eq:thermal-channel}--\eqref{eq:additive-noise-channel}\ are available
in \cite{ISS11}, which can be helpful for further understanding their action
on input quantum states.

Due to the entanglement-breaking regions discussed above, we are left with a
limited range of single-mode bosonic Gaussian channels to consider, which is
delineated by the white strip in Figure~1 of \cite{GPCH13}.

\subsection{Exact entanglement cost of thermal, amplifier, and additive-noise bosonic Gaussian channels}

We can now state our main result for this section, which applies to all thermal, amplifier, and additive-noise channels that are neither entanglement-breaking nor quantum-limited:

\begin{theorem}
\label{thm:bosonic-formulas} For a thermal  channel $\mathcal{L}_{\eta,N_B}$
with transmissivity $\eta\in(0,1)$ and thermal photon number $N_B\in(0,\eta/[1-\eta])$,  an amplifier channel $\mathcal{A}%
_{G,N_B}$ with gain $G>1$ and thermal photon number $N_B\in(0,1/[G-1])$, and an additive-noise channel $\mathcal{T}_{\xi}$ with noise variance $\xi\in(0,1]$, the following formulas characterize the exact entanglement
costs of these channels:%
\begin{align}
\EPPT(\mathcal{L}_{\eta,N_B})  &  = 
E_{\PPT}^{(p)}(\mathcal{L}_{\eta,N_B})
 \notag \\
 & =
\log_2\!\left(\frac{1+\eta}{(1-\eta)(2N_B+1)}\right) 
,\label{eq:loss-formula}\\
\EPPT(\mathcal{A}_{G,N_B})  &  =
E_{\PPT}^{(p)}(\mathcal{A}_{G,N_B}) \notag \\
& =
\log_2\!\left(\frac{G+1}{(G-1)(2N_B+1)}\right) 
, \label{eq:amp-formula}\\
\EPPT(\mathcal{T}_{\xi})  &  =
E_{\PPT}^{(p)}(\mathcal{T}_{\xi}) =
\log_2(1/\xi).
\end{align}
\end{theorem}

\begin{proof}
To arrive at the following inequalities:
\begin{align}
\EPPT(\mathcal{L}_{\eta,N_B})  &  
 \leq
\log_2\!\left(\frac{1+\eta}{(1-\eta)(2N_B+1)}\right) 
,\\
\EPPT(\mathcal{A}_{G,N_B})  &   \leq
\log_2\!\left(\frac{G+1}{(G-1)(2N_B+1)}\right) 
, \\
\EPPT(\mathcal{T}_{\xi})  &  \leq
\log_2(1/\xi),
\end{align} 
we apply Proposition~\ref{prop:PPT-sim}, along with some recent developments, to the single-mode thermal, amplifier, and
additive-noise channels that are neither entanglement-breaking nor quantum-limited. Some recent papers \cite{LMGA17,KW17,TDR18}\ have
shown how to simulate each of these channels by using a bosonic Gaussian
resource state along with variations of the continuous-variable quantum
teleportation protocol \cite{prl1998braunstein}. Of these works, the one most
relevant for us is the original one \cite{LMGA17}, because these authors proved
that the logarithmic negativity of the underlying resource state is equal
to the logarithmic negativity that results from transmitting through the channel one share of a
two-mode squeezed vacuum state with arbitrarily large squeezing strength. That
is, let $\mathcal{N}_{A\rightarrow B}$ denote a single-mode thermal,
amplifier, or additive-noise channel. Then one of the main results of
\cite{LMGA17} is that, associated to this channel, there is a bosonic Gaussian
resource state $\omega_{A^{\prime}B^{\prime}}$ and a Gaussian LOCC channel
$\mathcal{G}_{AA^{\prime}B^{\prime}\rightarrow B}$ such that
\begin{align}
E_{N}(\omega_{A^{\prime}B^{\prime}})  &  =\sup_{N_{S}\geq0}
E_{N}
(\sigma^{N_{S}}_{RB})
\label{eq:LMGA-1}\\
&  =\lim_{N_{S}\rightarrow\infty}E_{N}(\sigma^{N_{S}}_{RB}),
\label{eq:LMGA-2}%
\end{align}
where%
\begin{align}
\sigma^{N_{S}}_{RB}  &  \coloneqq \mathcal{N}_{A\rightarrow B}(\phi_{RA}^{N_{S}}),\\
\phi_{RA}^{N_{S}}  &  \coloneqq |\phi^{N_{S}}\rangle\!\langle\phi^{N_{S}}|_{RA},\\
|\phi^{N_{S}}\rangle_{RA}  &  \coloneqq \frac{1}{\sqrt{N_{S}+1}}\sum_{n=0}%
^{\infty}\sqrt{\left(  \frac{N_{S}}{N_{S}+1}\right)  ^{n}}|n\rangle
_{R}|n\rangle_{A}, \label{eq:TMSV}%
\end{align}
and for every input state $\rho_{A}$,%
\begin{equation}
\mathcal{N}_{A\rightarrow B}(\rho_{A})=\mathcal{G}_{AA^{\prime}B^{\prime
}\rightarrow B}(\rho_{A}\otimes\omega_{A^{\prime}B^{\prime}}).
\end{equation}
In the above, $\phi_{RA}^{N_{S}}$ is the two-mode squeezed vacuum state
\cite{S17}. Note that the equality in \eqref{eq:LMGA-2}\ holds because one
can always produce $\phi_{RA}^{N_{S}}$ from $\phi_{RA}^{N_{S}^{\prime}}$ such
that $N_{S}^{\prime}\geq N_{S}$, by using Gaussian LOCC\ and the local
displacements involved in the Gaussian LOCC commute with the channel
$\mathcal{N}_{A\rightarrow B}$ \cite{GECP03} (whether it be thermal,
amplifier, or additive-noise). Furthermore, the logarithmic negativity does
not increase under the action of an LOCC\ channel.

Thus, applying the above observations and
Proposition~\ref{prop:PPT-sim}, it follows that there exist bosonic
Gaussian resource states $\omega_{A^{\prime}B^{\prime}}^{\eta,N_{B}}$,
$\omega_{A^{\prime}B^{\prime}}^{G,N_{B}}$, and $\omega_{A^{\prime}B^{\prime}%
}^{\xi}$ associated to the respective thermal, amplifier, and additive-noise
channels in \eqref{eq:thermal-channel}--\eqref{eq:additive-noise-channel},
such that the following inequalities hold%
\begin{align}
\EPPT(\mathcal{L}_{\eta,N_{B}})  &  \leq E_{\kappa}(\omega^{\eta,N_{B}}_{A^{\prime}B^{\prime}}
)  = E_{N}(\omega^{\eta,N_{B}}_{A^{\prime}B^{\prime}}) \notag \\
& 
 = \log_2\!\left(\frac{1+\eta}{(1-\eta)(2N_B+1)}\right),
\label{eq:e-cost-upper-thermal}\\
\EPPT(\mathcal{A}_{G,N_{B}})  &  \leq E_{\kappa}(\omega^{G,N_{B}}_{A^{\prime}B^{\prime}}
) = E_{N}(\omega^{G,N_{B}}_{A^{\prime}B^{\prime}}
) \notag \\
& = \log_2\!\left(\frac{G+1}{(G-1)(2N_B+1)}\right) ,
\\
\EPPT(\mathcal{T}_{\xi})  &  \leq E_{\kappa}(\omega^{\xi}_{A^{\prime}B^{\prime}}) = E_{N}(\omega^{\xi}_{A^{\prime}B^{\prime}}) \notag 
\\
& = \log_2(1/\xi),
\label{eq:e-cost-upper-additive-noise}%
\end{align}
where the first equalities in each line follow because $E_{\kappa} = E_N$ for bosonic Gaussian states (see \eqref{eq:kappa-equal-log-neg} and \cite{Audenaert2003}), and the explicit formulas on the right-hand side are found in \cite{HW01,LMGA17}.

On the other hand, Theorems~\ref{thm:parallel-cost-asymp}  and \ref{thm:seq-ch-sim-kappa-ent} imply that%
\begin{align}
\EPPT(\mathcal{L}_{\eta,N_B})  & = 
E_{\PPT}^{(p)}(\mathcal{L}_{\eta,N_B}) \\& \geq\lim_{N_{S}\rightarrow\infty}E_{N}%
(\sigma^{\eta,N_B}(N_{S})_{RB}) \\
& = \log_2\!\left(\frac{1+\eta}{(1-\eta)(2N_B+1)}\right), 
\\
\EPPT(\mathcal{A}_{G,N_B})  &  =
E_{\PPT}^{(p)}(\mathcal{A}_{G,N_B}) \\
& \geq 
\lim_{N_{S}\rightarrow\infty}E_{N}(\sigma
^{G,N_B}(N_{S})_{RB})\\
&  = \log_2\!\left(\frac{G+1}{(G-1)(2N_B+1)}\right),\\
\EPPT(\mathcal{T}_{\xi})  &  =
E_{\PPT}^{(p)}(\mathcal{T}_{\xi}) \\
& \geq 
\lim_{N_{S}\rightarrow\infty}E_{N}(\sigma
^{\xi}(N_{S})_{RB}) \\
& = \log_2(1/\xi).
\end{align}
Combining the inequalities above, we conclude the statement of the theorem.
\end{proof}

\bigskip
The significance of Theorem~\ref{thm:bosonic-formulas} above is that it establishes a clear operational meaning of the Holevo--Werner quantity \cite{HW01} (partial transposition bound) for the basic bosonic channels that are not quantum limited. This quantity has been used for a variety of purposes in prior work, as an upper bound on unassisted quantum capacity \cite{HW01}, as an upper bound on LOCC-assisted quantum capacity \cite{MRW16}, as a tool in arriving at a no-go theorem for Gaussian quantum error correction \cite{NFC09}, and as a tool in the teleportation simulation of bosonic Gaussian channels \cite{LMGA17}. Finally, Theorem~\ref{thm:bosonic-formulas} solves the long-standing open problem of giving the Holevo--Werner quantity a direct operational meaning for the basic bosonic channels, in terms of exact entanglement cost of parallel and sequential channel simulation.

In light of the results stated in Theorem~\ref{thm:bosonic-formulas}, it is quite natural to conjecture that the following formulas hold for the pure-loss and pure-amplifier channels with $\eta\in(0,1)$ and $G>1$, respectively:
\begin{align}
\EPPT(\mathcal{L}_{\eta})  &  
 =
E_{\PPT}^{(p)}(\mathcal{L}_{\eta}) \overset{?}{=} \log_2\!\left(\frac{1+\eta}{1-\eta}\right) 
,\label{eq:Gauss-conj-1}\\
\EPPT(\mathcal{A}_{G})  &  =
E_{\PPT}^{(p)}(\mathcal{A}_{G}) \overset{?}{=} 
\log_2\!\left(\frac{G+1}{G-1}\right) .
\label{eq:Gauss-conj-2}
\end{align}
Theorems~\ref{thm:parallel-cost-asymp}  and \ref{thm:seq-ch-sim-kappa-ent} imply that the following inequalities hold
\begin{align}
\EPPT(\mathcal{L}_{\eta})  &  
 =
E_{\PPT}^{(p)}(\mathcal{L}_{\eta}) \geq \log_2\!\left(\frac{1+\eta}{1-\eta}\right) 
,\\
\EPPT(\mathcal{A}_{G})  &  =
E_{\PPT}^{(p)}(\mathcal{A}_{G}) \geq 
\log_2\!\left(\frac{G+1}{G-1}\right) .
\end{align}
However, what excludes us from making a rigorous statement about the opposite inequalities is the lack of a legitimate quantum state that can be used to simulate these channels exactly, as was the case for the channels considered in Theorem~\ref{thm:bosonic-formulas}. For example, it is not clear that we could simply ``plug in'' the ``EPR state'' (i.e., the limiting object $\lim_{N_S \to \infty} \phi^{N_S}_{RA})$ and use the teleportation simulation argument as before. There are several issues: the limiting object is not actually a state and any finite squeezing leads to a slight error or inexact simulation. In spite of these obstacles, we think that it is highly plausible that the equalities in \eqref{eq:Gauss-conj-1}--\eqref{eq:Gauss-conj-2} hold.
More generally, based on the results of \cite{NFC09}, we suspect that the following equality holds for an arbitrary Gaussian channel $\cN$ described by a scaling matrix $X$ and a noise matrix $Y$ \cite{S17}:
\begin{align}
\EPPT(\cN)\overset{?}{=}
Q_{\Theta}(\cN)\overset{?}{=} \frac{1}{2}\log_2 \min\left\{\frac{(1+\det X)^2}{\det Y},1\right\}.
\label{eq:conjecture-Gaussian}
\end{align}

\section{Concluding remarks}

In the zoo of entanglment measures \cite{Horodecki2009a,Christandl2006,Plenio2007}, the $\kappa$-entanglement of a bipartite state is the first entanglement measure that is efficiently computable while possessing a direct operational meaning for general bipartite states. This unique feature of $\EK$ may help us better understand the structure and power of quantum entanglement. As a generalization of this notion, the $\kappa$-entanglement of a quantum channel is also efficiently computable while possessing a direct operational meaning as the entanglement cost for exact parallel and sequential simulation of a quantum channel.

Going forward from here, the most pressing open question is to determine whether the formula in \eqref{eq:conjecture-Gaussian} holds, for the exact entanglement cost of quantum Gaussian channels. One could potentially require new methods beyond the scope of this paper in order to establish~\eqref{eq:conjecture-Gaussian}.

\begin{acknowledgements}
We are grateful to Renato Renner and Andreas Winter for insightful discussions. XW acknowledges support from the Department of Defense. Part of this work was done when XW was at the Joint Center for Quantum Information and Computer Science (QuICS) at the  University of Maryland. MMW acknolwedges support from the National Science Foundation under Award no.~1350397.
\end{acknowledgements}

 \bibliographystyle{alpha}
 \bibliography{Ref}

\appendix

\section{Equality of $E_{\kappa}$ and $E^{\text{dual}}_{\kappa}$ for states acting on separable Hilbert spaces}

\label{app:kappa-to-dual-infty}

In this appendix, we prove that%
\begin{equation}
E_{\kappa}(\rho_{AB})=E_{\kappa}^{\text{dual}}(\rho_{AB}%
),\label{eq:kappa-dual-infty}%
\end{equation}
for a state $\rho_{AB}$ acting on a separable Hilbert space. To begin with,
let us recall that the following inequality always holds from weak duality%
\begin{equation}
E_{\kappa}(\rho_{AB})\geq E_{\kappa}^{\text{dual}}(\rho_{AB}).
\end{equation}
So our goal is to prove the opposite inequality. We suppose throughout that
$E_{\kappa}^{\text{dual}}(\rho_{AB})<\infty$. Otherwise, the desired equality
in \eqref{eq:kappa-dual-infty}\ is trivially true. We also suppose that
$\rho_{AB}$ has full support. Otherwise, it is finite-dimensional and the
desired equality in \eqref{eq:kappa-dual-infty}\ is trivially true.

To this end, consider sequences $\{\Pi_{A}^{k}\}_{k}$ and $\{\Pi_{B}^{k}%
\}_{k}$ of projectors weakly converging to the identities $\1_{A}$ and $\1_{B}$
and such that $\Pi_{A}^{k}\leq\Pi_{A}^{k^{\prime}}$ and $\Pi_{B}^{k}\leq
\Pi_{B}^{k^{\prime}}$ for $k^{\prime}\geq k$. Furthermore, we suppose that
$[\Pi_{B}^{k}]^{T_{B}}=\Pi_{B}^{k}$ for all $k$. Then define%
\begin{equation}
\rho_{AB}^{k}\coloneqq\left(  \Pi_{A}^{k}\otimes\Pi_{B}^{k}\right)  \rho_{AB}\left(
\Pi_{A}^{k}\otimes\Pi_{B}^{k}\right)  .
\end{equation}
It follows that \cite{D67}%
\begin{equation}
\lim_{k\rightarrow\infty}\left\Vert \rho_{AB}-\rho_{AB}^{k}\right\Vert _{1}=0.
\end{equation}

We now prove that%
\begin{equation}
E_{\kappa}^{\text{dual}}(\rho_{AB})\geq E_{\kappa}^{\text{dual}}(\rho_{AB}%
^{k})\label{eq:dual-projection}%
\end{equation}
for all $k$. Let $A^{k}$ and $B^{k}$ denote the subspaces onto which $\Pi
_{A}^{k}$ and $\Pi_{B}^{k}$ project. Let $V_{A^{k}B^{k}}^{k}$ and
$W_{A^{k}B^{k}}^{k}$ be arbitrary operators satisfying $V_{AB}^{k}+W_{AB}%
^{k}\leq \1_{A^{k}B^{k}}=\left(  \Pi_{A}^{k}\otimes\Pi_{B}^{k}\right)  $,
$[V_{A^{k}B^{k}}^{k}]^{T_{B}},[W_{A^{k}B^{k}}^{k}]^{T_{B}}\geq0$. Set%
\begin{align}
\overline{V}_{AB}^{k}  & \coloneqq\left(  \Pi_{A}^{k}\otimes\Pi_{B}^{k}\right)
V_{A^{k}B^{k}}^{k}\left(  \Pi_{A}^{k}\otimes\Pi_{B}^{k}\right)  ,\\
\overline{W}_{AB}^{k}  & \coloneqq\left(  \Pi_{A}^{k}\otimes\Pi_{B}^{k}\right)
W_{A^{k}B^{k}}^{k}\left(  \Pi_{A}^{k}\otimes\Pi_{B}^{k}\right)  ,
\end{align}
and note that%
\begin{align}
\overline{V}_{AB}^{k}+\overline{W}_{AB}^{k}  & \leq \1_{AB},\\
\lbrack\overline{V}_{AB}^{k}]^{T_{B}},[\overline{W}_{AB}^{k}]^{T_{B}}  &
\geq0.
\end{align}
Then%
\begin{align}
& \operatorname{Tr}\rho_{AB}^{k}(V_{A^{k}B^{k}}^{k}-W_{A^{k}B^{k}}^{k})  \notag \\
&
=\operatorname{Tr}\left(  \Pi_{A}^{k}\otimes\Pi_{B}^{k}\right)  \rho
_{AB}\left(  \Pi_{A}^{k}\otimes\Pi_{B}^{k}\right)  (V_{A^{k}B^{k}}%
^{k}-W_{A^{k}B^{k}}^{k})\\
& =\operatorname{Tr}\rho_{AB}\left(  \Pi_{A}^{k}\otimes\Pi_{B}^{k}\right)
(V_{A^{k}B^{k}}^{k}-W_{A^{k}B^{k}}^{k})\left(  \Pi_{A}^{k}\otimes\Pi_{B}%
^{k}\right)  \\
& =\operatorname{Tr}\rho_{AB}(\overline{V}_{AB}^{k}-\overline{W}_{AB}^{k})\\
& \leq E_{\kappa}^{\text{dual}}(\rho_{AB}).
\end{align}
Since the inequality holds for arbitrary $V_{A^{k}B^{k}}^{k}$ and
$W_{A^{k}B^{k}}^{k}$ satisfying the conditions above, we conclude the
inequality in \eqref{eq:dual-projection}.

Thus, we conclude that%
\begin{equation}
E_{\kappa}^{\text{dual}}(\rho_{AB})\geq\limsup_{k\rightarrow\infty}E_{\kappa
}^{\text{dual}}(\rho_{AB}^{k}).\label{eq:lim-sup-lower}%
\end{equation}

Now let us suppose that $E_{\kappa}^{\text{dual}}(\rho_{AB})<\infty$. Then for
all $V_{AB}$ and $W_{AB}$ satisfying $V_{AB}+W_{AB}\leq \1_{AB}$,
$[V_{AB}]^{T_{B}},[W_{AB}]^{T_{B}}\geq0$, as well as $\operatorname{Tr}%
\rho_{AB}(V_{AB}-W_{AB})\geq0$, we have that%
\begin{equation}
\operatorname{Tr}\rho_{AB}(V_{AB}-W_{AB})<\infty.
\end{equation}
Since $\rho_{AB}$ has full support, this means that%
\begin{equation}
\left\Vert V_{AB}-W_{AB}\right\Vert _{\infty}<\infty.
\end{equation}
Considering that from H\"older's inequality%
\begin{multline}
\left\vert \operatorname{Tr}(\rho_{AB}-\rho_{AB}^{k})(V_{AB}-W_{AB}%
)\right\vert \leq\\
\left\Vert \rho_{AB}-\rho_{AB}^{k}\right\Vert _{1}\left\Vert
V_{AB}-W_{AB}\right\Vert _{\infty},
\end{multline}
and setting%
\begin{align}
V_{AB}^{k}  & \coloneqq\left(  \Pi_{A}^{k}\otimes\Pi_{B}^{k}\right)  V_{AB}\left(
\Pi_{A}^{k}\otimes\Pi_{B}^{k}\right)  ,\\
W_{AB}^{k}  & \coloneqq\left(  \Pi_{A}^{k}\otimes\Pi_{B}^{k}\right)  W_{AB}\left(
\Pi_{A}^{k}\otimes\Pi_{B}^{k}\right)  ,
\end{align}
we conclude that
\begin{align}
& \operatorname{Tr}\rho_{AB}(V_{AB}-W_{AB}) \notag  \\
& \leq\liminf_{k\rightarrow\infty
}\operatorname{Tr}\rho_{AB}^{k}(V_{AB}-W_{AB})\\
& =\liminf_{k\rightarrow\infty}\operatorname{Tr}\rho_{AB}^{k}(V_{AB}%
^{k}-W_{AB}^{k})\\
& \leq\liminf_{k\rightarrow\infty}\sup_{V^{k},W^{k}}\operatorname{Tr}\rho
_{AB}^{k}(V_{AB}^{k}-W_{AB}^{k})\\
& =\liminf_{k\rightarrow\infty}E_{\kappa}^{\text{dual}}(\rho_{AB}^{k}).
\end{align}
Since the inequality holds for arbitrary $V_{AB}$ and $W_{AB}$ satisfying the
above conditions, we conclude that%
\begin{equation}
E_{\kappa}^{\text{dual}}(\rho_{AB})\leq\liminf_{k\rightarrow\infty}E_{\kappa
}^{\text{dual}}(\rho_{AB}^{k}).\label{eq:lim-inf-upper}%
\end{equation}

Putting together \eqref{eq:lim-sup-lower} and \eqref{eq:lim-inf-upper}, we
conclude that%
\begin{equation}
E_{\kappa}^{\text{dual}}(\rho_{AB})=\lim_{k\rightarrow\infty}E_{\kappa
}^{\text{dual}}(\rho_{AB}^{k}).\label{eq:dual-limits}%
\end{equation}

From strong duality for the finite-dimensional case, we have for all $k$ that%
\begin{equation}
E_{\kappa}^{\text{dual}}(\rho_{AB}^{k})=E_{\kappa}(\rho_{AB}^{k}),
\end{equation}
and thus that%
\begin{equation}
\lim_{k\rightarrow\infty}E_{\kappa}^{\text{dual}}(\rho_{AB}^{k})=\lim
_{k\rightarrow\infty}E_{\kappa}(\rho_{AB}^{k}).\label{eq:finite-str-dual}%
\end{equation}

It thus remains to prove that%
\begin{equation}
\lim_{k\rightarrow\infty}E_{\kappa}(\rho_{AB}^{k})=E_{\kappa}(\rho_{AB}).
\end{equation}
We first prove that%
\begin{equation}
E_{\kappa}(\rho_{AB})\geq\limsup_{k\rightarrow\infty}E_{\kappa}(\rho_{AB}%
^{k}).\label{eq:lim-sup-lower_e_kappa}%
\end{equation}
Let $S_{AB}$ be an arbitrary operator satisfying
\begin{equation}
S_{AB}\geq0,\qquad-S_{AB}^{T_{B}}\leq\rho_{AB}^{T_{B}}\leq S_{AB}^{T_{B}}.
\label{eq:conditions-for-S-limits}
\end{equation}
Then, defining $S_{AB}^{k}=\left(  \Pi_{A}^{k}\otimes\Pi_{B}^{k}\right)
S_{AB}\left(  \Pi_{A}^{k}\otimes\Pi_{B}^{k}\right)  $, we have that%
\begin{equation}
S_{AB}^{k}\geq0,\qquad-[S_{AB}^{k}]^{T_{B}}\leq\lbrack\rho_{AB}^{k}]^{T_{B}%
}\leq\lbrack S_{AB}^{k}]^{T_{B}}.
\end{equation}
Then%
\begin{equation}
\log_2\operatorname{Tr}S_{AB}\geq\log_2\operatorname{Tr}S_{AB}^{k}\geq E_{\kappa
}(\rho_{AB}^{k}).
\end{equation}
Since the inequality holds for all $S_{AB}$ satisfying \eqref{eq:conditions-for-S-limits}, we conclude that%
\begin{equation}
E_{\kappa}(\rho_{AB})\geq E_{\kappa}(\rho_{AB}^{k})
\end{equation}
for all $k$, and thus \eqref{eq:lim-sup-lower_e_kappa}\ holds.

The rest of the proof follows \cite{FAR11} closely.
Since the
condition $\Pi_{A}^{k}\leq\Pi_{A}^{k^{\prime}}$ and $\Pi_{B}^{k}\leq\Pi
_{B}^{k^{\prime}}$ for $k^{\prime}\geq k$ holds, in fact the same sequence of
steps as above allows for concluding that%
\begin{equation}
E_{\kappa}(\rho_{AB}^{k^{\prime}})\geq E_{\kappa}(\rho_{AB}^{k}),
\end{equation}
meaning that the sequence is monotone non-decreasing with $k$. Thus, we can
define%
\begin{equation}
\mu\coloneqq\lim_{k\rightarrow\infty}E_{\kappa}(\rho_{AB}^{k})\in\mathbb{R}^{+},
\end{equation}
and note from the above that%
\begin{equation}
\mu\leq E_{\kappa}(\rho_{AB}).
\end{equation}
For each $k$, let $S_{AB}^{k}$ denote an optimal operator such that
$E_{\kappa}(\rho_{AB}^{k})=\log_2\operatorname{Tr}S_{AB}^{k}$. From the fact
that $S_{AB}^{k}\geq0$, and $\operatorname{Tr}S_{AB}^{k}\leq2^{\mu}$, we
conclude that $\{S_{AB}^{k}\}_{k}$ is a bounded sequence in the trace class
operators. Since the trace class operators form the dual space of the compact
operators $\mathcal{K}(\mathcal{H}_{AB})$ \cite{RS78}, we can apply the Banach--Alaoglu
theorem \cite{RS78} to find a subsequence $\{S_{AB}^{k}\}_{k\in\Gamma}$ with a
weak$^{\ast}$ limit $\widetilde{S}_{AB}$ in the trace class operators such that
$\widetilde{S}_{AB}\geq0$ and $\operatorname{Tr}[\widetilde{S}_{AB}]\leq2^{\mu}$.
Furthermore, the sequences $[\rho_{AB}^{k}]^{T_{B}}+[S_{AB}^{k}]^{T_{B}}$ and
$[S_{AB}^{k}]^{T_{B}}-[\rho_{AB}^{k}]^{T_{B}}$ converge in the weak operator
topology to $\rho_{AB}^{T_{B}}+\widetilde{S}_{AB}^{ T_{B}}$ and $\widetilde{S}_{AB}^{ T_{B}%
}-\rho_{AB}^{T_{B}}$, respectively, and we can then conclude that $\rho
_{AB}^{T_{B}}+\widetilde{S}_{AB}^{ T_{B}},\widetilde{S}_{AB}^{ T_{B}}-\rho_{AB}^{T_{B}}\geq0$.
But this means that%
\begin{equation}
E_{\kappa}(\rho_{AB})\leq\log_2\operatorname{Tr}\widetilde{S}_{AB}\leq\mu,
\end{equation}
which implies that%
\begin{equation}
E_{\kappa}(\rho_{AB})\leq\liminf_{k\rightarrow\infty}E_{\kappa}(\rho_{AB}%
^{k}).\label{eq:lim-inf-upper_e_kappa}%
\end{equation}
Putting together
\eqref{eq:lim-sup-lower_e_kappa} and \eqref{eq:lim-inf-upper_e_kappa}, we
conclude that%
\begin{equation}
E_{\kappa}(\rho_{AB})=\lim_{k\rightarrow\infty}E_{\kappa}(\rho_{AB}%
^{k}).\label{eq:primal-limits}%
\end{equation}
Finally, putting together \eqref{eq:dual-limits}, \eqref{eq:finite-str-dual},
and \eqref{eq:primal-limits},\ we conclude \eqref{eq:kappa-dual-infty}.
\end{document}